\documentclass{lmcs}
\pdfoutput=1

\usepackage{lastpage}
\renewcommand{\dOi}{14(3:12)2018}
\lmcsheading{}{1--\pageref{LastPage}}{}{}%
{May~19,~2017}{Aug.~31,~2018}{}


\usepackage{bussproofs}

\usepackage[all]{xy}
\usepackage{centernot}

\theoremstyle{plain} 
\newtheorem{notation}[thm]{Notation}
\newtheorem{exm}[thm]{Example}
\newtheorem{exms}[thm]{Examples}
\newtheorem{terminology}[thm]{Terminology}

\def\tbox#1#2{\fbox{\parbox{#1\testo}{\raggedright #2}}}
\def\ttbox#1#2{\parbox{#1\testo}{\raggedright #2}}

\let\al\textbf
\newdimen\testo
\mathcode`\<="4268     
\mathcode`\>="5269

\def\UA#1#2{\ensuremath{\xymatrix{#1\ar[r]^{{\bf A}_{#1#2}}&#2}}}
\def\PA#1#2{\ensuremath{\xymatrix{#1\ar@{}[rr]^{{\bf I}_{#1#2}}&\bullet\ar[l]\ar[r]&#2}}}
\def\UN#1#2{\ensuremath{\xymatrix{#1\ar[r]\ar@{}[rr]^{{\bf E}_{#1#2}}&\bullet&#2\ar[l]}}}
\def\PN#1#2{\ensuremath{\xymatrix{#1\ar@{}[rrr]^{{\bf O}_{#1#2}}&\bullet\ar[l]\ar[r]&\bullet&#2\ar[l]}}}

\def\revUA#1#2{\ensuremath{\xymatrix{#2
&#1\ar[l]_{{\bf A}_{#1#2}}}}}
\def\revPA#1#2{\ensuremath{\xymatrix{#2\ar@{}[rr]^{{\bf I}_{#2#1}}&\bullet\ar[l]\ar[r]&#1}}}
\def\revUN#1#2{\ensuremath{\xymatrix{#2\ar[r]\ar@{}[rr]^{{\bf E}_{#2#1}}&\bullet&#1\ar[l]}}}
\def\revPN#1#2{\ensuremath{\xymatrix{#2\ar[r]\ar@{}[rrr]^{{\bf O}_{#1#2}}&\bullet&\bullet\ar[r]\ar[l]&#1}}}

\begin{document}

\title[Ologisms]{Ologisms}

\author{Ruggero Pagnan}
\address{DIMA, Via Dodecaneso 35, 16146, Genova, Italy}	
\email{pagnan@dima.unige.it} 

\begin{abstract}
\noindent We introduce ologisms. They generate from ologs by extending their logical expressivity, from  the possibility of considering constraints of equational nature only to the possibility of considering  constraints of syllogistic nature, in addition. This is obtained by taking advantage of the peculiar features of an original diagrammatic logical calculus for the syllogistic, that make it well-behaved with respect to the design of ologs. 
\end{abstract}

\maketitle


\section{Introduction}
The word ``ologism'' is a crasis between the word ``olog'' and the word ``syllogism''. In turn, the word ``olog'' itself is a crasis which derives from the expression ``ontology log''. Ologs are category-theoretic models for knowledge representation. As such, they should be properly qualified as ``structured'' which in particular means that their construction, usually referred to as ``authoring'', is in accordance with rigorous mathematical criteria and also that it is possible to consider knowledge-preserving transformations between them. Ologs record the ontology or ``essential nature'' of a real situation, in terms of a subjective world-view. Disorienting as it may be, ologs are designed to communicate such world-views which although subjective, being planned in accordance with rigorous criteria turn out to be objectively communicable and their structural soundness objectively determinable. Ologs are strictly related to the conceptual planning of databases. They can be seen to correspond to relational database schemes but are much more readable and friendly. We will not describe the details of this connection but the reader should keep in mind that the matter treated in this paper could be thought of as potentially connected to that activity of planning as well as to logic for knowledge representation; in particular, concerning the latter, see section~\ref{forlogrepr}, where a comparison between ologisms and knowledge representation systems based on Description Logics, DL-KRS, is carried out. The main aim of this paper is to show how to increase the logical expressivity of ologs by combining them with the diagrammatic logical calculus for the syllogistic that we introduced in~\cite{DBLP:journals/jolli/Pagnan12, DBLP:journals/jolli/Pagnan13}, towards the authoring of ologisms, by taking advantage of the peculiar features of that calculus; mainly the fact that it is diagrammatic, formally compositional, sound and complete with respect to syllogistic reasoning. These features make the calculus at issue well-behaved with respect to the authoring of ologs and provide further structural soundness criterions for them; in particular, in connection with the managing of logical contradiction, see section~\ref{contraologisms}. As a final result the logical expressivity of ologs turns out to be actually increased since the possibility of imposing equational constraints, typical of ologs, is placed side by side with the possibility of imposing syllogistic constraints, in addition, which is the peculiarity of ologisms. The imposable syllogistic constraints correspond to the form of the four fundamental categorical propositions, that is ${\bf A}_{AB}$ for ``Every $A$ is $B$'', ${\bf E}_{AB}$ for ``Every $A$ is not $B$'', ${\bf I}_{AB}$ for ``Some $A$ is $B$'' and ${\bf O}_{AB}$ for ``Some $A$ is not $B$'', with $A, B$ types, one of the fundamental components of an olog. For the sake of argument only, one could imagine to authoring an ologism by means of a computer program which assists the authors in the verification of equational and syllogistic constraints, in particular for what regards the avoidance of logical contradictions. We hasten to highlight that for ologisms this is far from being actually so, differently from what it actually is for DL-KRS, see~\cite{MR1991592}. From an overhead standpoint the paper slides from ologs to ologisms through syllogistic by alternating informal approaches and formal approaches. To let the paper to be selfcontained as much as possible, the main mathematical prerequisites are provided as needed, in rigorous terms. More to the point, ologs are introduced informally in section~\ref{ologs} mainly to let the reader to become acquainted with their fundamental components, types, aspects and facts, and with them as graphical objects. Roughly, an olog is a graph with types as nodes and aspects as arcs, together with equational identities between parallel pairs of aspects declared by facts. In section~\ref{ologlinear} ologs are reprised and treated formally as mathematical objects which nature is graphical, algebraic and logical; this is done by looking at them as linear sketches. Linear sketches are graphical formal languages that allow the finite specification of the equational theory generated by a signature with unary operations only, the theory which is the free category over the underlying graph of an olog, quotiented by the smallest congruence relation generated by the imposed equational identities which are the facts declared in it. 
In~\cite{DBLP:journals/corr/abs-1102-1889} it is pointed out that ologs are presentations of categories by means of generators, that is types and aspects, and relations, that is facts, see section~\ref{ologlinear}. Moreover, in {\em loc. cit.}, it is also pointed out that this form of presentation for categories allows to separate the strictly graphical part of an olog, namely its types and aspects, from its equational part, provided by facts. By describing ologs by means of linear sketches we wanted to rigorously take account of this, since linear sketches are the mathematical structures that allow to neatly keep track of this sort of dichotomy. We wanted to apply a genuine principle of specification to make a useful distinction between two separate pieces of structure, although strictly interacting.
Since this paper has to do with the enlargement of the logical expressivity of ologs, we felt the need to reconsider these objects under a proof-theoretical perspective. Thus, section~\ref{ologlinear} further proceeds by treating ologs in formal terms, as labelled deductive systems with imposed equational constraints. If on one side the hint at approaching ologs via linear sketches was already contained in~\cite{DBLP:journals/corr/abs-1102-1889}, on the other side, as far as we know, looking at ologs as particular labelled deductive systems is in our opinion original at least up to a certain extent. In section~\ref{calculus} we reprise the fundamentals of syllogistic and describe the previously mentioned diagrammatic logical calculus for it. Although most of the contents of the whole section can be found in the bibliographic references cited above, they have been almost completely reworked in view of the main aim of this paper. From subsection~\ref{syllogistic} to subsection~\ref{contradiction} the approach is informal to let the reader to become familiar with the employment of the calculus at issue; whereas in subsection~\ref{justreq} that calculus is formally described within a suitable diagrammatic deductive system. In section~\ref{ologisms} we introduce ologisms informally to let the reader already familiar with ologs and syllogistic to understand how they comprise ologs as well as syllogistic and moreover, how it is possible to express logical contradiction in them. In section~\ref{extlabelded} we describe ologisms formally as deductive systems arising from a combination of ologs as labelled deductive systems and syllogistic in the form of the diagrammatic deductive system described in section~\ref{syllogistic}. We also describe the logical theory generated by an ologism. After section~\ref{models} in which we provide a definition of what a model of an ologism should be, in section~\ref{soundcompl} we introduce suitable logical and semantical consequence relations pertaining to ologisms and prove appropriate soundness and completeness theorems for those relations. As far as we know, the content of the sections~\ref{ologisms},~\ref{extlabelded},~\ref{models} and~\ref{soundcompl} is original. As already mentioned, section~\ref{contraologisms} is dedicated to the managing of logical contradiction in ologisms. More to the point,  in that section we discuss how the occurrence of pairs of contradictory syllogistic constraints in ologisms give rise to contradiction in them, so decreeing their structural unsoundness. Finally, section~\ref{forlogrepr} contains a comparison between ologisms and DL-KRS as regards some of the basic features shared by the two formalisms. We are aware that the connections between ologisms and other formal systems should be pointed out and extensively discussed. Among the ones to which we are thinking of there is intensional logic, term logic as regards in particular the problem of multiple generality, and the semantic web, for example. 
If on one side ologisms have been viewed by us as extending the logical expressivity of ologs in the way we previously briefly described , on another side one could look at ologisms as extending syllogistic reasoning in the form of the previously hinted at diagrammatic calculus by means of ologs providing an ontology for the terms involved in syllogisms. This viewpoint furnishes further motivations for future work that has to be done.


\section{Ologs}\label{ologs}
In this section we describe the fundamentals of {\bf ologs} in a rather informal way in order to establish terminology and notations that will be used throughout the paper and to let the reader acquire a certain familiarity with the subject, that will be presented in rigorous mathematical terms in section~\ref{ologlinear}. The reader who is interested in deepening her knowledge about ologs is invited to consult~\cite{MR3288752, DBLP:journals/corr/abs-1102-1889} and the bibliographic references therein. The way in which the content of this section is arranged derives from a personal rereading of related contents in the just cited bibliographic references. The word ``olog'' comes from the expression ``ontology log'' as a sort of crasis. Ologs provide category-theoretic models for knowledge representation but very little category theory is required in order to understand and start using them. Actually, some acquaintance with the notions of graph and function is sufficient. Nonetheless, the fact that the construction and the employment of ologs are in accordance with rules that are justifiable in rigorous mathematical terms has to be firmly kept in mind because it makes ologs {\bf structured models for knowledge representation}, so that it makes sense to consider knowledge-preserving transformations between them, for instance. Ologs are stricly connected to the {\bf conceptual planning of databases}. In fact, they can be read as relational database schemes and their design aims to be much more simple than learning a database definition language. Although the content of this paper is related to the conceptual planning of databases via ologs, we will not describe the details of how they correspond to relational database schemes but refer the reader to the previously cited bibliographic references, from which we have taken some of the examples in this section.

\subsection{Authoring ologs}
Following~\cite{MR3288752,DBLP:journals/corr/abs-1102-1889} we will henceforth refer to {\bf the author} of an olog as to the fictitious person who designed it. The action of designing an olog will be henceforth referred to as {\bf authoring}. 

\begin{rem}\label{perdopo}
As underlined in~\cite{MR3288752,DBLP:journals/corr/abs-1102-1889}, an olog captures a part of a {\bf subjective world-view} of its author. Thus, it may well be that a  person A examining the olog designed by a person B does not agree with the world-view captured by B's olog, but the dispute that consequently may arise is not a problem with the olog under consideration. As will be made clear in what follows, ologs are designed in accordance with rules that ensure that they be structurally sound rather than they faithfully reflect reality. 
\end{rem}

Types, aspects and facts are the fundamental components of an olog. In the rest of this section we will describe each of them in order.

\subsection{Types}\label{types}
In an olog, a type is represented by a box containing a singular indefinite phrase describing the generic representative example of a whole class of things. Examples of types are
$$
\begin{array}{lllllll}
\fbox{a man}&&&\fbox{a mammal}&&&
\tbox{.35}{a pair $(p,(l,r))$ where $p$ is a person and $(l,r)$ is a pair of shoes}
\end{array} 
$$
and it is a rule of good practice to explicitly provide variable names for the instances of a compound type such as the last two above. That is, it is preferable to replace a compound type such as 
$$
\fbox{a man and a woman}
$$
by 
$$
\fbox{a man $m$ and a woman $w$}
$$  
or also by
$$
\fbox{a pair $(m,w)$ where $m$ is a man and $w$ is a woman}
$$

\begin{notation}
For typographical reasons, it may be convenient to write types in-line by just hinting at the box surroinding them, like this: $\ulcorner\textrm{a man $m$ and a woman $w$}\urcorner$.
\end{notation}

The {\bf rules of good practice for authoring structurally sound types} prescribe that they be represented by a box with text inside that has to
\begin{itemize}
\item[-] begin with an indefinite article ``a'' or ``an'',
\item[-] refer to a distinction made and recognizable by the author,
\item[-] refer to a distinction for which instances can be documented,
\item[-] be the common name that each instance of that distinction can be called,
\item[-] not end by a punctuation mark,
\item[-] declare all variables in a compound type.
\end{itemize}

\subsection{Aspects}
Aspects apply to types. An aspect of a thing is a way of looking at it or measuring it. In an olog, aspects are represented by labelled arrows emanating from one type toward another. For types $A, B$ an aspect 
$\xymatrix{A\ar[r]^f& B}$ is an aspect of the things in $A$ that has things in $B$ as results of certain measurements performed in the way described by the label $f$. In an olog, the label of an aspect must be a fragment of text beginning with a verb, yielding a well-formulated English phrase when the reading of the text in the box of the type from which the aspect starts is followed by the reading of the label itself further followed by the reading of the text in the box of the type to which the aspect arrives. For instance, to establish that a woman is a person amounts to pointing out an aspect of a woman, which is represented as a labelled arrow
\begin{eqnarray}\label{isaperson}
\xymatrix{\fbox{a woman}\ar[rr]^{\al{\textrm{is}}}&&\fbox{a person}}
\end{eqnarray}
in an olog, and has to be read ``a woman is a person''. 
A further example of an aspect in an olog is
$\xymatrix{\ulcorner\textrm{a book}\urcorner\ar[rrrr]^{\al{\textrm{has as first author}}}&&&&\ulcorner\textrm{a person}\urcorner}$, which has to be read ``a book has as first author a person''.
Moreover, {\bf the label of a structurally sound aspect must describe a functional relationship}. Intuitively, if the types of an olog are interpreted as sets, as it will be rigorously done in section~\ref{ologlinear}, then an aspect $\xymatrix{A\ar[r]^f&B}$ will be interpreted as a function from the set interpreting $A$ to the set interpreting $B$.

\begin{notation}
Types and aspects allow the authoring of first examples of interesting ologs, such as 
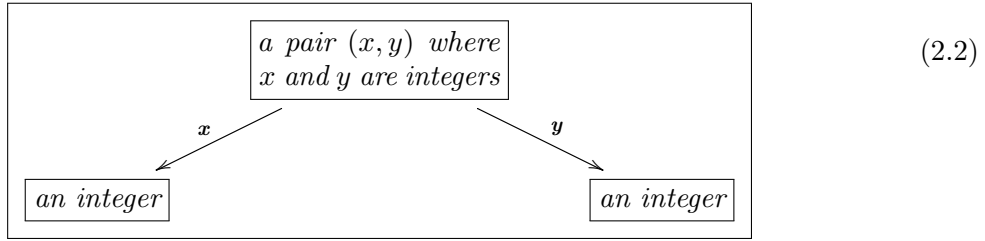
\begin{eqnarray}\label{alloscopo}
\fbox{\xymatrix{&\fbox{\parbox{.21\textwidth}{a pair $(x,y)$ where $x$ and $y$ are integers}}
\ar[dl]_(.6){\al{x}}\ar[dr]^(.6){\al{y}}
\\
\fbox{an integer}&&\fbox{an integer}
}}
\end{eqnarray}
which shows that in practice, if confusion is not likely to arise, the label of an aspect can be shortened or dropped altogether. With reference to the olog~\eqref{alloscopo} one can observe that the expression ``a pair $(x,y)$ where $x$ and $y$ are integers $x$ an integer'' is not a well-formulated English phrase. Actually, the $\al x$ labelling the arrow is an abbreviation for ``yields, as the value of $x$,''. So, in the end, the aspect in question has to be read as ``a pair $(x,y)$ where $x$ and $y$ are integers yields, as the value of $x$, an integer'' which is a well-formulated English phrase; also making clear that the single character $\al x$ has been used 
to make reference to the intended interpretation of the aspect that it is labelling, namely the projection function
$$
\al x:\xymatrix{\fbox{\parbox{.21\textwidth}{a pair $(x,y)$ where $x$ and $y$ are integers}}\ar[rr]^(.6){(x,y)\mapsto x}&&\fbox{an integer}}
$$ 
\end{notation}
To summarize, the {\bf rules of good practice for authoring structurally sound aspects} prescribe that they be represented by labelled arrows emanating from a type toward a type and that the label of the aspect be a fragment of text that must
\begin{itemize} 
\item[-] begin with a verb,
\item[-] yield a well-formulated English phrase, when the reading of the text in the box of the type from which the aspect starts is followed by the reading of the label of its arrow further followed by the reading of the text in the box of the type to which the aspect arrives,
\item[-] declare a functional relationship so that each instance of the type from which the aspect starts is related to a uniquely designated instance of the type to which the aspect arrives. 
\end{itemize}

\subsection{Paths of aspects}\label{pathasp}
In an olog there can be aspects in a row, that is aspects which follow one the other via intermediate types that link them. Such {\bf paths of aspects} have to be read by inserting the word ``which'' or ``who'' after the text of each intermediate type. For instance, in the olog
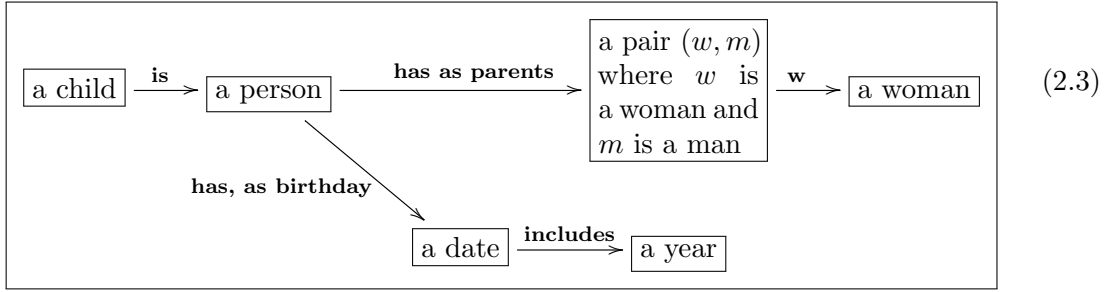
\begin{eqnarray}\label{ologi}
\fbox{\xymatrix{\fbox{a child}\ar[r]^(.44){{\al{\textrm{is}}}}&\fbox{a person}\ar[dr]_{\al{\textrm{has, as birthday}}}\ar[rr]^{\al{\textrm{has as parents}}}&&\fbox{\parbox{.14\textwidth}{a pair $(w,m)$ where $w$ is a woman and $m$ is a man}}\ar[r]^{\al{w}}&\fbox{a woman}\\
&&\fbox{a date}\ar[r]^{\al{\textrm{includes}}}&\fbox{a year}
}}
\end{eqnarray}
there are two different paths of aspects: the one from $\ulcorner\textrm{a child}\urcorner$ to  $\ulcorner\textrm{a woman}\urcorner$ and the one from $\ulcorner\textrm{a child}\urcorner$ to $\ulcorner\textrm{a year}\urcorner$. They have to be read as ``a child is a person who has as parents a pair $(w,m)$ where $w$ is a woman and $m$ is a man which yields, as the value of $w$, a woman'' and as ``a child is a person who has, as birthday a date which includes a year'', respectively.

\begin{rem}\label{forpathasp}
For future reference we here observe in passing that this way of reading paths of aspects in an olog provides a {\bf law of composition of aspects}. This point will be reprised in subsection~\ref{facts} below, further discussed in section~\ref{ologlinear}.
\end{rem}

\subsection{Facts}\label{facts}
In an olog, facts amount to the declaration of equational identities between pairs of aspects. It is the possibility of declaring facts in ologs that peculiarly makes explicit the category theoretic ground on which their authoring is based, and makes ologs much more expressive than they would be if they were just graphs with types as nodes and aspects as arcs. In an olog, the possibility of expressing facts arises because a pair of aspects one following the other, say $\xymatrix{A\ar[r]^f&B\ar[r]^g& C}$, can be read one after the other as illustrated in subsection~\ref{pathasp}, giving rise to a functionally designated {\bf composite aspect} $\xymatrix{A\ar[r]^{f;g}&C}$, so that the aspects $f$, $g$ and $f;g$ fit in what is, by construction, a paradigmatic example of {\bf commutative diagram}, that is
\begin{eqnarray}\label{compcomm}
\xymatrix{A\ar@/_/[drr]_{f;g}^{\quad\checkmark}\ar[rr]^f&&B\ar[d]^g\\
&&C}
\end{eqnarray}
declaring that the result of measuring via $f$ and then via $g$ is equal to the result of measuring in one step via $f;g$, by the very definition of composite aspect; and the checkmark inside the diagram is there to communicate precisely this, namely that the diagram~\eqref{compcomm} is commutative, thus representing a fact.
\begin{rem}
Diagrams and commutative diagrams 
will be rigorously discussed in section~\ref{ologlinear}. 
\end{rem}

\begin{exm}
As an example of a fact which holds 
by construction in the sense just illustrated, one can consider
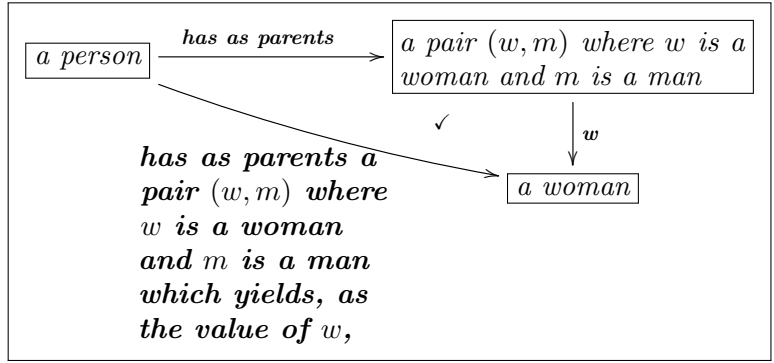
\begin{eqnarray*}
\fbox{\xymatrix{\fbox{a person}\ar@/_/[drrr]_{\ttbox{.2}{
\al{has as parents a pair $(w,m)$ where $w$ is a woman and $m$ is a man which yields, as the value of $w$,}}\quad}\ar[rrr]^(.35){\al{\textrm{has as parents}}}&&&\fbox{\parbox{.3\textwidth}{a pair $(w,m)$ where $w$ is a woman and $m$ is a man}}\ar@{}[dl]|{\checkmark}\ar[d]^(.6){\al{w}}\\
&&&\fbox{a woman}}}
\end{eqnarray*}
where, to be explicit, the labels involved are
\begin{eqnarray*}
f=\al{has as parents}\\[1ex]
g=\al{w}=\al{yields, as the value of $w$,}\\[1ex]
f;g=\ttbox{.2}{\al{has as parents a pair $(w,m)$ where $w$ is a woman and $m$ is a man which yields, as the value of $w$,}}
\end{eqnarray*}
but there is to say that a fact that holds by construction is not truly informative. The informative facts are those which are imposed to clarify matters. For instance, the previous fact can be turned into the following more informative one:
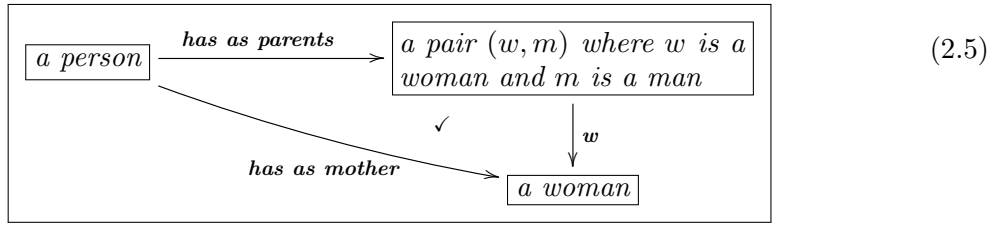
\begin{eqnarray}\label{has-mother}
\fbox{\xymatrix{\fbox{a person}\ar@/_/[drrr]_(.6){\al{\textrm{has as mother}\qquad}}\ar[rrr]^(.35){\al{\textrm{has as parents}}}&&&\fbox{\parbox{.3\textwidth}{a pair $(w,m)$ where $w$ is a woman and $m$ is a man}}\ar@{}[dl]|{\checkmark}\ar[d]^(.6){\al{w}}\\
&&&\fbox{a woman}}}
\end{eqnarray}
declaring that a woman parent is nothing but a mother, and the checkmark is there to communicate this. Explicitly, the imposed equational identity is
\begin{eqnarray*}
\ttbox{.2}{\al{has as parents a pair $(w,m)$ where $w$ is a woman and $m$ is a man which yields, as the value of $w$,}}=\al{has as mother}
\end{eqnarray*}

However, the rigorous way to impose equational identities between paths of aspects in an olog will be discussed in subsection~\ref{ologglinear}.
\end{exm}

\begin{rem}\label{parpath}
Although maybe already clear, it is worth to clarify that {\bf facts declare equational identities between parallel aspects}, where two aspects in an olog are said to be parallel if they emanate from the same type and point toward the same type. 
\end{rem}

The {\bf rules of good practice for authoring structurally sound facts} prescribe that facts are declared by drawing a checkmark in the planar space that a parallel pair of aspects circumscribe on the page or by means of an equational identity between them. Every fact in an olog must be so explicitly declared, no matter how obvious it may seem.


\section{Ologs as labelled deductive systems with imposed equational constraints}\label{ologlinear}
In section~\ref{ologs} we introduced and discussed ologs in a rather informal way. On the contrary, the aim of this section is to describe ologs first as linear sketches and then, to complete the description, as labelled deductive systems with imposed equational constraints, both of which are rigorously defined mathematical structures. For selfcontainment, we recall all the  fundamental notions that will be needed, mainly concerning graphs, categories, diagrams and commutative diagrams in categories, although not extensively but as much as will be useful for the sequel of the paper. The reader who is interested in deepening her knowledge about the notions that will be introduced in this section is referred to~\cite{MR2981171} and~\cite{MR939612}, for instance.

\subsection{Graphs and commutative diagrams}

\begin{defi}\label{graph}
A {\bf graph} is a $4$-tuple $\mathcal G=(G_0, G_1, \partial_0,\partial_1)$ where
\begin{itemize}
\item[-] $G_0$ is a set of {\bf nodes}, that will be usually denoted by upper case letters $A, B, C, \ldots$;
\item[-] $G_1$ is a set of {\bf arcs}, that will be usually denoted by lower case letters $f, g, h,\ldots$; 
\item[-] $\partial_0:G_1\rightarrow G_0$ is a function that will be referred to as {\bf source};
\item[-] $\partial_1:G_1\rightarrow G_0$ is a function that will be referred to as {\bf target}.
\end{itemize}
For any two nodes $A, B$ of $\mathcal G$, a {\bf path from} $A$ {\bf to} $B$ {\bf in} $\mathcal G$ is a $3$-tuple $(A,(f_1,f_2,\ldots,f_n),B)$ where $(f_1,f_2,\ldots,f_n)\in G_1^n$, 
$A=\partial_0(f_1)$, $B=\partial_1(f_n)$ and, for every $i=1,\ldots,n-1$, $\partial_1(f_i)=\partial_0(f_{i+1})$. Two paths of $\mathcal G$ are said to be {\bf parallel} if they have the same first and third component, that is if they start and end at the same nodes.
\end{defi}

\begin{rem}\label{zerolength}
The notion of path in a graph makes sense for $n=0$ too. In that case, the notion identifies the {\bf empty path} on the node $A$ coinciding with $B$, that is $(A,(),A)$.
\end{rem}

\begin{exm}
Most of the time we will be concerned with graphs with a finite number of nodes and a finite number of arcs. Because of this, we will be able to specify their components in explicit set-theoretic terms. For instance, the $4$-tuple 
\begin{eqnarray}\label{nn}
\mathcal G=(\{A,B,C\},\{f,g\},\{(f,A),(g,B)\},\{(f,B),(g,A)\})
\end{eqnarray}
identifies the graph whose components are 
\begin{eqnarray*}
G_0&=&\{A,B,C\}\\
G_1&=&\{f,g\}\\
\partial_0:\left[\begin{array}{l}
f\mapsto A\\
g\mapsto B
\end{array}\right]&:&G_1\rightarrow G_0\\
\partial_1:\left[\begin{array}{l}
f\mapsto B\\
g\mapsto A
\end{array}\right]&:&G_1\rightarrow G_0
\end{eqnarray*}
but observe that in the $4$-tuple~\eqref{nn} the functions $\partial_0$, $\partial_1$ have been given as the functional relations $\partial_0=\{(f,A),(g,B)\}\subseteq\{f,g\}\times\{A,B,C\}$, $\partial_1=\{(f,B),(g,A)\}\subseteq\{f,g\}\times\{A,B,C\}$, respectively.
\end{exm}

\begin{notation}
For $\mathcal G$ a graph and $f$ one of its arcs, the writing $\xymatrix{f:A\ar[r]&B}$ or the writing $\xymatrix{A\ar[r]^f&B}$, equivalently, will be used to indicate that the source of $f$ is the node $A$ and that the target of $f$ is the node $B$. 
\end{notation}

\begin{defi}\label{shapediagr}
For $\mathcal G$ and $\mathcal H$ graphs, a {\bf homomorphism of graphs} from $\mathcal G$ to $\mathcal H$ is a pair $F=(F_0,F_1)$ whose components are functions
$F_0:G_0\rightarrow H_0$, $F_1:G_1\rightarrow H_1$ such that for every arc $f:A\rightarrow B$ of $\mathcal G$, $F_1(f):F_0(A)\rightarrow F_0(B)$ in $\mathcal H$. A homomorphism of graphs from $\mathcal G$ to $\mathcal H$ will be written $F:\mathcal G\rightarrow\mathcal H$ and also referred to as a {\bf diagram of shape} $\mathcal G$ {\bf in} $\mathcal H$.
\end{defi}

\begin{rem}
When confusion is not likely to arise, sometimes we will identify a diagram with the drawing of the graph which is its image. For instance, let $\mathcal G$ be the graph~\eqref{nn} and  $\mathcal H$ be the graph 
$$
(\{0,1,2\},\{01,10\},\{(01,0),(10,1)\},\{(01,1),(10,0)\}).
$$
The diagram
$F:\mathcal H\rightarrow\mathcal G$ identified by the assignments
\begin{eqnarray*}
F_0:\left[\begin{array}{l}
0\mapsto A\\
1\mapsto B\\
2\mapsto C
\end{array}\right]:H_0\rightarrow G_0
\\[1ex]
F_1:\left[\begin{array}{l}
01\mapsto f\\
10\mapsto g\\
\end{array}\right]:H_1\rightarrow G_1
\end{eqnarray*}
has image which can be represented by the drawing
$$
\xymatrix{A\ar@/^/[rr]^f&&B\ar@/^/[ll]^g&C}
$$
\end{rem}

\begin{defi}\label{category}
Let $\mathcal G$ be a graph. Put 
$G_2\doteq\{(f,g)\in G_1\times G_1\mid \partial_1(f)=\partial_0(g)\}$.
$\mathcal G$ is a {\bf category} if it is equipped  with functions
{\bf composition} $\circ :(f,g)\mapsto f\circ g:G_2\rightarrow G_1$, and {\bf identity} $id:A\mapsto id_A: G_0\rightarrow G_1$, such that
\begin{enumerate}
\item for every $(f,g)\in G_2$, $\partial_0(f\circ g)=\partial_0(f)$ and $\partial_1(f\circ g)=\partial_1(g)$;
\item for every $f,g ,h\in G_1$, if $(f,g)\in G_2$ and $(g,h)\in G_2$, then $f\circ(g\circ h)=(f\circ g)\circ h$;
\item for every $A\in G_0$, $\partial_0(id_A)=\partial_1(id_A)=A$;
\item for every $f:A\rightarrow B$, $id_A\circ f=f\circ id_B=f$.
\end{enumerate}
Arbitrary categories will be henceforth usually denoted by upper case letters in blackboard style, such as $\mathbb C, \mathbb D,\ldots$, whereas particular categories will be henceforth usually denoted by upper case letters in bold face style.
\end{defi}

\begin{exm}\label{catexamples}
\hfill
\begin{enumerate}
\item
Sets as nodes, and functions between them as arcs, identify a category ${\bf Sets}$: 
\begin{itemize}
\item[-] for fuctions $f:a\mapsto f(a):A\rightarrow B$ and $g:b\mapsto g(b):B\rightarrow C$, their composition is the function
$f\circ g:a\mapsto g(f(a)):A\rightarrow C$;
\item[-] for $A$ a set, the identity arc at $A$ is the identity function $id_A:a\mapsto a:A\rightarrow A$.
\end{itemize}
\item Graphs as nodes, and homomorphisms of graphs between them as arcs, identify a category ${\bf Grph}$.
\begin{itemize}
\item[-] for homomorphisms of graphs $F=(F_0,F_1):\mathcal G\rightarrow \mathcal H$ and $G=(G_0,G_1):\mathcal H\rightarrow\mathcal K$, their composition is $F\bullet G=(F_0\circ G_0, F_1\circ G_1):\mathcal G\rightarrow\mathcal K$, where $\circ$ is composition in ${\bf Sets}$;
\item[-] for $\mathcal G$ a graph, the identity arc at $\mathcal G$ is $id_{\mathcal G}=(id_{G_0},id_{G_1}):\mathcal G\rightarrow\mathcal G$, with $id_{G_0}$, $id_{G_1}$ the identity functions at $G_0$, $G_1$ in ${\bf Sets}$, respectively.
\end{itemize}
\item\label{freecat} The {\bf free category generated by a graph}  $\mathcal G$ is $\mathcal G^*$, whose nodes are the nodes of $\mathcal G$, whose arcs from $A$ to $B$, say, are the paths from $A$ to $B$ in $\mathcal G$.
\begin{itemize}
\item[-] for arcs $(f_1,f_2,\ldots,f_n):A\rightarrow B$ and $(g_1,g_2,\ldots,g_m):B\rightarrow C$ in $\mathcal G^*$, that is for paths $(A,(f_1,f_2,\ldots,f_n),B)$, $(B,(g_1,g_2,\ldots,g_m),C)$ in $\mathcal G$, respectively, their composition is the arc
$(f_1,f_2,\ldots,f_n,g_1,g_2,\ldots,g_m):A\rightarrow C$
obtained by concatenation of the paths under consideration;
\item[-] for $A$ a node, the identity arc at $A$ is the empty path at $A$, that is $(A,(),A)$, see remark~\ref{zerolength}. 
\end{itemize}
\end{enumerate}
\end{exm}

\begin{defi}
Let $\mathcal G$ be a graph and let $\mathbb C$ be a category. A {\bf commutative diagram of shape} $\mathcal G$ {\bf in} $\mathbb C$ is a diagram $F:\mathcal G\rightarrow\mathbb C$ such that for every nodes $A, B$ of $\mathcal G$, for every parallel paths $(A,(f_1,f_2,\ldots,f_n),B)$, $(A,(g_1,g_2,\ldots,g_m),B)$ in $\mathcal G$, 
\begin{eqnarray}\label{withmeaning}
F_1(f_1)\circ F_1(f_2)\circ\cdots\circ F_1(f_n)=F_1(g_1)\circ F_1(g_2)\circ\cdots\circ F_1(g_m)
\end{eqnarray}
in $\mathbb C$.
\end{defi}

\begin{rem}
After remark~\ref{zerolength}, if, for instance, $n=0$, the identity~\eqref{withmeaning} amounts to
$id_A=F_1(g_1)\circ F_1(g_2)\circ\cdots\circ F_1(g_m)$ in $\mathbb C$.
\end{rem}

\subsection{Linear sketches}\label{llsketch}
Linear sketches are graphs with imposed commutativity conditions. Whereas equational theories are usually specified by means of formal languages of symbolic nature, linear sketches are graphical formal languages that allow the finite specification of the equational theory generated by a signature with unary operations only; namely, the equational theory which is the free category on the underlying graph of the sketch, quotiented by the smallest congruence relation generated by the imposed commutativity conditions. This construction will be described in subsection~\ref{eqsketch} below.

\begin{defi}\label{lsketch}
A {\bf linear sketch} is a pair $\mathfrak S=(\mathcal G, \mathcal D)$ where $\mathcal G$ is a graph and $\mathcal D$ is a set of diagrams in $\mathcal G$. For $\mathbb C$ a category, a {\bf model} of $\mathfrak S$ in $\mathbb C$ is a homomorphism of graphs $\mathcal M:\mathcal G\rightarrow\mathbb C$ such that for every diagram $D:\mathcal S\rightarrow \mathcal G$ in $\mathcal D$, the diagram $\mathcal M\bullet D:\mathcal S\rightarrow \mathbb C$ is commutative. A model of  $\mathfrak S$ in ${\bf Sets}$ will be henceforth referred to as {\bf set-theoretic}, but see remark~\ref{isa} below. 
\end{defi}

\begin{rem}\label{commevery}
Commutative diagrams are the category-theoretic way to express equations and in a sketch $\mathfrak S= (\mathcal G,\mathcal D)$ the elements of $\mathcal D$ represent the equations that must be true in every model of the sketch.
\end{rem}

\begin{exm}\label{forfut}
\begin{enumerate}
\item The {\bf linear sketch for graphs} is $\mathfrak G=(\mathcal G,\emptyset)$ where 
$$\mathcal G=(\{N,A\},\{s,t\},\{(s,A), (t,A)\},\{(s,N),(t,N)\})
$$ 
which can be more conveniently seen by drawing it as
$$
\xymatrix{A\ar@<.8ex>[r]^{s}\ar@<-.8ex>[r]_t&N}
$$
namely as ``the graph of graphs''.
\item\label{catsketch} The {\bf linear sketch underlying a category} $\mathbb C$ is $\mathfrak C=(\mathcal C,\mathcal D)$ where $\mathcal C$ is the graph underlying $\mathbb C$ and $\mathcal D$ is the set of all the commutative diagrams in $\mathbb C$.
\end{enumerate}
\end{exm}

\subsection{The equational theory generated by a linear sketch}\label{eqsketch}
If on one hand, a category identifies a linear sketch, 
see point~\eqref{catsketch} in examples~\ref{forfut}, on the other hand a linear sketch $\mathfrak S = (\mathcal G,\mathcal D)$ identifies a category $Th(\mathfrak S)$ which is the equational theory generated by $\mathfrak S$. In a way that we will rigorously make precise below, the category $Th(\mathfrak S)$ is a quotient of $\mathcal G^*$, the free category generated by the graph $\mathcal G$, see point~\eqref{freecat} of the examples~\ref{catexamples}.

\begin{defi}
Let $\mathbb C$ be a category with underlying graph $\mathcal C=(C_0,C_1,\partial_0,\partial_1)$, and let $\sim\subseteq C_1\times C_1$ be a relation. The {\bf congruence relation generated by} $\sim$ is the smallest equivalence relation $\sim^*$ that contains $\sim$ such that
\begin{enumerate}
\item for every $(f,g)\in C_1\times C_1$, $(f,g)\in\sim^*$ only if $\partial_0(f)=\partial_0(g)$ and $\partial_1(f)=\partial_1(g)$;
\item for every $f,g,h,k\in C_1$ as in the diagram
$$
\xymatrix{A\ar[r]^h&B\ar@<.8ex>[r]^f\ar@<-.8ex>[r]_g&C\ar[r]^k&D}
$$ 
if $f\sim^*g$ then $h\circ f\sim^*h\circ g$ and $f\circ k\sim^*g\circ k$.
\end{enumerate}
The {\bf quotient category of} $\mathbb C$ {\bf by} $\sim^*$ is $\mathbb C\slash\sim^*$, whose nodes are the nodes of $\mathbb C$, whose arcs are $\sim^*$-equivalence classes of arcs of $\mathbb C$ with composition and identity via representatives.
\end{defi}

\begin{defi}\label{tildestar}
Let $\mathfrak S=(\mathcal G,\mathcal D)$ be a linear sketch. The {\bf equational theory generated by $\mathfrak S$} is the category $Th(\mathfrak S)$ which is the quotient category $\mathcal G^*\slash\sim_{\mathcal D}$ where $\sim_{\mathcal D}$ is the smallest congruence relation that contains all the pair of parallel paths of $\mathcal G^*$ from any diagram $D\in\mathcal D$.
\end{defi}

\begin{rem}
Let $\mathfrak S=(\mathcal G,\mathcal D)$ be a linear sketch. Every diagram in $\mathcal D$ becomes commutative in $Th(\mathfrak S)$ and, in a suitable sense which can be made technically precise, the construction of $Th(\mathfrak S)$ out of $\mathfrak S$ is the best way to make this to happen. For details, the reader is referred to~\cite{MR2981171}.
\end{rem}

\begin{exm}\label{equivalents}
The equational theory generated by the linear sketch underlying a category $\mathbb C$, see point~\eqref{catsketch} in examples~\ref{forfut}, is isomorphic to $\mathbb C$, that is $\mathbb C\simeq Th(\mathfrak C)$.
\end{exm}

\subsection{Ologs as linear sketches}\label{ologglinear}
Ologs are made of types, aspects and facts to be rigorously specified in accordance with respective rules of good practice, described in section~\ref{ologs}. Also, as observed in subsection~\ref{facts}, it is because of the possibility of specifying facts, as a way to impose equational identities between aspects, that ologs are more expressive than they would be if they were just graphs with types as nodes and aspects as arcs. From this and what have been previously discussed in this section it follows that an olog can be described as a linear sketch $\mathfrak O=(\mathcal O,\mathcal F)$ where $\mathcal O$ is a graph with types as nodes and aspects as arcs, and $\mathcal F$ is a set of diagrams in $\mathcal O$, to be referred to as facts, that have to be commutative in every set-theoretic model $\mathcal M:\mathcal O\rightarrow{\bf Sets}$ of the sketch $\mathfrak O$, in particular, see remark~\ref{isa} below. 

\begin{rem}\label{story}
We hasten to warn the reader that this is not the end of the story. In subsection~\ref{pathasp} we described how paths of aspects in an olog have to be read, and in remark~\ref{forpathasp} we observed that that way of reading paths of aspects in an olog provides a law of composition of pairs of aspects in a row. 
The possibility of formally composing arcs in a graph, hence in particular of formally composing aspects in (the underlying graph of) an olog, is guaranteed in the so-called labelled deductive systems, that we will discuss in subsection~\ref{labelded}. Thus, ologs are linear sketches but {\it prima facie}. For the moment we are happy with this but, in section~\ref{ologwitheq}, we will recognize ologs under a more proof-theoretic perspective, as labelled deductive systems with imposed equational constraints on their formal proofs. This is needed for the pursueing of the main aim of this paper, which is to increase the logical expressivity of ologs with the possibility of considering syllogistic constraints beyond the equational ones.
\end{rem}

\begin{rem}\label{isa}
In this paper the models of ologs that we will consider will be exclusively set-theoretic. With respect to this, we draw the attention of the reader to observe that among the others, the aspects of an olog whose label is the predicate \al{is} are somewhat peculiar, because such an aspect, say $\xymatrix{A\ar[r]^{\al{is}}&B}$, indicates that 
something relevant to the instances of type $A$ is forgot to subsume them in a class of instances of type $B$, which are at least as general as those of type $A$, if not even genuinely more general. In view of this, it appears natural to request that those aspects be always interpreted as inclusion functions by any set-theoretic model. More explicitly, for $\mathfrak O$ an olog, $\xymatrix{A\ar[r]^{\al{is}}&B}$ an aspect of it and $\mathcal M:\mathcal O\rightarrow{\bf Sets}$ a set-theoretic model, we let $\al{is}$ to be modelled by the inclusion function of the set $\mathcal MA$ in the set $\mathcal MB$. 
\end{rem}

\begin{rem}\label{remrem}
In~\cite{DBLP:journals/corr/abs-1102-1889} it is pointed out that ologs are based on the mathematical notion of category. More precisely, in {\em loc. cit.} it is pointed out that an olog is a presentation of a category by generators, that is types and aspects, and relations, that is facts; moreover, it is also pointed out that this form of presentation for categories allows to separate the strictly graphical part of an olog, namely its types and aspects, from its propositional, we would say equational, part, provided by facts. By describing ologs by means of linear sketches we wanted to rigorously take account of this, since, towards the achievement of the purposes of this paper, see remark~\ref{story}, linear sketches are the mathematical structures that allow to neatly keep track of this sort of dichotomy. We wanted to apply a genuine principle of specification to make a useful distinction between two separate pieces of structure, although strictly interacting.
Now, the moral is the following: from a purely abstract mathematical point of view considering ologs as presentations of categories or as linear sketches amounts to the same thing, on the base of point~\eqref{catsketch} of example~\ref{forfut} and of example~\ref{equivalents}, but it is not like that in view of the purposes of this paper: we find that approaching ologs by means of linear sketches turned out to be more fruitful in respect to those.  
\end{rem}

\begin{exm}\label{nono}
The linear sketch $\mathfrak O=(\mathcal O,\mathcal F)$ that formally presents the olog~\eqref{has-mother}, is identified by the following data.
\begin{itemize}
\item[-] For $\mathcal O$:
\begin{itemize}
\item the set of types of $\mathcal O$ is
$$
\{\ulcorner\textrm{a person}\urcorner, \fbox{\parbox{.3\textwidth}{a pair $(w,m)$ where $w$ is a woman and $m$ is a man}},\ulcorner\textrm{a woman}\urcorner\}
$$
\item the set of aspects of $\mathcal O$ is
$$
\{\al{\textrm{has as parents}}, \al{yields, as the value of $w$,},\al{\textrm{has as mother}}\}
$$
\item the source and target functions are
\begin{eqnarray*}
\partial_0:\left[\begin{array}{l}
\al{\textrm{has as parents}}\mapsto\ulcorner\textrm{a person}\urcorner\\
\al{\textrm{has as mother}}\mapsto\ulcorner\textrm{a person}\urcorner\\
\al{yields, as the values of $w$,}\mapsto\fbox{\parbox{.14\textwidth}{a pair $(w,m)$ where $w$ is a woman and $m$ is a man}}
\end{array}\right]
\\
\partial_1:\left[\begin{array}{l}
\al{has as parents}\mapsto\fbox{\parbox{.14\textwidth}{a pair $(w,m)$ where $w$ is a woman and $m$ is a man}}\\
\al{has as mother}\mapsto\ulcorner\textrm{a woman}\urcorner\\
\al{yields, as the values of $w$,}\mapsto\ulcorner\textrm{a woman}\urcorner
\end{array}\right]
\end{eqnarray*}
\end{itemize}
\item[-] $\mathcal F$ has just one element $F$ that is the diagram of shape
$$
\xymatrix{0\ar[dr]_{02}\ar[r]^{01}&1\ar[d]^{12}\\
&2}
$$
in $\mathcal O$ identified by the assignments
\begin{eqnarray*}
F_0:\left[\begin{array}{l}
0\mapsto\ulcorner\textrm{a person}\urcorner\\
1\mapsto\fbox{\parbox{.14\textwidth}{a pair $(w,m)$ where $w$ is a woman and $m$ is a man}}\\
2\mapsto\ulcorner\textrm{a woman}\urcorner
\end{array}\right]
\\
F_1:\left[\begin{array}{l}
01\mapsto\al{\textrm{has as parents}}\\
02\mapsto\al{\textrm{has as mother}}\\
12\mapsto\al{\textrm{yields, as the values of $w$,}}
\end{array}\right]
\end{eqnarray*}
\end{itemize}
A model of $\mathfrak O=(\mathcal O,\mathcal F)$ 
is $\mathcal M:\mathcal O\rightarrow{\bf Sets}$ identified by the assignments
\begin{eqnarray*}
M_0:\left[\begin{array}{l}
\ulcorner\textrm{a person}\urcorner\mapsto\{\textrm{Michael}, \textrm{Diana}, \textrm{John}, \textrm{Mary}, \textrm{Susan}\}\\
\fbox{\parbox{.14\textwidth}{a pair $(w,m)$ where $w$ is a woman and $m$ is a man}}\mapsto\{(\textrm{Susan}, \textrm{Juan}), (\textrm{Elen}1, \textrm{Albert}), (\textrm{Elen}2, \textrm{Jerry})\}\\
\ulcorner\textrm{a woman}\urcorner\mapsto\{\textrm{Susan}, \textrm{Elen}1, \textrm{Elen}2, \textrm{Clare}\}
\end{array}\right]
\\
M_1:\left[\begin{array}{l}
\al{has as parents}\mapsto\left[\begin{array}{l}
\textrm{Michael}\mapsto(\textrm{Susan},\textrm{Juan})\\
\textrm{Diana}\mapsto(\textrm{Susan},\textrm{Juan})\\
\textrm{John}\mapsto(\textrm{Elen}1,\textrm{Albert})\\
\textrm{Mary}\mapsto(\textrm{Elen}2,\textrm{Jerry})\\
\textrm{Susan}\mapsto(\textrm{Elen}2,\textrm{Jerry})
\end{array}\right]
\\
\al{yields, as the value of $w$,}\mapsto\left[\begin{array}{l}
(\textrm{Susan},\textrm{Juan})\mapsto\textrm{Susan}\\
(\textrm{Elen}1,\textrm{Albert})\mapsto\textrm{Elen}1\\
(\textrm{Elen}2,\textrm{Jerry})\mapsto\textrm{Elen}2
\end{array}\right]
\\
\al{has as mother}\mapsto\left[\begin{array}{l}
\textrm{Michael}\mapsto\textrm{Susan}\\
\textrm{Diana}\mapsto\textrm{Susan}\\
\textrm{John}\mapsto\textrm{Elen}1\\
\textrm{Mary}\mapsto\textrm{Elen}2\\
\textrm{Susan}\mapsto\textrm{Elen}2
\end{array}\right]
\end{array}\right]
\end{eqnarray*}
for instance.
\end{exm}
The aim of the remaining part of this section is to complete the mathematical description of ologs in terms of linear sketches, by pointing out that ologs possess a proof-theoretic structure that have to be taken into account, especially if, as in our case, one is concerned with an investigation directed toward the extension of their logical expressivity. As already hinted at in remark~\ref{story} we will be dealt with so-called labelled deductive systems to see ologs as labelled deductive systems with imposed equational constraints, more precisely, in subsection~\ref{ologwitheq}. For selfcontainment we will recall all the fundamental notions that will be needed. 

\subsection{Labelled deductive systems}\label{labelded}
On labelled deductive systems our main reference is~\cite{MR939612}, from which we took most of the material concerning them, by conveniently rephrasing it.
\begin{defi}\label{labellabel}
A {\bf labelled deductive system} is a graph $\mathcal L$ with
\begin{enumerate}
\item distinguished arcs to be referred to as {\bf axioms}, among which, for every node $A$, there are specified {\bf identity axioms} $id_A:A\rightarrow A$;
\item {\bf rules of inference} for generating new arcs from old, among which, there is the rule  
\begin{eqnarray}\label{cut}
\AxiomC{$\xymatrix{A\ar[r]^f&B}$}
\AxiomC{$\xymatrix{B\ar[r]^g&C}$}
\BinaryInfC{$\xymatrix{A\ar[r]^{fg}&C}$}
\DisplayProof
\end{eqnarray}
\end{enumerate}
So, the nodes and the arcs of a labelled deductive system will be henceforth referred to as {\bf formulas} and {\bf labelled sequents}, respectively. A {\bf proof tree} of $\mathcal L$ is a tree where each node is a labelled sequent and each branching is an instance of an inference rule. 
\end{defi}

\begin{exm}\label{freecatlds} The {\bf free category generated by a labelled deductive system} $\mathcal L$ is the category $Th(\mathcal L)$ whose nodes are the formulas of $\mathcal L$, whose arcs are equivalence classes of labelled sequents of $\mathcal L$ by the smallest congruence relation generated by identifying
\begin{enumerate}
\item the proof tree
\begin{eqnarray*}
\AxiomC{$\xymatrix{A\ar[r]^f&B}$}
\AxiomC{$\xymatrix{B\ar[r]^g&C}$}
\BinaryInfC{$\xymatrix{A\ar[r]^{fg}&C}$}
\AxiomC{$\xymatrix{C\ar[r]^h&D}$}
\BinaryInfC{$\xymatrix{A\ar[r]^{(fg)h}&D}$}
\DisplayProof
\end{eqnarray*}
with the proof tree
\begin{eqnarray*}
\AxiomC{$\xymatrix{A\ar[r]^f&B}$}
\AxiomC{$\xymatrix{B\ar[r]^g&C}$}
\AxiomC{$\xymatrix{C\ar[r]^h&D}$}
\BinaryInfC{$\xymatrix{B\ar[r]^{gh}&C}$}
\BinaryInfC{$\xymatrix{A\ar[r]^{f(gh)}&D}$}
\DisplayProof
\end{eqnarray*}

\item the proof tree
\begin{eqnarray*}
\AxiomC{$\xymatrix{A\ar[r]^f&B}$}
\AxiomC{$\xymatrix{B\ar[r]^{id_B}&B}$}
\BinaryInfC{$\xymatrix{A\ar[r]^{fid_B}&B}$}
\DisplayProof
\end{eqnarray*}
with the proof tree
$$\xymatrix{A\ar[r]^f&B}$$

\item the proof tree
\begin{eqnarray*}
\AxiomC{$\xymatrix{A\ar[r]^{id_A}&A}$}
\AxiomC{$\xymatrix{A\ar[r]^f&B}$}
\BinaryInfC{$\xymatrix{A\ar[r]^{id_Af}&B}$}
\DisplayProof
\end{eqnarray*}
with the proof tree $$\xymatrix{A\ar[r]^f&B}$$
plus the identification conditions imposed by the inference rules of $\mathcal L$ other than~\eqref{cut}, if present. 
\end{enumerate}
\end{exm}

\begin{exm}\label{ldsfreely} 
The {\bf labelled deductive system freely generated} by a graph $\mathcal G$ is $L(\mathcal G)$ whose formulas are the nodes of $\mathcal G$, also referred to as {\bf atomic}, whose axioms are the arcs of $\mathcal G$ together with, for every formula $A$, the identity axioms $id_A:A\rightarrow A$, and~\eqref{cut} the sole rule of inference.
\end{exm}

\begin{exm}\label{equiv}
As a particular case of example~\ref{freecatlds} consider $Th(L(\mathcal G))$, namely 
the free category generated by the free labelled deductive system generated by a graph $\mathcal G$, see example~\ref{ldsfreely}. Explicitly, $Th(L(\mathcal G))$ is the category whose nodes are the nodes of $\mathcal G$, whose arcs are equivalent classes of labelled sequents of $L(\mathcal G)$ by the equivalence relation generated by exclusively the identification conditions (a), (b), (c) in example~\eqref{freecatlds}. By construction, in $Th(L(\mathcal G))$ the diagrams commute solely on the base of the axioms for the category structure, see definition~\ref{category} and see the discussion in the subsection~\ref{comeqproof} below. Moreover, the categories $Th(L(\mathcal G))$ and $\mathcal G^*$, see point~\eqref{freecat} of examples~\ref{catexamples}, are equivalent, in a sense which can be made technically precise; but for what concern us 
it is enough to observe that those categories have the same nodes and that an arc of $Th(L(\mathcal G))$, say from $A$ to $B$, has a canonical representative which can be written as a formal product $f_1f_2\cdots f_n:A\rightarrow B$ of arcs $f_1:A\rightarrow A_1, f_2:A_1\rightarrow A_2,\ldots,f_n:A_{n-1}\rightarrow B$ of $\mathcal G$, without occurrences of parenthesis and of identity axioms, because of the identification conditions (a), (b), (c), except for $n=0$ and $A$ necessarily coinciding with $B$, in which case $f_1f_2\cdots f_n=id_A:A\rightarrow A$. Now, up to unessential formal details such canonical representatives are nothing but morphisms of $\mathcal G^*$.
\end{exm}

\subsection{Commutative diagrams as equivalent proofs}\label{comeqproof}
Roughly, in category theory a coherence theorem is a result establishing that all the diagrams constructed in accordance with the definitory data of a categorial structure commute in all the categories that carry such a structure, that is in every model of the categorial structure at issue. So, coherence theorems are about categorial structures in abstract. The fundamental reference on the subject is~\cite{MR0170925}, but the reader is invited to 
consult~\cite{MR0330253, MR0283045} as well. After the pioneering work of Lambek, see~\cite{MR0235979, MR0242637, MR0349356, MR0340366}, proof-theoretic methods toward the obtainment of specific coherence theorems have been widely and successfully employed, see~\cite{MR1489307, MR1090747} for instance, giving rise to what is nowadays referred to as {\bf categorial proof theory}, on which the reader is invited to consult~\cite{dosen2004proof}.
In essence, for what concern us, the result that has to be kept in mind is that {\bf commutative diagrams coincide with equivalent proofs} in suitable labelled deductive systems, see~\cite{MR685975}; a fact that we will henceforth freely use.

\subsection{The equational theory generated by a linear sketch, reprised}\label{proofeqsketch}
Recall from subsection~\ref{eqsketch}, definition~\ref{tildestar}, the construction of the equational theory $Th(\mathfrak S)$ generated by a linear sketch $\mathfrak S=(\mathcal G,\mathcal D)$. On the base of the equivalence of the categories $\mathcal G^*$ and $Th(L(\mathcal G))$, see example~\ref{equiv}, the fact that commutative diagrams coincide with equivalent proofs, see subsection~\ref{comeqproof}, and the fact that the diagrams in $\mathcal D$ must commute in every model of $\mathfrak S$, see remark~\ref{commevery}, it is possible to give a proof-theoretical description of $Th(\mathfrak S)$. It is the category $Th(L(\mathcal G))$ plus the identification conditions on proof trees that derive from the identification of all the pairs of parallel paths in any diagram $D\in\mathcal D$. Explicitly, for every diagram $D$ and pair of parallel paths, say  $(A,(f_1,f_2,\ldots,f_n),B)$ and $(A,(g_1,g_2,\ldots,g_m),B)$ in the image of $D$ in $\mathcal G$, one has the identification relation
$Df_1Df_2\cdots Df_n\sim Dg_1Dg_2\cdots Dg_m$ in $Th(L(\mathcal G))$ that becomes equality 
in $Th(\mathfrak S)$.

\subsection{Ologs as labelled deductive systems with imposed equational constraints}\label{ologwitheq}
In subsection~\ref{ologglinear} we described ologs as linear sketches. 
In remark~\ref{story} we observed that that was not the end of the story, because the way the aspects in a row in an olog $\mathfrak O=(\mathcal O, \mathcal F)$ are read, see subsection~\ref{pathasp}, provides a law of composition of aspects which, by virtue of definition~\ref{labellabel}, turns the graph $\mathcal O$ into the labelled deductive system $L(\mathcal O)$, see example~\ref{ldsfreely}. More to the point, for $\mathfrak O=(\mathcal O,\mathcal F)$ an olog, the labelled deductive system $L(\mathcal O)$ has 
\begin{itemize}
\item[-] as formulas, the types of $\mathcal O$;
\item[-] as axioms, the aspects of $\mathcal O$, together with, for every type $A$, the identity aspect $\xymatrix{A\ar[r]^{\al{is}}&A}$;
\item[-] as rule of inference the sole rule
\begin{eqnarray}\label{cat}
\AxiomC{$\xymatrix{A\ar[r]^f&B}$}
\AxiomC{$\xymatrix{B\ar[r]^g&C}$}
\BinaryInfC{$\xymatrix{A\ar[r]^{f;g}&C}$}
\DisplayProof
\end{eqnarray}
with $f;g$ the label obtained by reading $f$ followed by $g$ as described in subsection~\ref{pathasp}.
\end{itemize}
As a further step, the application of the identification conditions listed in example~\ref{freecatlds} provides the free category $Th(L(\mathcal O))$ and, moreover, the quotient of this by the smallest congruence relation generated by the facts $\mathcal F$, provides the equational theory generated by $\mathfrak O$, $Th(\mathfrak O)$.

\begin{exm}
The olog~\eqref{has-mother} can be seen to be built in steps has just described. Consider the linear sketch $\mathfrak O=(\mathcal O,\mathcal F)$ that formally presents it in example~\ref{nono}. First, it is convenient to draw the graph $\mathcal O$ as
$$
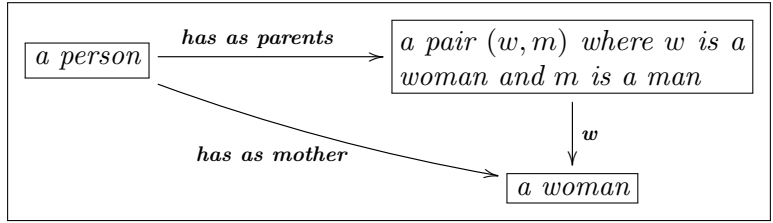
\fbox{\xymatrix{\fbox{a person}\ar@/_/[drrr]_{\al{has as mother\qquad}}\ar[rrr]^(.35){\al{\textrm{has as parents}}}&&&\fbox{\parbox{.3\textwidth}{a pair $(w,m)$ where $w$ is a woman and $m$ is a man}}\ar[d]^(.6){\al{w}}\\
&&&\fbox{a woman}}}
$$
then to draw the category $Th(L(\mathcal O)$ as
$$
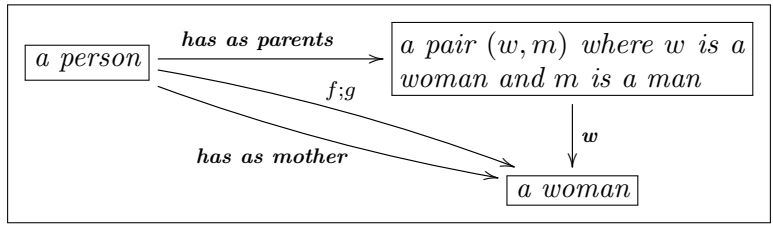
\fbox{\xymatrix{\fbox{a person}\ar@/^/[drrr]^{f;g}\ar@/_/[drrr]_{\al{has as mother\qquad}}\ar[rrr]^(.35){\al{\textrm{has as parents}}}&&&\fbox{\parbox{.3\textwidth}{a pair $(w,m)$ where $w$ is a woman and $m$ is a man}}\ar[d]^(.6){\al{w}}\\
&&&\fbox{a woman}}}
$$
where we did not indicate the identity aspects associated to each type since they are irrelevant by virtue of the identification conditions (b) and (c) in example~\ref{freecatlds}, and where
$$
f;g=\ttbox{.2}{\al{has as parents a pair $(w,m)$ where $w$ is a woman and $m$ is a man which yields, as the value of $w$,}}
$$ 
then to draw the equational theory $Th(\mathfrak O)$
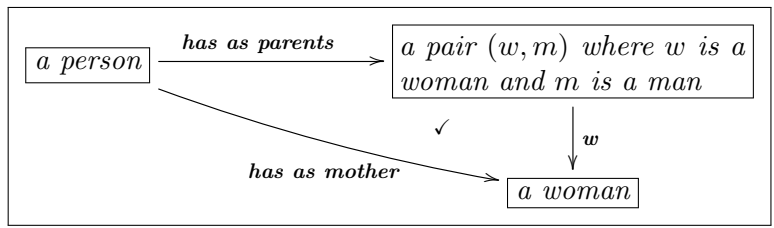
\begin{eqnarray*}
\fbox{\xymatrix{\fbox{a person}\ar@/_/[drrr]_(.6){\al{\textrm{has as mother}\qquad}}\ar[rrr]^(.35){\al{\textrm{has as parents}}}&&&\fbox{\parbox{.3\textwidth}{a pair $(w,m)$ where $w$ is a woman and $m$ is a man}}\ar@{}[dl]|{\checkmark}\ar[d]^(.6){\al{w}}\\
&&&\fbox{a woman}}}
\end{eqnarray*} 
as the resulting olog, where the checkmark is there to communicate the application of the identification $\al{has as mother}=f;g$ imposed by the fact $F\in\mathcal F$ described in example~\ref{nono}.
\end{exm}


\section{A diagrammatic calculus for the syllogistic}\label{calculus}
In this section we recall some basics on syllogistic and describe the main features of the diagrammatic calculus for it that we introduced in~\cite{DBLP:journals/jolli/Pagnan12}, and further investigated in~\cite{DBLP:journals/jolli/Pagnan13}. We will not be exhaustive about syllogistic. The reader who is interested in deepening her knowledge about syllogistic is invited to consult the relevant bibliographic references listed in the previously cited papers, in particular~\cite{Lukasiewicz, Smiley1973-SMIWIA} to begin with. More to the point for what concern the mentioned diagrammatic calculus, in this section we will start introducing it in subsection~\ref{thecalculus} and proceed futher in describing its main features up to subsection~\ref{contradiction} in a rather informal way but, in subsection~\ref{justreq}, for reasons that have to do with the investigation of the possibility of extending the logical expressivity of ologs by imposing syllogistic constraints beyond the equational ones, via the diagrammatic calculus at issue, we will formally treat it within a suitable diagrammatic deductive system that we will describe in detail in subsection~\ref{justreq}.

\subsection{Syllogistic}\label{syllogistic}
Syllogistic was identified by Aristotle as a formal system for logical argumentation in the natural language, probably the first ever known. It is based on the {\bf categorical propositions}
\begin{eqnarray}\label{square}
\begin{tabular}{llllllllll}
${\bf A}_{SP}$:&Every $S$ is $P$&&&&${\bf E}_{SP}$:&Every $S$ is not $P$\\[1ex]
${\bf I}_{SP}$:&Some $S$ is $P$&&&&${\bf O}_{SP}$:&Some $S$ is not $P$
\end{tabular}
\end{eqnarray}   
where the upper case letters $S, P$ are substitutable with meaningful expressions of the natural language, that henceforth will be referred to as {\bf terms}. Thus, ``black dog'' or ``good friend'' are examples of terms and  ``Every good friend is a black dog'' and ``Some black dog is not a good friend'' are examples of categorical propositions; in each categorical proposition, $S$ and $P$ are referred to as {\bf subject} and {\bf predicate}, respectively; ${\bf A}_{SP}$, ${\bf I}_{SP}$, ${\bf E}_{SP}$, ${\bf O}_{SP}$ are the traditional names that the medieval logicians gave to each scheme of categorical proposition, respectively. Specifically, {\bf A} and {\bf I} are taken as vowels from the latin word {\bf A}df{\bf I}rmo, for affirmation, whereas {\bf E} and {\bf O} are taken as vowels from the latin word n{\bf E}g{\bf O}, for negation, so that the categorical propositions which can be bringed back to the schemes ${\bf A}_{SP}$, ${\bf I}_{SP}$ are said to be {\bf universal affirmative}, {\bf particular affirmative} propositions, respectively, whereas those which can be bringed back to the schemes ${\bf E}_{SP}$, ${\bf O}_{SP}$ are said to be {\bf universal negative}, {\bf particular negative} propositions, respectively. As they stand, the categorical propositions occupy the vertices of an ideal square, which is the so-called {\bf square of opposition}. Various interrelations link the categorical propositions in the square but notably, the categorical propositions
that occupy its diagonally opposite vertices are in {\bf contradiction}, since they are the negation of each other and so cannot be jointly affirmed. A {\bf syllogism} in the natural language is a logical argument of the form
\begin{eqnarray}\label{syllosyllo}
P_1,P_2\vdash C
\end{eqnarray}
where $\vdash$ stays for ``therefore'', and $P_1, P_2, C$ are categorical propositions that are distinguished as first premise, second premise and conclusion, respectively. More precisely, for~\eqref{syllosyllo} to be a syllogism, there must be three terms involved in it, typically denoted by $S,M,P$, as follows: $M$ must occur in both the premisses and cannot occur in the conclusion, whereas, in accordance with the traditional way of writing syllogisms, $P$ must occur in the first premise and in the conclusion as the predicate of the latter, whereas $S$ must occur in the second premise and in the conclusion as the subject of the latter. $M$ is usually referred to as {\bf middle term}.

\begin{exm}
The logical argument
\begin{eqnarray}\label{syllo1}
\begin{tabular}{l}
Every mammal is not able to fly,\\
Every donkey is a mammal\\
therefore\\
Every donkey is not able to fly
\end{tabular}
\end{eqnarray}
is a syllogism in which ``mammal'' is $M$, ``donkey'' is $S$ and ``able to fly'' is $P$. Schematically, the previous syllogism can be bringed back to the form ${\bf E}_{MP},{\bf A}_{SM}\vdash{\bf E}_{SP}$.
\end{exm}

\begin{exm}
The logical argument
\begin{eqnarray}\label{syllo2}
\begin{tabular}{l}
Every mammal is able to fly,\\
Some mammal is a donkey\\
therefore\\
Some donkey is able to fly
\end{tabular}
\end{eqnarray}
is a syllogism that can be bringed back to the form ${\bf A}_{MP},{\bf I}_{MS}\vdash{\bf I}_{SP}$.
\end{exm}

\subsection{A diagrammatic calculus for the syllogistic}\label{thecalculus}
The schemes of logical arguments of the form~\eqref{syllosyllo} that are recognizable as syllogisms in accordance with the previously described structure are 256 in total. Among them, Aristotle has been able to classify the 24 syllogisms which are {\bf valid}, namely those whose conclusion logically follows from their premisses by virtue of the {\bf canon of inference} which is the deletion of the middle term. In this sense, the syllogisms~\eqref{syllo1} and~\eqref{syllo2} are both valid but the syllogism
\begin{eqnarray}\label{syllo3}
\begin{tabular}{l}
Every mammal is not able to fly,\\
Some mammal is a donkey\\
therefore\\
Some donkey is able to fly
\end{tabular}
\end{eqnarray}
is not; but we admit that we established these three facts by appealing, up to a certain extent, to our experience about donkeys, mammals and animals that are able to fly.
\begin{rem}\label{perpoi}
We hasten to observe that a logical argument must be recognizable as valid by virtue of its form, exclusively, let alone any reference to the experience of the rational being considering it or to whatsoever notion of truth.  Relevantly, we point out that the syllogisms~\eqref{syllo1} and~\eqref{syllo2} are valid although in each of them one premise is false, assuming truth to coincide with reality, as it seems reasonable for the syllogistic, in accordance with~\cite{MR0010521}. In  fact, for the  syllogism~\eqref{syllo2} it is false that every mammal is able to fly whereas for the syllogism~\eqref{syllo1} it is false that every mammal is not able to fly, since bats are mammals that are able to fly. 
\end{rem}

With regard to syllogisms the question is: is it possible to formally implement the above mentioned canon of inference as an easy to recognize deductive step of calculation?

We answer to this question by describing a diagrammatic deductive calculus for the syllogistic that captures the valid syllogisms exactly, in the sense that on the base of it the valid syllogisms are those whose conclusion follows from their premisses as the result of a diagrammatic step of calculation and just those. This is proved in~\cite{DBLP:journals/jolli/Pagnan12}.

We rephrase the square of opposition~\eqref{square} by considering suitable diagrammatic representations of the corresponding categorical propositions, as {\bf labelled diagrams}, like this
\begin{eqnarray}\label{diagrams}
\begin{tabular}{llllllll}
\UA{S}{P}&&&&\UN{S}{P}\\[2ex]
\PA{S}{P}&&&&\PN{S}{P}
\end{tabular}
\end{eqnarray}
and, for the future, we here request to consider each of the previous diagrams as an indecomposable piece of syntax, as witnessed by the associated label. Consequently, no proper part of a diagram among the previous will be considered as conveying any meaningful piece of information in itself. The rigorous justification of the previous request will be given in subsection~\ref{justreq}. 

\begin{rem}\label{remoremo}
We observe in passing that the diagrams~\eqref{diagrams} are diagrams in the sense of definition~\ref{shapediagr}, of shapes

\begin{eqnarray}\label{shapesforsyll}
\begin{tabular}{llllllll}
\xymatrix{0\ar[r]^{01}&1}&&&&\xymatrix{0\ar[r]^{01}&1&2\ar[l]_{21}}
\\[2ex]
\xymatrix{0&1\ar[l]_{10}\ar[r]^{12}&2}&&&&\xymatrix{0&1\ar[r]^{12}\ar[l]_{10}&2&3\ar[l]_{32}}
\end{tabular}
\end{eqnarray}
respectively, in the graph $\mathcal G=(G_0,G_1,\partial_0,\partial_1)$ identified by the following data:
\begin{eqnarray}\label{S}
G_0=\{S,M,P,\bullet\}\\
G_1=\{\al{sp},\al{sm},\al{mp},\bullet\al s, \bullet\al m,\bullet\al p, \bullet\bullet, \al s\bullet,\al m\bullet,\al p\bullet \}\nonumber\\
\partial_0:\left[\begin{array}{l}
\al{sp}\mapsto S\\
\al{sm}\mapsto S\\
\al{mp}\mapsto S\\
\bullet\al s\mapsto\bullet\\
\bullet\al m\mapsto\bullet\\
\bullet\al p\mapsto\bullet\\
\bullet\bullet\mapsto\bullet\\
\al s\bullet\mapsto S\\
\al m\bullet\mapsto M\\
\al p\bullet\mapsto P
\end{array}\right]
:G_1\rightarrow G_0\nonumber\\[3ex]
\partial_1:\left[\begin{array}{l}
\al{sp}\mapsto P\\
\al{sm}\mapsto M\\
\al{mp}\mapsto P\\
\bullet\al s\mapsto S\\
\bullet\al m\mapsto M\\
\bullet\al p\mapsto P\\
\bullet\bullet\mapsto\bullet\\
\al s\bullet\mapsto\bullet\\
\al m\bullet\mapsto\bullet\\
\al p\bullet\mapsto\bullet
\end{array}\right]
:G_1\rightarrow G_0\nonumber
\end{eqnarray}
in which we have indicated the node $M$ too, in view of the way the diagrams~\eqref{diagrams} will be used below to verify the validity of syllogistic arguments. 

Formally the diagrams~\eqref{diagrams} are identified by the following homomorphisms of graphs:
\begin{itemize}
\item[-] for ${\bf A}_{SP}$, $D_{\bf A}(S,P):[\xymatrix{0\ar[r]^{01}&1}]\rightarrow
\mathcal G$ identified by the following assignments:
\begin{itemize}
\item on the nodes: 
$$
\left[\begin{array}{l}
0\mapsto S\\ 
1\mapsto P
\end{array}\right]
$$
\item on the arcs: 
$$
\left[\begin{array}{l}
01\mapsto\al{sp}
\end{array}\right]
$$
\end{itemize}
\item[-] for ${\bf E}_{SP}$, $D_{\bf E}(S,P):[\xymatrix{0\ar[r]^{01}&1&2\ar[l]_{21}}]\rightarrow\mathcal G$ identified by the following assignments:
\begin{itemize}
\item on the nodes: 
$$
\left[\begin{array}{l}
0\mapsto S\\
1\mapsto \bullet\\
2\mapsto P
\end{array}\right]
$$
\item on the arcs
$$
\left[\begin{array}{l}
01\mapsto \al s \bullet\\
21\mapsto \al p\bullet
\end{array}\right]
$$
\end{itemize}
\item[-] for ${\bf I}_{SP}$, $D_{\bf I}(S,P):[\xymatrix{0&1\ar[l]_{10}\ar[r]^{12}&2}]\rightarrow\mathcal G$ identified by the following assignments:
\begin{itemize}
\item on the nodes: 
$$
\left[\begin{array}{l}
0\mapsto S\\
1\mapsto \bullet\\
2\mapsto P
\end{array}\right]
$$
\item on the arcs
$$
\left[\begin{array}{l}
10\mapsto \bullet\al s\\
12\mapsto \bullet\al p
\end{array}\right]
$$
\end{itemize}
\item[-] for ${\bf O}_{SP}$, $D_{\bf O}(S,P):[\xymatrix{0&1\ar[r]^{12}\ar[l]_{10}&2&3\ar[l]_{32}}]\rightarrow\mathcal G$ identified by the following assignments:
\begin{itemize}
\item on the nodes:
$$
\left[\begin{array}{l}
0\mapsto S\\
1\mapsto \bullet\\
2\mapsto \bullet\\
3\mapsto P
\end{array}\right]
$$
\item on the arcs
$$
\left[\begin{array}{l}
10\mapsto \bullet\al s\\
12\mapsto \bullet\bullet\\
32\mapsto \al p\bullet
\end{array}\right]
$$
\end{itemize}
\end{itemize}
so that it is possible to understand from where the labels of diagrams~\eqref{diagrams} come: they are just abbreviations of the names of the homomorphisms of graphs that identify those diagrams. As an exercise the reader could explicitly describe the remaining diagrams $D_{\bf A}(S, M)$, $D_{\bf A}(M, P)$, $D_{\bf E}(S, M)$, $D_{\bf E}(M, S)$, \ldots
\end{rem}

Now we describe how the diagrams~\eqref{diagrams} can be employed to verify the validity of a given syllogism, in accordance with the following algorithm:
\begin{enumerate}
\item Recognize the terms $S,M,P$ in the given syllogism.
\item Represent the premisses and the conclusion of the syllogism by means of the corresponding diagrams.
\item\label{2} Superpose the diagrams for the premisses on the middle term $M$.
\item If, after~\eqref{2}, $M$ is between two discordantly oriented arrows, such as in the sequence $\rightarrow M\leftarrow$ or in the sequence $\leftarrow M\rightarrow$, then the calculation halts and the syllogism under consideration is not valid.
\item\label{4} If, after~\eqref{2}, $M$ is between two concordantly oriented arrows, such as in the sequence $\rightarrow M\rightarrow$ or in the sequence $\leftarrow M\leftarrow$, then the calculation proceeds futher by substituting the whole sequence with one concordantly oriented arrow, namely with $\rightarrow$ or $\leftarrow$ respectively, thus deleting the middle term $M$. 
\item If, after~\eqref{4}, the obtained diagram is not among the diagrams~\eqref{diagrams}, the calculation halts and the syllogism under consideration is not valid.
\item If, after~\eqref{4}, the obtained diagram is among the diagrams~\eqref{diagrams} but does not coincide with the one for the conclusion, then the calculation halts and the syllogism under consideration is not valid.
\item If, after~\eqref{4}, the obtained diagram is among the diagrams~\eqref{diagrams} and coincides with the one for the conclusion, then the calculation halts and the syllogism under consideration is valid.
\end{enumerate}

The procedure of calculation we just described can be seen to be very easily implementable by having a look at some examples of its employment. For instance, we can formally prove that the syllogisms~\eqref{syllo1} and~\eqref{syllo2} are valid and that~\eqref{syllo3} is not. We repropose them below together with the calculation of their validity beside, like this: 
\begin{itemize}
\item[-] for~\eqref{syllo1},
\begin{eqnarray*}
\begin{tabular}{llllllll}
Every mammal is not able to fly,&&&&\UN{M}{P}&&by (1),(2)\\
Every donkey is a mammal,&&&&\UA{S}{M}&&by (1),(2)\\
therefore&&&&$\xymatrix{S\ar[r]^{{\bf A}_{SM}}&M\ar@{}[rr]^{{\bf E}_{MP}}\ar[r]&\bullet&P\ar[l]}$
&&by (3)\\
Every donkey is not able to fly&&&&\UN{S}{P}&&by (1),(2),(5),(8)
\end{tabular}
\end{eqnarray*}
\item[-] for~\eqref{syllo2},
\begin{eqnarray*}
\begin{tabular}{llllllll}
Every mammal is able to fly,&&&&\UA{M}{P}&&by (1),(2)\\
Some mammal is a donkey,&&&&\PA{M}{S}&&by (1),(2)\\
therefore&&&&$\xymatrix{S\ar@{}[rr]^{{\bf I}_{SM}}&\bullet\ar[r]\ar[l]&M\ar[r]^{{\bf A}_{MP}}&P}$
&&by (3)\\
Some donkey is able to fly&&&&\PA{S}{P}&&by (1),(2),(5),(8)
\end{tabular}
\end{eqnarray*}
\item[-] for~\eqref{syllo3},
\begin{eqnarray*}
\begin{tabular}{llllll}
Every mammal is not able to fly,&&&\UN{M}{P}&&by (1),(2)\\
Some mammal is a donkey,&&&\PA{M}{S}&&by (1),(2)\\
therefore&&&$\xymatrix{S\ar@{}[rr]^{{\bf I}_{SM}}&\bullet\ar[l]\ar[r]&M\ar@{}[rr]^{{\bf E}_{MP}}\ar[r]&\bullet&P\ar[l]}$&&by (3)\\
Some donkey is able to fly&&&\PN{S}{P}&&by (5) but,\\ 
\end{tabular}
\end{eqnarray*}
\end{itemize}
by (7), the syllogism is not valid, since the diagram
for the conclusion is
\PA{S}{P}.

\begin{rem}\label{formalcomp}
The diagrammatic calculus at issue is even more rapidly implementable by making the crucial deductive step of calculation in (5) amount to a formal operation of composition. For instance, in order to conclude that the syllogism~\eqref{syllo1} is valid it suffices to consider the drawing
\begin{eqnarray}\label{primodr}
\xymatrix{S\ar[r]^{{\bf A}_{SM}}\ar@/_1pc/@{-->}[rr]&M\ar@{}[rr]^{{\bf E}_{MP}}\ar[r]&\bullet&P\ar[l]}
\vspace{2ex}
\end{eqnarray}
from which it is possible to read off the diagram
$\PA{S}{P}$ for the seeked conclusion. 
In order to conclude that the syllogism~\eqref{syllo3} is not valid it suffices to consider the drawing
\begin{eqnarray}\label{secondodr}
\xymatrix{S\ar@{}[rr]^{{\bf I}_{SM}}&\bullet\ar@/_1pc/@{-->}[rr]\ar[r]\ar[l]&M\ar@{}[rr]^{{\bf E}_{MP}}\ar[r]&\bullet&P\ar[l]}
\end{eqnarray}
from which it is possible to read off the diagram $\PN{S}{P}$ which is not the one for the indicated conclusion.
\end{rem}

\begin{notation}\label{forsempre}
The dashed arrows employed in the drawings~\eqref{primodr},~\eqref{secondodr} result by virtue of the implementation of a step of logical calculation. Because of this they have been put in evidence by drawing them as dashed. This is a convenient notation that we will employ throughout this section and later on in the paper.
\end{notation}

Theorem~\ref{hauptsatz} below is the main result concerning the diagrammatic calculus under discussion. As already mentioned it is proved in detail in~\cite{DBLP:journals/jolli/Pagnan12}. 

\begin{thm}\label{hauptsatz}
A syllogism if valid if and only if its conclusion follows from its premisses as the result of the diagrammatic procedure of logical calculation that we previously described in steps from (1) to (8).
\end{thm}

\subsection{A rejection criterion}\label{rejection}
In implementing the diagrammatic procedure of calculation that we described in subsection~\ref{thecalculus} no symbol $\bullet$ is created or deleted; so a necessary condition for the validity of a syllogism is the following: in the diagram for the conclusion of a valid syllogism there must be as many symbols $\bullet$ as they are in total in the diagrams for its premisses. On the base of this, it is immediate to establish that the syllogism~\eqref{syllo3} is not valid, since two symbols $\bullet$ are in the diagrams for its premisses whereas just one is in the diagram for its conclusion. On the other hand, it is immediate to verify that the syllogism
\begin{eqnarray*}
\begin{tabular}{l}
Every mammal is not able to fly,\\
Every donkey is not a mammal\\
therefore\\
Some donkey is not able to fly
\end{tabular}
\end{eqnarray*}
is not valid, although two symbols $\bullet$ are in the diagrams for its premisses as well as in the diagram for its conclusion; thus showing that the condition at issue is not sufficient for the validity of a syllogism. The necessary condition that we have just described is a {\bf rejection criterion} intrinsic to the calculus; thanks to it, the invalid 232 syllogisms can be easily rejected; 

\subsection{The managing of existential import}\label{eximport}
With reference to the categorical propositions in the square of opposition it is observed that propositions of the form ${\bf I}_{SP}$ logically follow from propositions of the form ${\bf A}_{SP}$ and that propositions of the form ${\bf O}_{SP}$ logically follow from propositions of the form ${\bf E}_{SP}$. This is correct only if {\bf existential import} is assumed; otherwise it is possible to exhibit suitable counterexamples. Existential import is the explicit assumption of existence of some individual in the extension of a term. To clarify matters, let us consider the first order encoding of syllogistic. Syllogistic can be encoded in a first order logical calculus by means of three unary predicates, that is $S(x)$, $M(x)$ and $P(x)$. The categorical propositions in the square of opposition take the form of corresponding well formed formulas, like this:
\begin{eqnarray*}
\begin{tabular}{llllllllll}
${\bf A}_{SP}$:&$\forall x.(S(x)\Rightarrow P(x))$ &&&&${\bf E}_{SP}$:&$\forall x.(S(x)\Rightarrow\neg P(x))$\\[1ex]
${\bf I}_{SP}$:&$\exists x.(S(x)\wedge P(x))$&&&&${\bf O}_{SP}$:&$\exists x.(S(x)\wedge\neg P(x))$
\end{tabular}
\end{eqnarray*}   

Now, on one hand it is possible to prove
\begin{eqnarray}\label{AI}
\exists x.S(x), \forall x.(S(x)\Rightarrow P(x))\vdash\exists x.(S(x)\wedge P(x))
\end{eqnarray}
and 
\begin{eqnarray}\label{EO}
\exists x.S(x), \forall x.(S(x)\Rightarrow \neg P(x))\vdash\exists x.(S(x)\wedge \neg P(x))
\end{eqnarray}
by virtue of the condition of existential import which is $\exists x.S(x)$; on the other hand it is possible to produce counterexamples to 
\begin{eqnarray}\label{noAI}
\forall x.(S(x)\Rightarrow P(x))\vdash\exists x.(S(x)\wedge P(x))
\end{eqnarray}
and
\begin{eqnarray}\label{noEO}
\forall x.(S(x)\Rightarrow \neg P(x))\vdash\exists x.(S(x)\wedge \neg P(x))
\end{eqnarray}

As a consequence of the rejection criterion discussed in subsection~\ref{rejection}, in the diagrammatic calculus that we are describing the logical deductions corresponding to~\eqref{noAI} and~\eqref{noEO} cannot be implemented. On the other hand, the logical deductions corresponding to~\eqref{AI} and~\eqref{EO} can be implemented provided a suitable diagrammatic condition of existential import is taken into account, that is 
$$
\PA{S}{S}
$$
to be read as ``Some $S$ is $S$'' to mean ``Some $S$ exists'', by virtue of which it is possible to calculate
\begin{itemize}
\item[-] for~\eqref{AI}:
$$
\xymatrix{S\ar@{}[rr]^{{\bf I}_{SS}}&\bullet\ar@/_1pc/@{-->}[rr]\ar[l]\ar[r]&S\ar[r]^{{\bf A}_{SP}}&P}
$$
that is ${\bf I}_{SS}, {\bf A}_{SP}\vdash{\bf I}_{SP}$;
\item[-] for~\eqref{EO}:
$$
\xymatrix{S\ar@{}[rr]^{{\bf I}_{SS}}&\bullet\ar@/_1pc/@{-->}[rr]\ar[l]\ar[r]&S\ar@{}[rr]^{{\bf E}_{SP}}\ar[r]&\bullet& P\ar[l]}
$$
that is ${\bf I}_{SS},{\bf E}_{SP}\vdash{\bf O}_{SP}$.
\end{itemize}
Among the 24 valid syllogisms identified by Aristotle, those which are valid under existential import are 9. They are precisely those with premisses which are both universal and conclusion which is particular. 

\subsection{Contradiction via diagrams}\label{contradiction}
The diagrams~\eqref{diagrams} have been introduced in place of corresponding categorical propositions toward the obtainment of a ``diagrammatic square of opposition''. It is possibile to verify that all of the interrelations that we did not mention, linking the categorical propositions in the original square of opposition~\eqref{square} hold correspondingly in its diagrammatic counterpart~\eqref{diagrams}; see~\cite{DBLP:journals/jolli/Pagnan12} for details, but also consider that what we actually did in subsection~\ref{eximport} was showing that the so-called laws of subalternation are diagrammatically deducible, for instance. Now, {\bf diagrammatic contradiction} is expressed by a diagram such as 
\PN{X}{X} that has to be contradictorially read as ``Some $X$ is not $X$''. 
By means of a suitable diagrammatic calculation, one can show that the diagrams that occupy the diagonally opposite vertices in the diagrammatic square~\eqref{diagrams} are in contradiction:
\begin{itemize}
\item[-] for \UA{S}{P} and \PN{S}{P}:
$$
\xymatrix{S\ar@{}[rrr]^{{\bf O}_{SP}}&\bullet\ar[l]\ar[r]&\bullet&P\ar[l]&S\ar[l]_{{\bf A}_{SP}}\ar@/^1pc/@{-->}[ll]}
$$
from which it is possible to read off the diagram \PN{S}{S};
\item[-] for \PA{S}{P} and \UN{S}{P}:
$$
\xymatrix{S\ar@{}[rr]^{{\bf I}_{SP}}&\bullet\ar@/_1pc/@{-->}[rr]\ar[l]\ar[r]&P\ar@{}[rr]^{{\bf E}_{SP}}\ar[r]&\bullet&S\ar[l]}
$$
from which it is possible to read off the diagram \PN{S}{S}.
\end{itemize}
To further clarify matters, we observe that in the original square of opposition~\eqref{square} the categorical propositions of the form ${\bf A}_{SP}$, ${\bf E}_{SP}$ are {\bf contrary} to each other. A priori, categorical propositions of that form can be jointly affirmed without producing a contradiction, since to assume that they hold together leads to conclude that the extension of $S$ must be empty. In first order logic, this is given by the provability of the logical consequence 
$$
\forall x.(S(x)\Rightarrow P(x)), \forall x.(S(x)\Rightarrow\neg P(x))\vdash\forall x.(S(x)\Rightarrow\neg S(x))
$$
that is diagrammatically very quickly expressed by the drawing 
$$
\xymatrix{S\ar@/_1pc/@{-->}[rr]\ar[r]^{{\bf A}_{SP}}&P\ar@{}[rr]^{{\bf E}_{SP}}\ar[r]&\bullet&S\ar[l]}
$$
from which it is possibile to read off the diagram \UN{S}{S} for the categorical proposition ``Every $S$ is not $S$'' or, informally, for ``$S$ is empty''. On the other hand, if the existence of some individual in the extension of the term $S$ is assumed, via the condition of existential import \PA{S}{S}, see subsection~\ref{eximport}, then from it and the previously deduced diagram for 
${\bf E}_{SS}$, the contradictory diagram for ${\bf O}_{SS}$ is diagrammatically deducible.

\subsection{Formalizing the calculus}\label{justreq}
Just after the diagrammatic square of opposition~\eqref{diagrams} we requested to consider each of the diagrams that form it as an
indecomposable piece of syntax. In this subsection we rigorously justify this request by formalizing the employment that we made so far of the diagrams~\eqref{diagrams}, 
within a suitable {\bf diagrammatic deductive system} that we name SYLL. Most of the content of this subsection is from~\cite{DBLP:journals/jolli/Pagnan12}, but appropriately reworked.

\begin{defi}\label{defi1}
The {\bf syntactic primitives} of SYLL are the symbols $\bullet$, $\rightarrow$, $\leftarrow$ toghether with a denumerable amount of term-variables $A, B, C,\ldots, M, P, S,\ldots$. The {\bf syllogistic diagrams} of SYLL are the labelled diagrams in the diagrammatic square~\ref{diagrams}. A {\bf diagram} of SYLL is a finite list of arrow symbols separated by a single bullet symbol $\bullet$ or a single term-variable, beginning and ending at a term-variable. The {\bf reversal} of a diagram is a diagram, the one obtained by specular symmetry from the original diagram. A {\bf part} of a diagram is a finite list of consecutive components of a diagram.
\end{defi}

\begin{rem}\label{revequiv}
The reversal of the syllogistic diagrams~\eqref{diagrams} are
\begin{eqnarray}\label{revdiagrams}
\begin{tabular}{llllllll}
\revUA{S}{P}&&&&\revUN{S}{P}\\[2ex]
\revPA{S}{P}&&&&\revPN{S}{P}
\end{tabular}
\end{eqnarray}
to clarify that because of their asymmetry ${\bf A}_{SP}$ and ${\bf O}_{SP}$ are not equivalent to their respective reversals, that one could be tempted to identify with ${\bf A}_{PS}$ and ${\bf O}_{PS}$; whereas because of their symmetry ${\bf I}_{SP}$ and ${\bf E}_{SP}$ are equivalent to their respective reversals ${\bf I}_{PS}$, ${\bf E}_{PS}$, in accordance with the fact that the formulas
$\exists x.(S(x)\wedge P(x))$, $\exists x.(P(x)\wedge S(x))$ are equivalent, and the formulas $\forall x.(S(x)\Rightarrow\neg P(x))$, $\forall x.(P(x)\Rightarrow\neg S(x))$
are equivalent.
\end{rem}

\begin{notation}
Whenever there will be the need to anonimously refer to non-empty parts of diagrams, they will be henceforth denoted by a symbol $*$. In order to explicitly distinguish a part with respect to a whole diagram we adopt a sort of heterogeneous notation mixing $*$ symbols and syntactic primitives. For instance, the writing $*\rightarrow A$ refers to a diagram in which the part $\rightarrow A$ has been distinguished with respect to the remaining part $*$. So, it may be the case that the whole diagram looks like $S\rightarrow A$ or $B\rightarrow\bullet\rightarrow A$, for instance, so that $*$ would refer to $S$ or $B\rightarrow\bullet$, respectively. The occurrence of more than one symbol $*$ may refer to different parts of a whole diagram. For instance, the writing $*A\leftarrow*$ makes explicit the part $A\leftarrow$ of a diagram that in its entirety could look like $A\leftarrow\bullet\leftarrow A\leftarrow C$, where the part at issue is the one on the right of the part $\bullet\leftarrow$ not the one on the left of it, that is $A\leftarrow *$, because of the symbol $*$ on the left of $A$ in the writing $*A\leftarrow*$, denoting a non-empty part.
\end{notation}

\begin{defi}
A {\bf superposable pair} of diagrams is a pair of diagrams $(*A, A*)$ or $(A*, *A)$ whose {\bf superposition} is, in both cases, the diagram $*A*$, which is obtained by superposing the components of the pair on the common extremal term-variable $A$. A {\bf composable pair} of diagrams is a superposable pair of diagrams $(*\rightarrow A, A\rightarrow *)$ or $(A\rightarrow *,*\rightarrow A)$, $(*\leftarrow A, A\leftarrow *)$ or $(A\leftarrow *, *\leftarrow A)$. In the first two cases, a composable pair gives rise to a {\bf composite diagram} $*\rightarrow *$ obtained by substituting the part $\rightarrow A\rightarrow $ in the superposition $*\rightarrow A\rightarrow *$ with the sole, accordingly oriented, arrow symbol $\rightarrow$. Analogously, in the second two cases a composable pair gives rise to a composite diagram $*\leftarrow *$. 
\end{defi}

\begin{defi}
A {\bf well formed diagram} of SYLL is inductively defined as follows:
\begin{enumerate}
\item a syllogistic diagram is a well formed diagram;
\item the reversal of a syllogistic diagram is a well formed diagram;
\item a diagram which is the superposition of a superposable pair of well formed diagrams is a well formed diagram;
\item nothing else is a well formed diagram.
\end{enumerate}
\end{defi}

\begin{defi}
The {\bf rules of inference} of SYLL are
\begin{eqnarray*}
\AxiomC{}
\UnaryInfC{\PA{A}{A}}
\DisplayProof
\end{eqnarray*}

\begin{eqnarray*}
\AxiomC{\UA{S}{P}}
\doubleLine
\UnaryInfC{\revUA{S}{P}}
\DisplayProof
\qquad&\qquad
\AxiomC{\UN{S}{P}}
\doubleLine
\UnaryInfC{\revUN{S}{P}}
\DisplayProof
\\
\\
\\
\AxiomC{\PA{S}{P}}
\doubleLine
\UnaryInfC{\revPA{S}{P}}
\DisplayProof
\qquad&\qquad
\AxiomC{\PN{S}{P}}
\doubleLine
\UnaryInfC{\revPN{S}{P}}
\DisplayProof
\\
\\
\\
\AxiomC{$*A$}
\AxiomC{$A*$}
\BinaryInfC{$*A*$}
\DisplayProof
\qquad&\qquad
\AxiomC{$A*$}
\AxiomC{$*A$}
\BinaryInfC{$*A*$}
\DisplayProof
\\
\\
\\
\AxiomC{$*\rightarrow A\rightarrow *$}
\UnaryInfC{$*\rightarrow*$}
\DisplayProof
\qquad&\qquad
\AxiomC{$*\leftarrow A\leftarrow *$}
\UnaryInfC{$*\leftarrow*$}
\DisplayProof
\end{eqnarray*}
A {\bf proof tree} of SYLL is a tree where each node is a well formed diagram and each branching is an instance of a rule of inference. 
In SYLL, a {\bf formal proof} of a syllogism is a proof tree with root which is the conclusion of the syllogism and leaves which are the premisses of the syllogism. In SYLL, a syllogism is {\bf provable} if there is a formal proof for it.
\end{defi}

\begin{exm}
The syllogism ${\bf A}_{PM},{\bf E}_{SM}\vdash{\bf E}_{SP}$ is provable in SYLL, because the proof tree
$$
\AxiomC{\UN{S}{M}}
\AxiomC{\UA{P}{M}}
\UnaryInfC{\revUA{P}{M}}
\BinaryInfC{$\xymatrix{S\ar@{}[rr]^{{\bf E}_{SM}}\ar[r]&\bullet&M\ar[l]&P\ar[l]_{{\bf A}_{PM}}}$}
\UnaryInfC{\UN{S}{P}}
\DisplayProof
$$
is a formal proof for it.
\end{exm}

Finally, in view of what has been discussed in this subsection we observe that theorem~\ref{hauptsatz} can be rephrased as
\begin{thm}
A syllogism is valid if and only if it is provable in SYLL
\end{thm} 
\begin{proof}
See~\cite{DBLP:journals/jolli/Pagnan12}.
\end{proof}


\section{Ologisms}\label{ologisms}
In subsection~\ref{facts} we discussed facts in ologs and pointed out that they provide the possibility of declaring equational constraints between pairs of parallel aspects; in subsection~\ref{ologwitheq}, that has been shown to correspond to the imposition of suitable identification conditions between certain formal proofs. The aim of this section is that of showing how to extend the logical expressivity of ologs by moreover making available the possibility of declaring constraints of syllogistic nature in them. Here the approach is informal. We will be mathematically rigorous in section~\ref{extlabelded}. Toward the pursueing of the previously hinted at aim, in this section the idea is that of giving an intuition of how it is possible to take advantage of the nature and peculiar features of the diagrammatic calculus for the syllogistic that we discussed in section~\ref{calculus}; firstly, its compositional nature, see remark~\ref{formalcomp}, which makes it immediately compatible with the natural arising of composite aspects in ologs. The structured models for knowledge representation that appear because of the extension of ologs with syllogistics will be referred to as {\bf ologisms}, as a name which is a crasis between the words ``olog'' and ``syllogism''. Ologisms extend both ologs and the calculus described in section~\ref{calculus}. 
 
\subsection{Categorical propositions in extended ologs}\label{recat}
To begin the modelling of syllogistic reasoning, we here describe how to read the labelled diagrams~\eqref{diagrams} in a suitably extended olog. More to the point, we will be able to identify those diagrams in a suitable extension $\mathcal O^{\bullet}$ of the underlying graph of a fixed olog $\mathfrak O=(\mathcal O,\mathcal F)$; that is, properly as diagrams of shapes~\eqref{shapesforsyll} in $\mathcal O^{\bullet}$, which, throughout this section, will be the graph obtained from $\mathcal O$ by the addition of a distinguished node $\bullet$ to be not considered as a type and, possibly, of arcs having $\bullet$ as source or target, toward a type or $\bullet$ in the first instance, from a type or $\bullet$ in the second instance. We hasten to say that this added arcs are not intended to be considered as aspects.
See example~\ref{below} for a clarification. 

\begin{rem}
They obtainement of $\mathcal O^{\bullet}$ out of $\mathcal O$ may be somewhat obscure at this point, but the reader should consider that while writing this section we are projected toward the actual authoring of ologisms. In view of this, consider that extending an olog $(\mathcal O,\mathcal F)$ at an ologism requires the extension of $\mathcal O$ at $\mathcal O^{\bullet}$, firstly, in a way that depends on the syllogistic constraints that the author evaluates to impose; thus on a certain amount of subjectivity, as already observed for ologs, in remark~\ref{perdopo}. 
\end{rem}

\begin{exm}\label{below}
Consider the olog $\mathfrak O=(\mathcal O,\mathcal F)$ described in example~\ref{nono}. An ologism extending $\mathfrak O$ could be obtained by imposing the syllogistic constraint provided by the particular affirmative categorical proposition ``Some person is a woman'' by the addition of the syllogistic diagram
$$
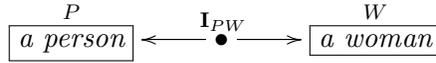
\xymatrix{\stackrel{P}{\fbox{a person}}\ar@{}[rr]^{{\bf I}_{PW}}&\bullet\ar[r]\ar[l]&\stackrel{W}{\fbox{a woman}}}
$$
to the drawing~\eqref{has-mother}, toward the obtainment of the drawing
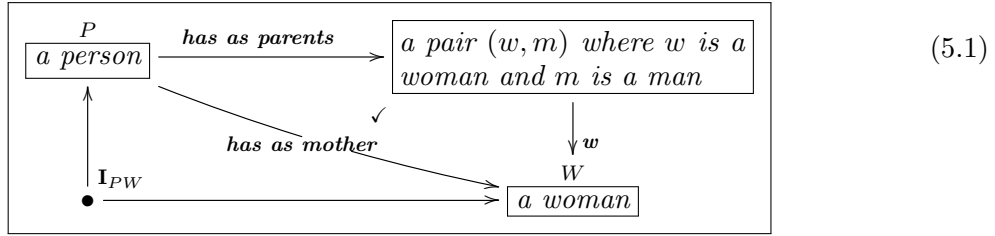
\begin{eqnarray}\label{motherologism}
\fbox{\xymatrix{\stackrel{P}{\fbox{a person}}\ar@/_/[drrr]|(.5){\al{\textrm{has as mother}\qquad}}\ar[rrr]^(.35){\al{has as parents}}&&&\fbox{\parbox{.3\textwidth}{a pair $(w,m)$ where $w$ is a woman and $m$ is a man}}\ar[d]^(.6){\al{w}}\\
\bullet\ar@{}[urrr]|(.6){\checkmark}\ar@{}[urr]|(.15){{\bf I}_{PW}}\ar[rrr]\ar[u]&&&\stackrel{W}{\fbox{a woman}}}}
\end{eqnarray}
or, more formally, by extending the data that identify $\mathcal O$, like this:
\begin{itemize}
\item[-] For $\mathcal O^{\bullet}$:
\begin{itemize}
\item the set of nodes of $\mathcal O^{\bullet}$ is
$$
\{\stackrel{P}{\ulcorner\textrm{a person}\urcorner}, \fbox{\parbox{.3\textwidth}{a pair $(w,m)$ where $w$ is a woman and $m$ is a man}},\stackrel{W}{\ulcorner\textrm{a woman}\urcorner},\bullet\}
$$
\item the set of arcs of $\mathcal O^{\bullet}$ is
$$
\{\al{\textrm{has as parents}}, \al{w},\al{has as mother}, \bullet P,\bullet W\}
$$
where $\bullet P,\bullet W$ are the labels of exactly one arc from $\bullet$ to $\ulcorner\textrm{a person}\urcorner$, and exactly one arc from $\bullet$ to $\ulcorner\textrm{a woman}\urcorner$, respectively. See remark~\ref{nnaa} below. 

\item the source and target functions are
\begin{eqnarray*}
\partial_0:\left[\begin{array}{l}
\al{\textrm{has as parents}}\mapsto\stackrel{P}{\ulcorner\textrm{a person}\urcorner}\\
\al{\textrm{has as mother}}\mapsto\stackrel{P}{\ulcorner\textrm{a person}\urcorner}\\
\al{w}\mapsto\fbox{\parbox{.14\textwidth}{a pair $(w,m)$ where $w$ is a woman and $m$ is a man}}\\
\bullet P\mapsto\bullet\\
\bullet W\mapsto\bullet
\end{array}\right]
\\
\partial_1:\left[\begin{array}{l}
\al{has as parents}\mapsto\fbox{\parbox{.14\textwidth}{a pair $(w,m)$ where $w$ is a woman and $m$ is a man}}\\
\al{has as mother}\mapsto W\\
\al{w}\mapsto W\\
\bullet\ulcorner\textrm{a person}\urcorner\mapsto\ulcorner\textrm{a person}\urcorner\\
\bullet\ulcorner\textrm{a woman}\urcorner\mapsto\ulcorner\textrm{a woman}\urcorner
\end{array}\right]
\end{eqnarray*}
\end{itemize}
\end{itemize}
\end{exm}

\begin{rem}\label{nnaa}
In the previous example~\ref{below}, the arcs $\bullet P$, $\bullet W$ have been left unlabelled in the drawing~\eqref{motherologism} since they are not intended to be aspects. 
\end{rem}

\begin{notation}\label{AEIO}
In subsection~\ref{prescriptions} below we will proceed by prescribing how to recognize and read each form of categorical proposition in a fixed extended olog $\mathfrak O^{\bullet}=(\mathcal O^{\bullet},\mathcal F)$. To this end it is convenient to explicitly describe the structure of the graphs which identify the shapes~\eqref{shapesforsyll}, correspondingly, that is {\bf A}, {\bf E}, {\bf I}, {\bf O}. More to the point, we put
\begin{eqnarray*}
{\bf A}=(\{0,1\},\{01\},\{(01,0)\},\{(01,1)\})\\
{\bf E}=(\{0,1,2\},\{01,21\},\{(01,0),(21,2)\},\{(01,1),(21,1)\})\\
{\bf I}=(\{0,1,2\},\{10,12\},\{(10,1),(12,1)\},\{(10,0),(12,2)\})\\
{\bf O}=(\{0,1,2,3\},\{10,12,32\},\{(10,1),(12,1),(32,3)\},\{(10,0),(12,1),(32,2)\})
\end{eqnarray*}
\end{notation}

\subsection{Prescriptions}\label{prescriptions}

\begin{enumerate}
\item For $A,B$ any types of $\mathfrak O^{\bullet}=(\mathcal O^{\bullet},\mathcal F)$, a universal affirmative categorical proposition ${\bf A}_{AB}$ is interpreted as the diagram 
\begin{eqnarray}\label{A}
D_{\bf A}(A,B):{\bf A}\rightarrow\mathcal O^{\bullet}
\end{eqnarray} 
identified by the assignments $0\mapsto A$, $1\mapsto B$ on nodes, and $01\mapsto\al{is}$, for $\xymatrix{A\ar[r]^{\al{is}}&B}$ an aspect of the olog $\mathfrak O$, if present. That is, ${\bf A}_{AB}$ is interpreted as an aspect $\xymatrix{A\ar[r]^{\al{is}}&B}$.
In general, as already observed, an aspect like $\xymatrix{A\ar[r]^f&B}$ declares that every entity of type $A$ is measured by an entity of type $B$ via a specified functional relationship $f$ or, in other words, that the units of measurement which are the entities of type $B$ apply to all the entities of type $A$, in a way that is specified by $f$. We point out that the reading of ``Every $A$ is $B$'' made by Aristotle was something like ``$B$ applies to all of $A$'s'', that is what an aspect $\xymatrix{A\ar[r]^{\al{is}}&B}$ says in terms of measuring. Maybe, the label \al{is} determines the simplest functional relationship and the one which is immediately adherent to the usual reading of the proposition ${\bf A}_{AB}$ as ``Every $A$ is $B$'',
indicating that something relevant to the instances of type $A$ is forgot to subsume them in a class of instances of type $B$, which are at least as general as those of type $A$, if not even genuinely more general, see also remark~\ref{isa}, by means of an aspect $\xymatrix{A\ar[r]^{\al{is}}&B}$.
In view of all the above, we prescribe that an aspect such as $\xymatrix{A\ar[r]^{\al{is}}&B}$, as a syllogistic constraint, has to be read by pronouncing the indefinite adjective ``Every'' followed by the reading of the label of the type $A$ omitting the reading of the indefinite article which begins it, followed by ``is'', then followed by the label of the type $B$. Thus, for instance, as a syllogistic constraint the aspect~\eqref{isaperson}, has to be read as ``Every woman is a person''.
To summarize, the aspects of an olog identified by the label \al{is} interpret the syllogistic constraints which are imposed by means of universal affirmative categorical propositions.

\item For $A,B$ any types of $\mathfrak O^{\bullet}=(\mathcal O^{\bullet},\mathcal F)$, a particular affirmative categorical proposition ${\bf I}_{AB}$ is interpreted as the diagram
\begin{eqnarray}\label{I}
D_{\bf I}(A,B):{\bf I}\rightarrow\mathcal O^{\bullet}
\end{eqnarray}
identified by the assignments $0\mapsto A$, $1\mapsto\bullet$, $2\mapsto B$ on nodes, and $10\mapsto \bullet A$, $12\mapsto \bullet B$ on arcs, where $\bullet A$ and $\bullet B$ are the labels of whichever, but fixed by the assignments, arcs with source $\bullet$ and target $A$ and source $\bullet$ and target $B$ of $\mathcal O^{\bullet}$, respectively, if present. We hasten to specify that these arcs are not aspects and consequently their labels are not text expressions that are supposed to be read, in any sense. For this reason, if present, 
they will be unanonimously represented in any graphical representation of an ologism. See example~\ref{below}. To help the intuition, we observe in passing that a syllogistic constraint
\begin{eqnarray}\label{i}
\PA{A}{B}
\end{eqnarray}
could be intuitively thought of as the declaration that there is a momentarily unspecified type $\bullet$ whose instances can be measured either by means of the instances of $A$ or by means of the instances of $B$, via momentarily unspecified aspects $\bullet\to A$, $\bullet\to B$ that turn out to be someway comparable since they are common to the type $\bullet$, although not specified. In other words, 
by virtue of the syllogistic constraint ${\bf I}_{AB}$ the entities of type $A$ and those of type $B$ are thought of to be at a certain extent comparable, although in a momentarily unspecified way and measure. Because of remark~\ref{perdopo}, the kind of declaration at issue could be interpreted as a prescriptive one, to be accepted in the world-view proposed by the author of the ologism. In view of all the above, we prescribe that a syllogistic constraint such as~\eqref{i} has to be read by pronouncing the indefinite adjective ``Some'' followed by the reading of the label of the type $A$ omitting the reading of the indefinite article which begins it, followed by ``is'', followed by the reading of the label of the type $B$. Thus, for instance, the syllogistic constraint ${\bf I}_{AB}$ in the ologism~\eqref{motherologism} has to be read as ``Some person is a woman'' or, equivalently as ``Some woman is a person''.

\item For $A, B$ any types of the extended olog $\mathfrak O^{\bullet}=(\mathcal O^{\bullet},\mathcal F)$, a universal negative categorical proposition ${\bf E}_{AB}$ is interpreted as the diagram
$$
D_{\bf E}(A,B):{\bf E}\rightarrow\mathcal O^{\bullet}
$$
identified by the assignments $0\mapsto A$, $1\mapsto\bullet$, $2\mapsto B$ on nodes, and $01\mapsto A\bullet$, $21\mapsto B\bullet$ on arcs, where $A\bullet$ and $B\bullet$ are the labels of whichever, but fixed by the assignments, 
arcs with source $A$ and target $\bullet$ and source $B$ and target $\bullet$ of $\mathcal O^{\bullet}$, respectively, if present. We hasten to specify that these arcs are not aspects and consequently their labels are not text expressions that are supposed to be read, in any sense. For this reason, if present, they will be unanonimously represented in any graphical representation of an ologism.
To help the intuition, we observe in passing that a syllogistic constraint 
\begin{eqnarray}\label{ee}
\UN{A}{B}
\end{eqnarray}
could be intuitively thought of as the declaration that there are momentarily unspecified aspects $A\to\bullet$, $B\to\bullet$ which make the instances of type $A, B$, respectively, someway not comparable in a measure whatsoever. In other words, the instances of $A, B$, become uncomparable in $\bullet$, via unspecified aspects $A\to\bullet$, $B\to\bullet$. Because of remark~\ref{perdopo}, the kind of declaration at issue could be interpreted as a prescriptive one, to be accepted in the world-view proposed by the author of the ologism. In view of all the above, we prescribe that a syllogistic constraint such as~\eqref{ee} has to be read by pronouncing the indefinite adjective ``Every'' followed by the reading of the label of the type $A$ omitting the reading of the indefiniete article which begins it, followed by ``is not'', followed by the reading of the label of the type $B$. 

\item For $A, B$ any types of the extended olog $\mathfrak O^{\bullet}=(\mathcal O^{\bullet},\mathcal F)$, a particular negative categorical proposition ${\bf O}_{AB}$ is interpreted as the diagram 
$$
D_{{\bf O}}(A,B):{\bf O}\rightarrow\mathcal O^{\bullet}
$$
identified by the assignments $0\mapsto A$, $1\mapsto \bullet$, $2\mapsto \bullet$, $3\mapsto B$ on nodes, and $10\mapsto \bullet A$, $12\mapsto \bullet\bullet$, $32\mapsto B\bullet$, where $\bullet A$, $\bullet\bullet$ and $B\bullet$ are the labels of whichever, but fixed by the assignments,
arcs with source $\bullet$ and target $A$, source $\bullet$ and target $\bullet$, source $B$ and target $\bullet$, respectively, if present. We hasten to specify that these arcs are not aspects and consequently their labels are not text expressions that are supposed to be read, in any sense. For this reason, if present, they will be unanonimously represented in any graphical representation of an ologism. To help the intuition, we observe in passing that a syllogistic constraint
\begin{eqnarray}\label{oioi}
\PN{A}{B}
\end{eqnarray}
can be intuitively thought of as declaring
the negation of the syllogistic constraint ${\bf A}_{AB}$, namely as declaring something like ``$B$ does not apply to all of $A$'s''. In general, by means of  an aspect $\xymatrix{A\ar[r]^f&B}$, everything which is measured by the instances of $A$ can be measured by the instance in $B$, like this
$$
\xymatrix{X\ar@/^1pc/[rr]^{g;f}\ar[r]_g&A\ar[r]_f&B}
$$
by composition of aspects. Now a syllogistic constraint ${\bf O}_{AB}$ amounts to the declaration which is the negation of this kind of situation when $f=g=\al{is}$, namely that there is a momentarily unspecified type $\bullet$ which can be measured by the instances of the type $A$ by means of a momentarily unspecified aspect $\bullet\to A$ but such that, at the same time, is someway definitely uncomparable with the instances of the type $B$. Because of remark~\ref{perdopo}, the kind of declaration at issue could be interpreted as a prescriptive one, to be accepted in the world-view proposed by the author of the ologism. In view of all the above, we prescribe that a syllogistic constraint such as~\eqref{oioi} has to be read by pronouncing the indefinite adjective ``Some'' followed by the reading of the label of the type $A$ omitting the reading of the indefinite article which begins it, followed by ``is not'', followed by the reading of the label of the type $B$. 
\end{enumerate}

\subsection{Ologs as ologisms}\label{ologsasologisms}
In view of what have been discussed in subsection~\ref{recat}, ologs can be seen as ologisms in which the only syllogistic constraints that are imposed are those expressed by the aspects identified by the label \al{is} provided with the olog, if present, or by the aspects that arise as composite of two such. In fact, in accordance with subsection~\ref{recat}, for $A, B, C$ any types in a fixed olog, on one side the commutativity of the diagram 
\begin{eqnarray}\label{caldo}
\xymatrix{A\ar[rr]^{\al{is}}\ar@/_/[drr]_{\al{is}}&&B\ar[d]^{\al{is}}\\
\ar@{}[urr]|(.6){\checkmark}&&C}
\end{eqnarray}
captures the well-known scheme of valid syllogism
\begin{eqnarray}\label{syllofirst}
\begin{tabular}{l}
Every $B$ is $C$,\\
Every $A$ is $B$ \\
therefore\\
Every $A$ is $C$
\end{tabular}
\end{eqnarray} 
whereas, on the other side, because of the rejection criterion described in subsection~\ref{rejection}, the only way to impose a syllogistic constraint ${\bf A}_{AC}$ as an aspect $\xymatrix{A\ar[r]^{\al{is}}&C}$, by means of the diagrammatic calculus described in section~\ref{calculus}, is by composition of two aspects, see remark~\ref{formalcomp}, as in diagram~\eqref{caldo}.

\subsection{Syllogisms as ologisms}\label{ologrammi}
Ologisms extend the diagrammatic calculus described in section~\ref{calculus}. We here exhibit a series of examples that convey a sufficiently clear idea of that, although intuitively. Throughout, the notation~\eqref{forsempre} will be freely employed.
\begin{enumerate}
\item As a particular example of an ologism that captures the fundamental (scheme of) valid syllogism~\eqref{syllofirst} one may consider the ologism
\begin{eqnarray}\label{fundamentsyll}
\fbox{\xymatrix{\stackrel{S}{\fbox{a square}}\ar@/_/@{-->}[drr]_{\al{is}}\ar[rr]^{\al{is}}&&\stackrel{R}{\fbox{a rectangle}}\ar[d]^{\al{is}}\\
\ar@{}[urr]|(.6){\checkmark}&&\stackrel{Q}{\fbox{a quadrilateral}}}}
\end{eqnarray}

\item The syllogism~\eqref{syllo1} is captured by the ologism
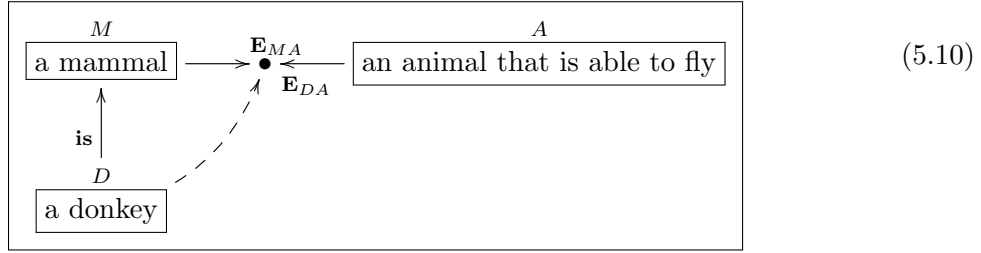
\begin{eqnarray}\label{primo}
\fbox{\xymatrix{\stackrel{M}{\fbox{a mammal}}\ar@{}[rr]^(.4){{\bf E}_{MA}}\ar[r]&\bullet\ar@{}[dr]|(.15){{\bf E}_{DA}}&\stackrel{A}{\fbox{an animal that is able to fly}}\ar[l]\\
\stackrel{D}{\fbox{a donkey}}\ar[u]^{\al{is}}\ar@/_1pc/@{-->}[ur]
&&}}
\end{eqnarray}
in which, the dashed arrow is the one that has been deduced by means of the diagrammatic calculus described in section~\ref{calculus}, from the 
premisses 
$$
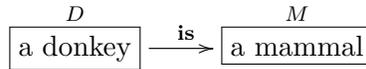
\xymatrix{\stackrel{D}{\fbox{a donkey}}\ar[r]^{\al{is}}&\stackrel{M}{\fbox{a mammal}}}
$$
of universal affirmative form, and ${\bf E}_{MA}$ of universal negative form, toward the obtainment of the conclusion ${\bf E}_{DA}$. On the base of what has been prescribed in subsection~\ref{recat} about the recognition and the reading of syllogistic constraint in ologisms, we can say that the ologism~\eqref{primo} communicates that 
the syllogistic constraints ``Every donkey is a mammal'' and ``Every mammal is not an animal that is able to fly'' have to be considered as imposed; from which the further syllogistic constraint ``Every donkey is not an animal that is able to fly'' follows.

\item The syllogism~\eqref{syllo2} is captured by the ologism 
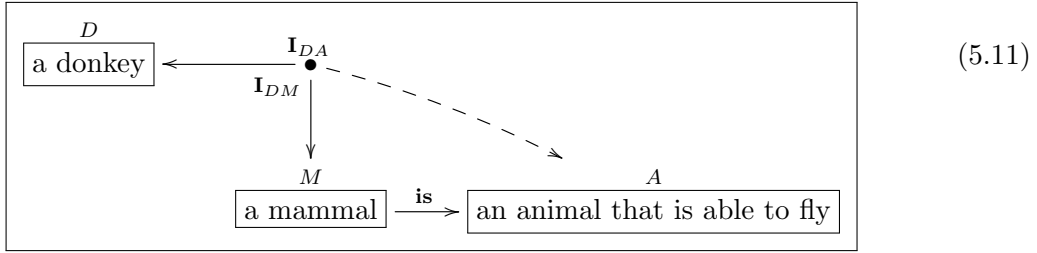
\begin{eqnarray}\label{secondo}
\fbox{\xymatrix{
\stackrel{D}{\fbox{a donkey}}&\bullet\ar@{}[dl]|(.15){{\bf I}_{DM}}\ar@/^/@{-->}[dr]\ar[l]_(.01){{\bf I}_{DA}}\ar[d]\\
&\stackrel{M}{\fbox{a mammal}}\ar[r]^(.33){\al{is}}&\stackrel{A}{\fbox{an animal that is able to fly}}}}
\end{eqnarray}
where the dashed arrow is the one that has been deduced by means of the diagrammatic calculus described in section~\ref{calculus}, from the premisses
$$
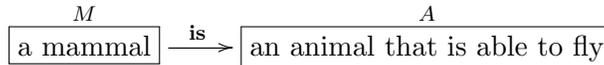
\xymatrix{\stackrel{M}{\fbox{a mammal}}\ar[r]^(.33){\al{is}}&\stackrel{A}{\fbox{an animal that is able to fly}}}
$$
of universal affirmative form and ${\bf I}_{DM}$ of particular affirmative form, toward the obtainment of the conclusion ${\bf I}_{DA}$. On the base of what has been prescribed in subsection~\ref{recat} about the recognition and the reading of syllogistic constraint in ologisms, we can say that the ologism~\eqref{secondo} communicates that 
the syllogistic constraints ``Some donkey is a mammal'' and ``Every mammal is an animal that is able to fly'' have to be considered as imposed; from which the further syllogistic constraint ``Some donkey is an animal that is able to fly'' follows.

\item Just before remark~\ref{formalcomp} we proved that the syllogism~\eqref{syllo3} was not valid. 
The syllogistic constraints which are imposable by means of its premisses and its suggested conclusion can be represented in the ologism
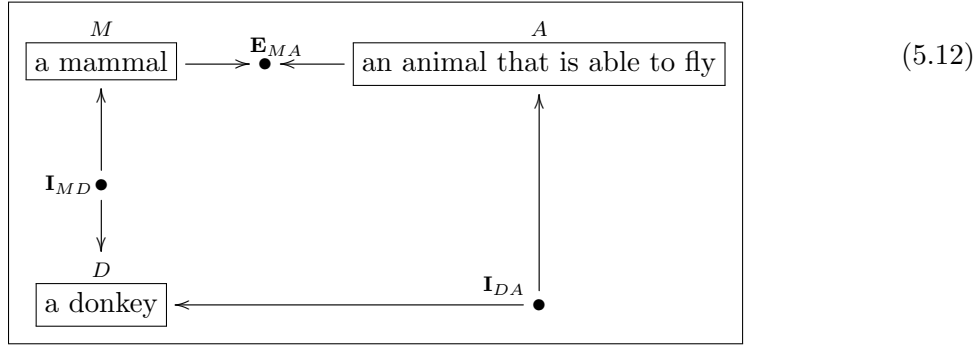
\begin{eqnarray}\label{unsound}
\fbox{\xymatrix{\stackrel{M}{\fbox{a mammal}}\ar@{}[rr]^(.4){{\bf E}_{MA}}\ar@{}[ddrr]|(.92){{\bf I}_{DA}}
\ar[r]&\bullet&\stackrel{A}{\fbox{an animal that is able to fly}}\ar[l]\\
\bullet\ar[u]\ar[d]\\
\stackrel{D}{\fbox{a donkey}}\ar@{}[uu]^{{\bf I}_{MD}}&&\bullet\ar[uu]\ar[ll]}}
\end{eqnarray}
Now, by having a look at the ologism~\eqref{unsound},
it is immediate to see that ${\bf I}_{DA}$ does not follow from the premisses ${\bf I}_{MD}$ and ${\bf E}_{MA}$, because of the rejection criterion described in subsection~\ref{rejection}, to confirm the invalidity of syllogism~\eqref{syllo3}. 
Nonetheless, from the syllogistic constraints shown in the ologism it is possible to deduce two further syllogistic constraints:
\begin{enumerate}
\item ``Some donkey is not an animal that is able to fly'', that is ${\bf O}_{DA}$, from the premisses ${\bf I}_{MD}$, ${\bf E}_{MA}$, as in the ologism
$$
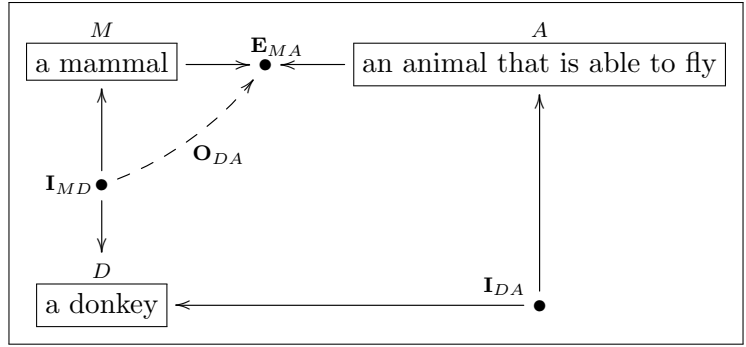
\fbox{\xymatrix{\stackrel{M}{\fbox{a mammal}}\ar@{}[rr]^(.4){{\bf E}_{MA}}\ar@{}[ddrr]|(.92){{\bf I}_{DA}}
\ar[r]&\bullet&\stackrel{A}{\fbox{an animal that is able to fly}}\ar[l]\\
\bullet\ar@/_/@{-->}[ur]_{{\bf O}_{DA}}\ar[u]\ar[d]\\
\stackrel{D}{\fbox{a donkey}}\ar@{}[uu]^{{\bf I}_{MD}}&&\bullet\ar[uu]\ar[ll]}}
$$
where the dashed arrow indicates the execution of the deduction at issue.

\item ``Some donkey is not a mammal'', that is ${\bf O}_{DM}$, from the premisses ${\bf I}_{DA}$ and ${\bf E}_{MA}$, as in the ologism
$$
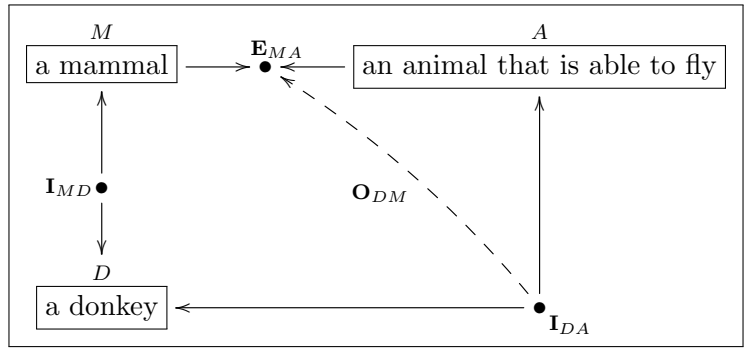
\fbox{\xymatrix{\stackrel{M}{\fbox{a mammal}}\ar@{}[rr]^(.4){{\bf E}_{MA}}\ar@{}[ddrr]|(1.07){{\bf I}_{DA}}
\ar[r]&\bullet&\stackrel{A}{\fbox{an animal that is able to fly}}\ar[l]\\
\bullet\ar[u]\ar[d]\\
\stackrel{D}{\fbox{a donkey}}\ar@{}[uu]^{{\bf I}_{MD}}&&\bullet\ar@/_/@{-->}[uul]^{{\bf O}_{DM}}\ar[uu]\ar[ll]}}
$$
where the dashed arrow indicates the execution of the  deduction at issue.
\end{enumerate}
Now, a couple of remarks are in order. First, we observe that the deduced syllogistic constraints ${\bf O}_{DA}$ in (i), ${\bf O}_{DM}$ in (ii), are not in contradiction. That is, apart from minor drawing  difficulties, they could be depicted jointly in a single ologism. Contradiction in ologisms will be discussed in subsection~\ref{contraologisms} below.
Second, we observe that in both the cases (i), (ii), one could be tempted to further going on to deduce: in (i) from the premisses ${\bf I}_{DA}$, ${\bf O}_{DA}$ whereas in (ii) from the premisses ${\bf I}_{MD}$, ${\bf O}_{DM}$; toward the obtainment, in both cases, of the diagram
$$
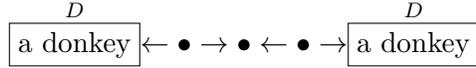
\stackrel{D}{\fbox{a donkey}}\leftarrow\bullet\rightarrow\bullet\leftarrow\bullet\rightarrow\stackrel{D}{\fbox{a donkey}}
$$
which is not among the diagrams~\eqref{diagrams}, so that it is meaningless, and does not impose any further syllogistic constraint. The moral is that in an ologism not all the syllogistic deduction that you would calculate correspond to the imposition of syllogistic constraints. 
By the appropriate employment of the rejection criterion discussed in subsection~\ref{rejection} one is able to stop at the right level of deduction.
\end{enumerate}

\begin{rem}
At this point it is worth to observe that remark~\ref{perpoi} is in line with remark~\ref{perdopo}. A person examining the ologism~\eqref{primo} could not agree with the constraint ``Every mammal is not an animal that is able to fly'', but this is not a problem with the ologism under examination, since it is the soundness of the diagrammatic calculus that we are employing to ensure that the ologism in question is structurally sound. Similarly, one could not agree with the syllogistic constraint ``Every mammal is an animal that is able to fly'' prescribed in ologism~\eqref{secondo} but nonetheless it is structurally sound for the same previously given reasons.
\end{rem}


\section{Ologisms, but formally}\label{extlabelded}
The aim of this section is that of situate the ideas that have been informally described in section~\ref{ologisms} within a rigorous mathematical framework, with a focus on the way the logical expressivity of ologs turns out to be extended with the possibility of imposing constraints of syllogistic nature, by means of the logical calculus described in section~\ref{calculus}. The idea that we pursue is that of looking at ologisms as to extended ologs. Whereas in section~\ref{labelded} ologs have been described from a proof-theoretical viewpoint as labelled deductive systems with imposed equational constraints, we here describe ologisms under a similar viewpoint as ologs extended with imposed syllogistic constraints, in the way that we are going to illustrate.

\subsection{Ologisms, but formally}
In this subsection we describe ologisms in rigorous mathematical terms as ologs extended with given syllogistic constraints, as premisses from which to deduce further syllogistic constraints as conclusions; that is as ologs provided with a deductive equipment which we precisely describe below. Before reading definition~\ref{ologism} below reconsider notation~\ref{AEIO}.

\begin{notation}\label{firstofall}
Let $\mathcal O$ be a graph equipped with a distinguished node labelled $\bullet$. We write
$\mathcal O_{\bullet}$ for the graph obtained from $\mathcal O$ by removing just the node $\bullet$ and every arc having $\bullet$ as source or target. 
\end{notation}

\begin{defi}\label{ologism}
An {\bf ologism} is a $6$-tuple
$\mathfrak O=(\mathcal O, \mathcal F, \alpha, \epsilon, \iota, o)$ where 

\begin{enumerate}

\item $\mathcal O$ is a graph with a distinguished node labelled $\bullet$. 

\item\label{due} $\mathcal F$ is a set of diagrams in $\mathcal O_{\bullet}$, see notation~\ref{firstofall}, and
the pair $\mathfrak O_{\bullet}=(\mathcal O_{\bullet},\mathcal F)$ is required to be an olog, that is the {\bf underlying olog} of $\mathfrak O$. Every node, or arc, of $\mathcal O_{\bullet}$ is a type, 
or an aspect, respectively. More explicitly, $\bullet$ is not a type, and the arcs of $\mathcal O$ that have $\bullet$ as source or target are not aspects. All the nodes of $\mathcal O$ which are different from $\bullet$ are types.

\item $\alpha$ is a set of diagrams
\begin{eqnarray}\label{ua}
D_{{\bf A}}(A,B):\bf A\rightarrow \mathcal O_{\bullet}
\end{eqnarray}
for $A, B$ any nodes of $\mathcal O_{\bullet}$. We prescribe that a diagram in $\alpha$, such as~\eqref{ua}, is identified by the assignments $0\mapsto A, 1\mapsto B$ on nodes and $01\mapsto\al{is}$ on arcs, where $\xymatrix{A\ar[r]^{\al{is}}&B}$ is an aspect in $\mathfrak O_{\bullet}$, if present. 

\item $\epsilon$ is a set of diagrams 
\begin{eqnarray}\label{un}
D_{{\bf E}}(A,B):{\bf E}\rightarrow\mathcal O
\end{eqnarray}
for $A, B$ any nodes of $\mathcal O_{\bullet}$. We prescribe that a diagram in $\epsilon$, such as~\eqref{un}, is identified by the assignments $0\mapsto A$, $1\mapsto\bullet$, $2\mapsto B$ on nodes, and $01\mapsto A\bullet$, $21\mapsto B\bullet$ on arcs, where $A\bullet$ and $B\bullet$ are the labels of whichever arcs with source $A$ and target $\bullet$ and source $B$ and target $\bullet$ of $\mathcal O$, respectively, if present.

\item $\iota$ is a set of diagrams 
\begin{eqnarray}\label{pa}
D_{{\bf I}}(A,B):{\bf I}\rightarrow\mathcal O
\end{eqnarray}
for $A, B$ any nodes of $\mathcal O_{\bullet}$. We prescribe that a diagram in $\iota$, such 
as~\eqref{pa}, is identified by the assignments $0\mapsto A$, $1\mapsto\bullet$, $2\mapsto B$ on nodes, and $10\mapsto \bullet A$, $12\mapsto \bullet B$ on arcs, where $\bullet A$ and $\bullet B$ are the labels of whichever arcs with source $\bullet$ and target $A$ and source $\bullet$ and target $B$ of $\mathcal O$, respectively, if present.

\item $o$ is a set of diagrams
\begin{eqnarray}\label{pn}
D_{{\bf O}}(A,B):{\bf O}\rightarrow\mathcal O
\end{eqnarray}
for $A, B$ any nodes of $\mathcal O_{\bullet}$. We prescribe that a diagram in $o$, such as~\eqref{pn}, is identified by the assignments $0\mapsto A$, $1\mapsto \bullet$, $2\mapsto \bullet$, $3\mapsto B$ on nodes, and $10\mapsto \bullet A$, $12\mapsto \bullet\bullet$, $32\mapsto B\bullet$, where $\bullet A$, $\bullet\bullet$ and $B\bullet$ are the labels of whichever arcs with source $\bullet$ and target $A$, source $\bullet$ and target $\bullet$, source $B$ and target $\bullet$, respectively, if present.
\end{enumerate}
\end{defi}

\begin{terminology}\label{terminpremiss}
Let $\mathfrak O=(\mathcal O,\mathcal F, \alpha, \epsilon, \iota, o)$ be an ologism. Henceforth
\begin{itemize}
\item[-] the elements of $\alpha$  will be referred to as {\bf universal affirmative premisses};
\item[-] the elements of $\epsilon$ will be henceforth referred to as {\bf universal negative premisses};
\item[-] the elements of $\iota$ will be henceforth referred to as {\bf particular affirmative premisses};
\item[-] the elements of $o$ will be henceforth referred to as {\bf particular negative premisses}.
\end{itemize}
Altogether, the elements of $\alpha, \epsilon, \iota, o$ will be henceforth referred to as {\bf syllogistic premisses}.
\end{terminology}

\begin{defi}\label{dedeq}
The {\bf deductive equipment} of 
an ologism $\mathfrak O$ consists of its syllogistic premisses
together with the following rules of inference:
\begin{enumerate}
\item
\begin{eqnarray}\label{cutcut}
\AxiomC{$\xymatrix{A\ar[r]^{\al{is}}&B}$}
\AxiomC{$\xymatrix{B\ar[r]^{\al{is}}&C}$}
\BinaryInfC{$\xymatrix{A\ar[r]^{\al{is}}&C}$}
\DisplayProof
\end{eqnarray}

\item
\begin{eqnarray}
\AxiomC{$\UN{A}{B}$}
\AxiomC{$\xymatrix{B&C\ar[l]_{\al{is}}}$}
\BinaryInfC{$\UN{A}{C}$}
\DisplayProof\label{ae}
\\[2ex]
\AxiomC{$\xymatrix{A\ar[r]^{\al{is}}&B}$}
\AxiomC{$\UN{B}{C}$}
\BinaryInfC{$\UN{A}{C}$}
\DisplayProof\label{ea}
\end{eqnarray}

\item 
\begin{eqnarray}
\AxiomC{$\xymatrix{A\ar@{}[rr]^{{\bf I}_{AB}}&\bullet\ar[l]\ar[r]&B}$}
\AxiomC{$\xymatrix{B\ar[r]^{\al{is}}&C}$}
\BinaryInfC{$\PA{A}{C}$}
\DisplayProof\label{25}
\\[2ex]
\AxiomC{$\xymatrix{A&B\ar[l]_{\al{is}}}$}
\AxiomC{$\xymatrix{B\ar@{}[rr]^{{\bf I}_{BC}}&\bullet\ar[l]\ar[r]&C}$}
\BinaryInfC{$\PA{A}{C}$}
\DisplayProof\label{26}
\end{eqnarray}

\item
\begin{eqnarray}
\AxiomC{$\PA{A}{B}$}
\AxiomC{$\UN{B}{C}$}
\BinaryInfC{$\PN{A}{C}$}
\DisplayProof\label{oo}
\\[2ex]
\AxiomC{$\xymatrix{A&B\ar[l]_{\al{is}}}$}
\AxiomC{$\PN{B}{C}$}
\BinaryInfC{$\PN{A}{C}$}
\DisplayProof\label{210}
\\[2ex]
\AxiomC{$\PN{A}{B}$}
\AxiomC{$\xymatrix{B&C\ar[l]_{\al{is}}}$}
\BinaryInfC{$\PN{A}{C}$}
\DisplayProof\label{211}
\end{eqnarray}
\end{enumerate}
\end{defi}

\begin{rem}
Among the universal affirmative premisses in the deductive equipment of an ologism there are, for every type $A$, distinguished identical universal affirmative premisses $\xymatrix{A\ar[r]^{\al{is}}&A}$, see subsection~\ref{ologwitheq}. From a syllogistic viewpoint, those particular universal affirmative premisses correspond to the identity laws, identified and taken for granted by Aristotle. Since in any case they will be interpreted by identical functional relationships they won't be pictorially represented in ologisms, as already they were not in ologs.
\end{rem}

\begin{exms}\label{examplesofologisms}
\begin{enumerate}
\item\label{poorly} On the base of definition~\ref{ologism}, from a purely formal point of view, the data that consist of a graph whose only node is $\bullet$, with no arcs at all, and consequently $\alpha=\epsilon=\iota=o=\emptyset$ identify an example of ologism; actually a poorly expressive one.

\item\label{b} The author of an ologism could have the need to fix the state of things described by the following affirmations:
\begin{itemize}
\item[-] Some bird is not an animal that is able to fly.
\item[-] Some mammal is an animal that is able to fly.
\item[-] Every bird is not a mammal.
\item[-] Every bird is a vertebrate.
\item[-] Every mammal is a vertebrate.
\end{itemize}
To that end, she could draw the ologism
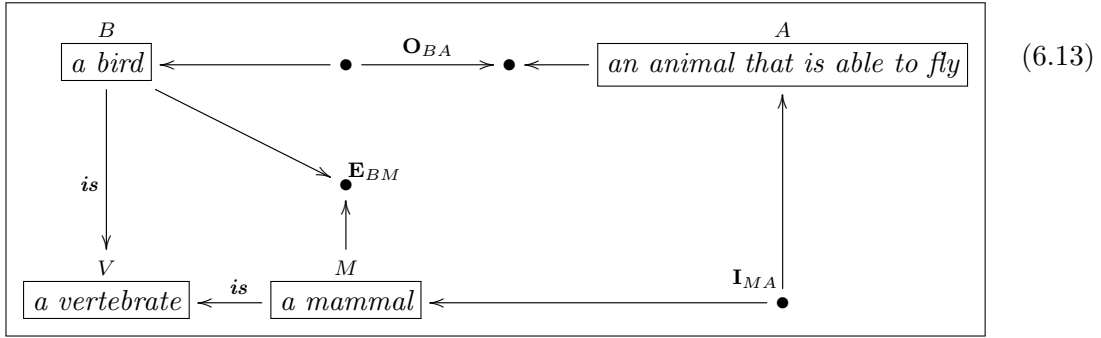
\begin{eqnarray}\label{animals}
\fbox{\xymatrix{\stackrel{B}{\fbox{a bird}}
\ar[dd]_{\al{is}}\ar[dr]&\bullet\ar[l]\ar[r]^{{\bf O}_{BA}}&\bullet\ar@{}[ddr]|(.9){{\bf I}_{MA}}&\stackrel{A}{\fbox{an animal that is able to fly}}\ar[l]\\
&\bullet\\
\stackrel{V}{\fbox{a vertebrate}}\ar@{}[ur]|(1.12){{\bf E}_{BM}}&\stackrel{M}{\fbox{a mammal}}\ar[u]\ar[l]_(.45){\al{is}}&&\bullet\ar[ll]\ar[uu]
}}
\end{eqnarray}
and subsequently to deduce further syllogistic constraints, if possible. In the present case some of the deducible syllogistic constraints are 
\begin{itemize}
\item[-] ${\bf O}_{AB}$ for ``Some animal that is able to fly is not a bird'', by virtue of the following proof trees in which an instance of the rule~\eqref{oo} occurs:
$$
\AxiomC{$\PA{A}{M}$}
\AxiomC{$\UN{B}{M}$}
\UnaryInfC{$\UN{M}{B}$}
\BinaryInfC{$\PN{A}{B}$}
\DisplayProof
$$
\item[-] ${\bf I}_{AV}$ for ``Some animal that is able to fly is a vertebrate'', by virtue of the following proof tree in which an instance of the rule~\eqref{25} occurs:
$$
\AxiomC{$\PA{M}{A}$}
\UnaryInfC{$\PA{A}{M}$}
\AxiomC{$\xymatrix{M\ar[r]^{\al{is}}&V}$}
\BinaryInfC{$\PA{A}{V}$}
\DisplayProof
$$

\item[-]
${\bf O}_{VA}$ for ``Some vertebrate is not an animal that is able to fly'', by virtue of the followig proof tree in which an instance of the rule of inference~\eqref{210} occurs:
$$
\AxiomC{$\xymatrix{B\ar[r]^{\al{is}}&V}$}
\UnaryInfC{$\xymatrix{V&B\ar[l]_{\al{is}}}$}
\AxiomC{\PN{B}{A}}
\BinaryInfC{\PN{V}{A}}
\DisplayProof
$$
\end{itemize}
which can be pictorially represented in accordance with notation~\ref{forsempre} by means of the dashed arrows in the ologism 
$$
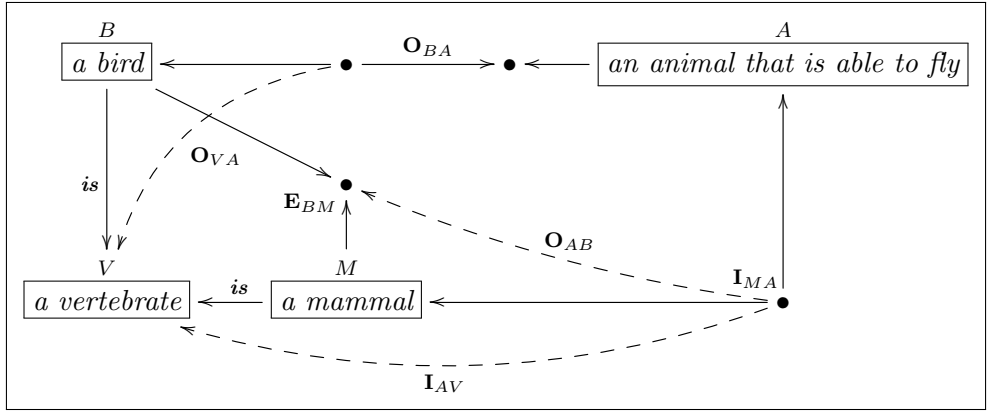
\fbox{\xymatrix{\stackrel{B}{\fbox{a bird}}
\ar[dd]_{\al{is}}\ar[dr]&\bullet\ar@{-->}@/_2pc/[ddl]^{{\bf O}_{VA}}\ar[l]\ar[r]^{{\bf O}_{BA}}&\bullet\ar@{}[ddr]|(.9){{\bf I}_{MA}}&\stackrel{A}{\fbox{an animal that is able to fly}}\ar[l]\\
&\bullet\\
\stackrel{V}{\fbox{a vertebrate}}\ar@{}[ur]|(.85){{\bf E}_{BM}}&\stackrel{M}{\fbox{a mammal}}\ar[u]\ar[l]_(.45){\al{is}}&&\bullet\ar@{-->}@/^/[ull]_{{\bf O}_{AB}}\ar@/^2pc/@{-->}[lll]^{{\bf I}_{AV}}\ar[ll]\ar[uu]
}}
$$

From a formal point of view, in accordance
with definition~\ref{ologism} and definition~\ref{dedeq}, the ologism~\eqref{animals} can be specified by the assignment of the following data:
\begin{itemize}
\item[-] for $\mathcal O$:
\begin{eqnarray*}
O_0=\{\stackrel{B}{\ulcorner\textrm{a bird}\urcorner}, \stackrel{V}{\ulcorner\textrm{a vertebrate}\urcorner}, \stackrel{A}{\ulcorner\textrm{an animal that is able to fly}\urcorner}, \stackrel{M}{\ulcorner\textrm{a mammal}\urcorner},\bullet\}\\
O_1=\{\al{is}_{BV}, \al{is}_{MV}, M\bullet, B\bullet, \bullet M,\bullet A, A\bullet, \bullet\bullet,\bullet B\}\\
\partial_0:\left[\begin{array}{l}
\al{is}_{BV}\mapsto B\\
\al{is}_{MV}\mapsto M\\
M\bullet\mapsto M\\
B\bullet\mapsto B\\
\bullet M\mapsto \bullet\\
\bullet A\mapsto \bullet\\
A\bullet\mapsto\bullet\\
\bullet\bullet\mapsto\bullet\\
\bullet B\mapsto \bullet
\end{array}\right]:O_1\rightarrow O_0\\
\partial_1:\left[\begin{array}{l}
\al{is}_{BV}\mapsto V\\
\al{is}_{MV}\mapsto V\\
M\bullet\mapsto \bullet\\
B\bullet\mapsto \bullet\\
\bullet M\mapsto M\\
\bullet A\mapsto A\\
A\bullet\mapsto\bullet\\
\bullet\bullet\mapsto\bullet\\
\bullet B\mapsto B
\end{array}\right]:O_1\rightarrow O_0
\end{eqnarray*}
\item[-] $\mathcal F=\emptyset$.
\item[-] $\alpha=\{D_{{\bf A}}(B,V), D_{{\bf A}}(M,V)\}$ where $D_{{\bf A}}(B,V):\bf A\rightarrow\mathcal O_{\bullet}$ is identified by the assignments $0\mapsto B, 1\mapsto V$ on nodes, and $01\mapsto\al{is}_{BV}$ on arcs; whereas 
$D_{{\bf A}}(M,V):\bf A\rightarrow\mathcal O_{\bullet}$ is identified by the assignments $0\mapsto M, 1\mapsto V$ on nodes, and $01\mapsto\al{is}_{MV}$ on arcs.
\item[-] $\epsilon=\{D_{{\bf E}}(B,M)\}$ where $D_{{\bf E}}(B,M):{\bf E}\rightarrow \mathcal O$ is identified by the assignments
$0\mapsto B, 1\mapsto \bullet,
2\mapsto M$
on nodes, and $01\mapsto B\bullet, 21\mapsto  M\bullet$ on arcs.
\item[-] $\iota=\{D_{{\bf I}}(M,A)\}$ with  $D_{{\bf I}}(M,A):{\bf I}\rightarrow\mathcal O$ identified by the assignments $0\mapsto M, 1\mapsto \bullet, 2\mapsto A$ on nodes, and $10\mapsto \bullet M, 12\mapsto\bullet A$.
\item[-] $o=\{D_{\bf O}(B,A)\}$ where $D_{\bf O}(B,A):\bf O\rightarrow\mathcal O$ is identified by the assignments $0\mapsto B, 1\mapsto\bullet, 2\mapsto\bullet, 3\mapsto A$ on nodes, and $10\mapsto\bullet B, 12\mapsto\bullet\bullet , 32\mapsto A\bullet$ on arcs.
\end{itemize}

The deductive equipment of the ologism that we are specifying consists of the universal negative premise $D_{\bf E}(B,M)$, the particular affirmative premise $D_{\bf I}(M, A)$, the particular negative premise $D_{\bf O}(B,A)$ and the universal affirmative premisses provided by the aspects $\xymatrix{M\ar[r]^{\al{is}}&V}$, $\xymatrix{B\ar[r]^{\al{is}}&V}$. The deduction of the syllogistic constraints ${\bf O}_{AB}$, ${\bf I}_{AV}$ and ${\bf O}_{VA}$ from the given premisses is by means of the proof trees described above.
\end{enumerate}
\end{exms}

\subsection{The logical theory generated by an ologism}
In subsection~\ref{proofeqsketch} and in subsection~\ref{ologwitheq} we pointed out that every olog generates an equational theory. In this section, we point out that an ologism $\mathfrak O=(\mathcal O,\mathcal F,\alpha, \epsilon, \iota, o)$ generates a logical theory $Th(\mathfrak O)$ which extends the equational theory generated by the olog underlying it, $\mathfrak O_{\bullet}=(\mathcal O_{\bullet},\mathcal F)$, see definition~\ref{ologism}, with every syllogistic conclusion deducible from the syllogistic premisses provided in the deductive equipment of the ologism $\mathfrak O$, by means of the rules of inference listed in points (i) to (iv)  in definition~\ref{dedeq}. In more mathematical terms, for $\mathfrak O$ an ologism, $Th(\mathfrak O)$ is a graph which is obtained from the category $Th(\mathfrak O_{\bullet})$ by adding to it all the diagrammatic syllogistic conclusions which are computable by means of the previously mentioned rules of inference. In general, $Th(\mathfrak O)$ 
is not a category, because the only syllogistic premisses which are aspects are the universal affirmative ones; considering also that up to the application of the identification condition $\al{is};\al{is}=\al{is}$, the rule of inference~\eqref{cutcut} is a particular instance of the rule of inference~\eqref{cat}. 

\begin{defi}\label{logthologism}
Let $\mathfrak O=(\mathcal O,\mathcal F,\alpha, \epsilon, \iota, o)$ be an ologism. 
The logical theory generated by $\mathfrak O$ is the graph $Th(\mathfrak O)$ obtained by adding to the category $Th(\mathfrak O_{\bullet})$ the node $\bullet$ and the arcs which are generated by means of all the possible applications of the rules of inference listed in points (i) to (iv) in definition~\ref{dedeq} to the syllogistic premisses in the deductive equipment of $\mathfrak O$. Thus, we will describe $Th(\mathfrak O)$ as a $5$-tuple
$(Th(\mathfrak O_{\bullet}),\alpha^*,\epsilon^*,\iota^*, o^*)$ 
where 
\begin{itemize}
\item[-] $\alpha^*$ is the set of universal affirmative propositions which are deducible from the universal affirmative propositions in $\alpha$ under all the possible applications of the rule~\eqref{cutcut};
\item[-] $\epsilon^*$ is the set of universal negative propositions which are deducible from the syllogistic premisses in $\epsilon\cup\alpha^*$ under all the possible applications of the rules of inference~\eqref{ae},~\eqref{ea};
\item[-] $\iota^*$ is the set of particular affirmative propositions which are deducible from the syllogistic premisses in $\iota\cup\alpha^*$ under all the possible applications of the rules of inference~\eqref{25},~\eqref{26};
\item[-] $o^*$ is the set of particular negative propositions which are deducible from the syllogistic premisses in $(\epsilon^*\cup\iota^*)\cup(\alpha^*\cup o)$ under all the possible applications of the rules of inference~\eqref{oo},~\eqref{210},~\eqref{211}.
\end{itemize}
\end{defi}

\begin{rem}\label{star}
On the base of what has been pointed out in subsection~\ref{ologwitheq} it must be kept in mind that the equational theory generated by an olog provides identity axioms as identity aspects $\xymatrix{A\ar[r]^{\al{is}}&A}$, for $A$ a type. Because of this $\alpha\subseteq\alpha^*$, thanks to the following instances of the rule~\eqref{cutcut}:
\begin{eqnarray*}
\AxiomC{$\xymatrix{A\ar[r]^{\al{is}}&A}$}
\AxiomC{$\xymatrix{A\ar[r]^{\al{is}}&B}$}
\BinaryInfC{$\xymatrix{A\ar[r]^{\al{is}}&B}$}
\DisplayProof
\\
\\
\AxiomC{$\xymatrix{A\ar[r]^{\al{is}}&B}$}
\AxiomC{$\xymatrix{B\ar[r]^{\al{is}}&B}$}
\BinaryInfC{$\xymatrix{A\ar[r]^{\al{is}}&B}$}
\DisplayProof
\end{eqnarray*}
and, for the same reason $\alpha^*$ contains at least the identity aspects, so that $\alpha^*\neq\emptyset$ even in case $\alpha=\emptyset$. Similarly, $\epsilon\subseteq\epsilon^*$, $\iota\subseteq\iota^*$ and $o\subseteq o^*$, thanks to suitable instances of the rules of inference~\eqref{ea},~\eqref{ae},~\eqref{25},~\eqref{26}and~\eqref{210},~\eqref{211}, respectively. 
\end{rem}

\begin{exms}
Let $\mathfrak O$ be the ologism that has been formally specified in point~\eqref{b} of examples~\ref{examplesofologisms}. The logical theory $Th(\mathfrak O)$ is identified by the following data:
\begin{itemize}
\item[-] the nodes of $Th(\mathfrak O_{\bullet})$ are the types of $\mathfrak O$.
\item[-] the arcs of $Th(\mathfrak O_{\bullet})$ are the identity aspects freely
  added to each of its nodes plus the aspects already present, plus the aspects
  generated by formal composition, that is in accordance with the rule of
  inference~\eqref{cat}. So, the identified set of arcs is $\{\al{is}_{BB},
  \al{is}_{AA}, \al{is}_{MM}, \al{is}_{VV}, \al{is}_{MV},$ $\al{is}_{BV}\}$ and, pictorially, $Th(\mathfrak O_{\bullet})$ appears as
$$
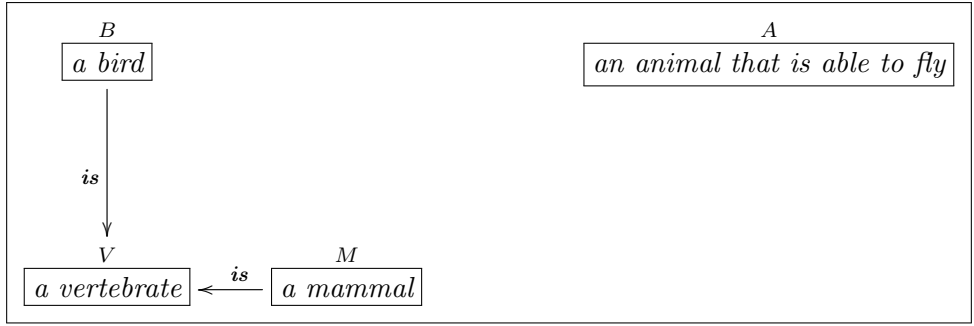
\fbox{\xymatrix{\stackrel{B}{\fbox{a bird}}
\ar[dd]_{\al{is}}
&&&\stackrel{A}{\fbox{an animal that is able to fly}}
\\
&\\
\stackrel{V}{\fbox{a vertebrate}}
&\stackrel{M}{\fbox{a mammal}}
\ar[l]_(.45){\al{is}}}}
$$
in which, to save space, we did not represent the above mentioned set of arcs.
\item[-] $\alpha^*=\{\al{is}_{VV},\al{is}_{MM},\al{is}_{BB},\al{is}_{AA}, \al{is}_{MV}, \al{is}_{BV}\}$.
\item[-] $\epsilon^*=\epsilon=\{\al{E}_{BM}\}$.
\item[-] $\iota^*=\{\al{I}_{MA},\al{I}_{AV}\}$.
\item[-] $o^*=\{\al{O}_{BA}, \al{O}_{VA}, \al{O}_{AB}\}$.
\end{itemize}
Thus, pictorially, $Th(\mathfrak O)$ appears as
$$
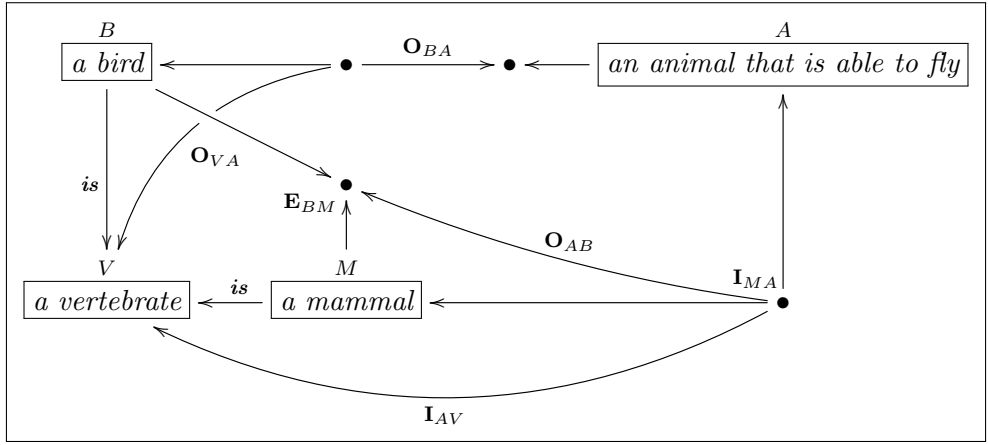
\fbox{\xymatrix{\stackrel{B}{\fbox{a bird}}
\ar[dd]_{\al{is}}\ar[dr]&\bullet\ar@/_2pc/[ddl]|(.4){\hole}^{{\bf O}_{VA}}\ar[l]\ar[r]^{{\bf O}_{BA}}&\bullet\ar@{}[ddr]|(.9){{\bf I}_{MA}}&\stackrel{A}{\fbox{an animal that is able to fly}}
\ar[l]
\\
&\bullet\\
\stackrel{V}{\fbox{a vertebrate}}
\ar@{}[ur]|(.85){{\bf E}_{BM}}&\stackrel{M}{\fbox{a mammal}}
\ar[u]\ar[l]_(.45){\al{is}}&&\bullet\ar@/^/[ull]_{{\bf O}_{AB}}
\ar@/^3pc/[lll]^{{\bf I}_{AV}}\ar[ll]\ar[uu]}}
$$
\end{exms}


\section{Ologisms and their models}\label{models}
In this section we provide a definition of what  accounts for a model of an ologism and discuss a series of examples.
  
\begin{defi}\label{modologism}
Let $\mathfrak O=(\mathcal O,\mathcal F, \alpha, \epsilon, \iota, o)$ be an ologism. A {\bf model} of $\mathfrak O$ is a homomorphism of graphs $\mathcal M:\mathcal O_{\bullet}\rightarrow{\bf Sets}$, which is a model of the olog $\mathfrak O_{\bullet}=(\mathcal O_{\bullet}, \mathcal F)$ underlying the ologism $\mathfrak O$, requested to furthermore satisfy the following prescriptions:
\begin{enumerate}
\item for every types $A, B$ linked by a universal affirmative premise $D_{\bf A}(A,B)\in\alpha$ as an aspect $\xymatrix{A\ar[r]^{\al{is}}&B}$, it must be the case that $\mathcal MA\subseteq\mathcal MB$; 
\item for every types $A, B$ linked by a universal negative premise $D_{\bf E}(A,B)\in\epsilon$ as a diagram $\UN{A}{B}$ it must be the case that $\mathcal MA\cap\mathcal MB=\emptyset$;
\item for every types $A, B$ linked by a 
particular affirmative premise $D_{\bf I}(A,B)\in\iota$ as a diagram $\PA{A}{B}$, it must be the case that 
$\mathcal MA\cap\mathcal MB\neq\emptyset$;
\item for every types $A, B$ linked by a 
particular negative premise $D_{\bf O}(A,B)\in o$ as a diagram $\PN{A}{B}$, it must be the case that $\mathcal MA\centernot\subseteq\mathcal MB$.
\end{enumerate}
\end{defi}

\begin{rem}
A series of observations are in order.
\begin{itemize}
\item[-] A model of an ologism is nothing but a particular model of its underlying olog. This is in line with what we requested in point~\ref{due} of definition~\ref{ologism}, namely that $\bullet$ is not a type and that the arcs of $\mathcal O$ that have $\bullet$ as source or target are not aspects.
\item[-] In accordance with what we observed in subsection~\ref{syllogistic} and subsection~\ref{contradiction}, in definition~\ref{modologism} the prescription (i) is contradictory to (iv) and the prescription (ii) is contradictory to (iii).
\item[-] The inclusion condition $\mathcal MA\subseteq \mathcal MB$ in 
point (i) of definition~\ref{modologism} holds automatically, because, as already observed in remark~\ref{isa},
an aspect $\xymatrix{A\ar[r]^{\al{is}}&B}$ is modelled by the function $\xymatrix{MA\ar@{^{(}->}[r]^{M\al{is}}&MB}$ which is the inclusion of the set $\mathcal MA$ in the set $\mathcal MB$. 
\item[-] The condition $\mathcal MA\cap\mathcal MB=\emptyset$ can always be forced to be satisfied. There is a standard way to do this. In case $x$ is an instance of the type $A$ and also an instance of the type $B$, it is enough to let $x$ be tagged as ``$A$'' to consider it as an instance of $A$, to let $x$ be tagged as ``$B$'' to consider it as an instance of $B$. Formally, all this amunts to put $\mathcal MA=\{(x,A)\mid x\textrm{ is an instance of }A\}$ and $\mathcal MB=\{(x,B)\mid x\textrm{ is an instance of }B\}$. For $x$ an instance of both the types $A, B$, the element $(x,A)\in\mathcal MA$ is not the same as the element $(x,B)\in\mathcal MB$. In practice, the necessity of forcing the condition $\mathcal MA\cap \mathcal MB=\emptyset$ may arise not only by virtue of the fact that different types have one or more instances in common that the author of an ologism would like to regard as distinct, but also because it may happen that instances which are already distinct in different types have the same label. 
\end{itemize}
\end{rem}

\begin{exms}
\begin{enumerate}
\item The ologism described in point~\eqref{poorly} of examples~\ref{examplesofologisms} has exactly one possible model, that is the empty one; since $\mathcal O_{\bullet}$ is the empty graph.
\item A model of the ologism~\eqref{animals} is identified by the assignments
\begin{eqnarray*}
M_0:\left[\begin{array}{l}
B\mapsto\{\textrm{eagle}, \textrm{hawk}, \textrm{flamingo}, \textrm{penguin}\}\\
A\mapsto\{\textrm{bat}, \textrm{dragonfly}\}\\
V\mapsto\{\textrm{eagle}, \textrm{hawk}, \textrm{flamingo}, \textrm{penguin}, \textrm{dog}, \textrm{cat}\}\\
M\mapsto\{\textrm{bat}, \textrm{dog}, \textrm{cat}\}
\end{array}\right]
\\
\\
M_1:\left[\begin{array}{l}
\al{is}_{BV}\mapsto\left[\begin{array}{l}
\textrm{eagle}\mapsto \textrm{eagle}\\
\textrm{hawk}\mapsto\textrm{hawk}\\ 
\textrm{flamingo}\mapsto\textrm{flamingo}\\
\textrm{penguin}\mapsto\textrm{penguin}
\end{array}\right]
\\
\\
\al{is}_{MV}\mapsto\left[\begin{array}{l}
\textrm{bat}\mapsto\textrm{bat}\\ 
\textrm{dog} \mapsto \textrm{dog}\\
\textrm{cat} \mapsto\textrm{cat}
\end{array}\right]
\end{array}\right]
\end{eqnarray*}
All the imposed constraints are satisfied because a penguin is a bird which is not an animal that is able to fly, a bat is an animal that is able to fly and a mammal, every bird is not a mammal. The set of mammals is included in the set of vertebrates, the set of birds is included in the set of vertebrates. As a consequence of this, every deducible syllogistic constraint hold as well, as the reader can easily verify.

\item Imagine that a system administrator has to manage the daily shift change between custodians, inspectors and helpers. At the level of conceptual planning she could design an ologism on the base of the following prescriptions:
\begin{itemize}
\item[-] every day there must be at least one custodian;
\item[-] every day there must be at least one inspector;
\item[-] no custodian is an inspector;
\item[-] to every custodian is associated an helper;
\item[-] every inspector is a helper;
\end{itemize}
and the result could appear as
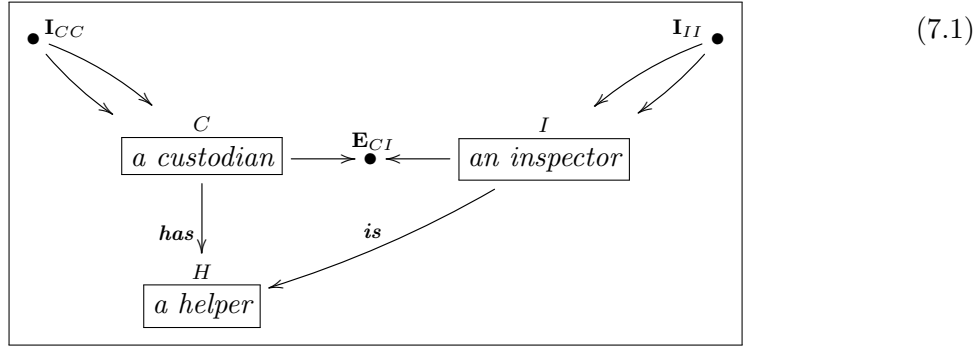
\begin{eqnarray}\label{cih}
\fbox{\xymatrix{
\bullet\ar@/_/[dr]\ar@/^/[dr]^(.1){{\bf I}_{CC}}&&&&\bullet\ar@/_/[dl]_(.1){{\bf I}_{II}}\ar@/^/[dl]\\
&\stackrel{C}{\fbox{a custodian}}\ar[d]_{\al{has}}\ar[r]\ar@{}[rr]^{{\bf E}_{CI}}&\bullet&\stackrel{I}{\fbox{an inspector}}\ar@/^/[dll]_{\al{is}}\ar[l]\\
&\stackrel{H}{\fbox{a helper}}
}}
\end{eqnarray}
a model of which corresponds to the situation about which custodians, inspectors and helpers are involved in the shift of a current day. Day by day the model has to be updated. More to the point, a few comments are in order. The particular affirmative premisses ${\bf I}_{CC}$, ${\bf I}_{II}$ are existential import conditions, see subsection~\ref{eximport}, which are there to witness that the types $C, I$ cannot be modelled by the emptyset. It is required that at least one custodian and at least one inspector must be present. Such premisses have been pictorially represented as shown rather than extensively as $\PA{C}{C}$ and $\PA{I}{I}$, respectively, just to save space. The types $C, I$ must be modelled by disjoint sets, as imposed by the universal negative premise ${\bf E}_{CI}$. Finally, we observe that from the ologism~\eqref{cih} it would be an error to deduce that the inspectors are the helpers associated to the custodians. Nothing in the ologism forces this deduction. In any case, to exclude that things be necessarily that way, one could add the syllogistic constraint ${\bf O}_{HI}$, to let there be helpers which are not inspectors. The syllogistic constraints which are deducible from those imposed in the ologism~\eqref{cih} are witnessed by the dashed arrows in the ologism
$$
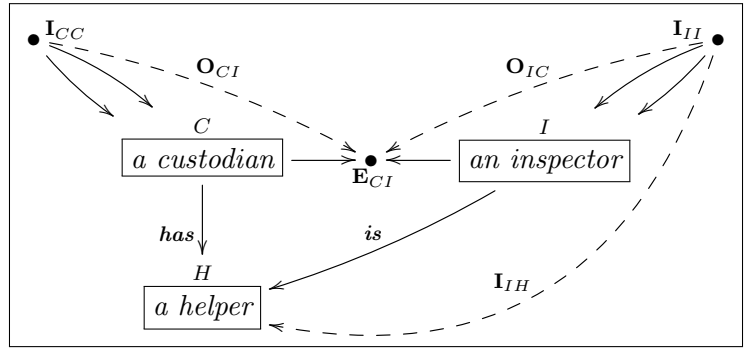
\fbox{\xymatrix{
    \bullet\ar@{-->}@/^/[drr]^{{\bf O}_{CI}}\ar@/_/[dr]\ar@/^/[dr]^(.1){{\bf I}_{CC}}&&&&\bullet\ar@{-->}@/^4.5pc/[ddlll]_{{\bf I}_{IH}}\ar@{-->}@/_/[dll]_{{\bf O}_{IC}}\ar@/_/[dl]_(.1){{\bf I}_{II}}\ar@/^/[dl]\\
&\stackrel{C}{\fbox{a custodian}}\ar[d]_{\al{has}}\ar[r]\ar@{}[rr]_{{\bf E}_{CI}}&\bullet&\stackrel{I}{\fbox{an inspector}}\ar@/^/[dll]_{\al{is}}\ar[l]\\
&\stackrel{H}{\fbox{a helper}}
}}
$$
that is 
\begin{itemize}
\item[-] ${\bf O}_{CI}$ and ${\bf O}_{IC}$ by virtue of the following proof trees in which suitable instances of the rule of inference~\eqref{210} occur:
\begin{eqnarray*}
\AxiomC{\PA{C}{C}}
\AxiomC{\UN{C}{I}}
\BinaryInfC{$\xymatrix{C\ar@{}[rr]^{{\bf I}_{CC}}&\bullet\ar[r]\ar[l]&C\ar@{}[rr]^{{\bf E}_{CI}}\ar[r]&\bullet&I\ar[l]}$}
\UnaryInfC{\PN{C}{I}}
\DisplayProof
\\
\\
\\
\AxiomC{\PA{I}{I}}
\AxiomC{\UN{C}{I}}
\UnaryInfC{\UN{I}{C}}
\BinaryInfC{$\xymatrix{I\ar@{}[rr]^{{\bf I}_{II}}&\bullet\ar[r]\ar[l]&I\ar@{}[rr]^{{\bf E}_{CI}}\ar[r]&\bullet&I\ar[l]}$}
\UnaryInfC{\PN{I}{C}}
\DisplayProof
\end{eqnarray*}
respectively.
\item[-] ${\bf I}_{IH}$ by virtue of the following proof tree in which a suitable instance of the rule of inference~\eqref{25} occurs:
\begin{eqnarray*}
\AxiomC{$\xymatrix{I\ar@{}[rr]^{{\bf I}_{II}}&\bullet\ar[l]\ar[r]&I}$}
\AxiomC{$\xymatrix{I\ar[r]^{\al{is}}&H}$}
\BinaryInfC{$\PA{I}{H}$}
\DisplayProof
\end{eqnarray*}
\end{itemize}
Now, assuming that custodians, inspectors and helpers are identified by means of suitable identification alphanumeric strings, a model of the ologism~\eqref{cih} for the shift change of tomorrow is identified by the assignments
\begin{eqnarray*}
\mathcal M_0:\left[\begin{array}{l}
C\mapsto\{C10, C11, C12, C13\}\\
I\mapsto\{H10I1, H11I2, H12I3, H13I4\}\\
H\mapsto\{H01, H02, H03, H10I1, H11I2, H12I3, H13I4\}
\end{array}\right]
\\
\\
\mathcal M_1:\left[\begin{array}{l}
\al{has}\mapsto\left[\begin{array}{l}
C10\mapsto H01\\
C11\mapsto H02\\
C12\mapsto H12I3\\
C13\mapsto H11I2
\end{array}\right]
\\
\\
\al{is}\mapsto\left[\begin{array}{l}
H10I1\mapsto H10I1\\
H11I2\mapsto H11I2\\
H12I3\mapsto H12I3\\
H13I4\mapsto H13I4
\end{array}\right]
\end{array}\right]
\end{eqnarray*}
which clearly satisfy the imposed constraints, as well as the deducible ones. 
\end{enumerate}
\end{exms}


\section{Soundness and completeness}\label{soundcompl}
Let $\mathfrak O$ be an ologism. In this section we prove that everything which is deducible from $\mathfrak O$ is true, that is satisfied in every model $\mathcal M$ of $\mathfrak O$. 
This is what is often referred to as a {\bf soundness theorem}. On the other hand we will also prove that everything which is true is deducible from $\mathfrak O$. This is what is often referred to as a
{\bf completeness theorem}.

\begin{defi}\label{consrel}
Let $\mathfrak O=(\mathcal O,\mathcal F,\alpha,\epsilon,\iota,o)$ be an ologism.
\begin{enumerate} 
\item For $(f_1,f_2,\ldots,f_n):A\rightarrow B$, 
$(g_1,g_2,\ldots,g_m):A\rightarrow B$ a pair of parallel morphisms in $\mathcal{(\mathcal O_{\bullet})}^*$, see point~\eqref{freecat} in examples~\eqref{catexamples}, an $\mathfrak O$-{\bf equality} between them is a formal writing $f_1;f_2;\cdots ;f_n=g_1;g_2;\cdots ;g_m$. For $\sigma$ an $\mathfrak O$-equality, we write
$\mathfrak O\vdash\sigma$ to mean that the pair $(f_1;f_2;\cdots;f_n,g_1;g_2;\cdots ;g_m)$ is in the smallest congruence relation generated by $\mathcal F$, see definition~\ref{tildestar}.
\item\label{b} 
An $\mathfrak O$-{\bf categorical proposition} is a diagram of the form 
$D_{\bf X}(A,B):{\bf X}\rightarrow Th(\mathfrak O)$, see definition~\ref{logthologism}, with ${\bf X}\in\{{\bf A}, {\bf E}, {\bf I}, {\bf O}\}$. For $\sigma$ an $\mathfrak O$-categorical proposition, we write $\mathfrak O\vdash\sigma$ to mean that $\sigma\in\alpha^*\cup\epsilon^*\cup\iota^*\cup o^*$.
\item An $\mathfrak O$-{\bf proposition} is an $\mathfrak O$-equality or an $\mathfrak O$-categorical proposition. For $\sigma$ an $\mathfrak O$-proposition, in either cases (a) or (b) above, the writing $\mathfrak O\vdash\sigma$ is read as $\sigma$ {\bf is a logical consequence of} $\mathfrak O$.
\end{enumerate}
\end{defi}

\begin{defi}
Let $\mathfrak O=(\mathcal O,\mathcal F,\alpha,\epsilon,\iota, o)$ be an ologism, $\sigma$ be an $\mathfrak O$-proposition and $\mathcal M$ be a model of $\mathfrak O$. We write $\mathcal M\models\sigma$ and say that $\mathcal M$ {\bf satisfies} $\sigma$ if 
\begin{enumerate}
\item in case $\sigma$ is an $\mathfrak O$-equality $f_1;f_2;\cdots;f_n=g_1;g_2;\cdots;g_m$,
the equality 
$$
\mathcal M f_1\circ\cdots\circ\mathcal Mf_n=\mathcal Mg_1\circ\cdots\circ\mathcal Mg_m
$$ 
holds in {\bf Sets};
\item in case $\sigma$ is an $\mathfrak O$-categorical proposition, see point~\eqref{b} of definition~\eqref{consrel},  one of the following alternatives holds accordingly:
\begin{itemize}
\item[-] $\mathcal MA\subseteq\mathcal MB$, if $\sigma$ is of the form $D_{\bf A}(A,B)$;
\item[-] $\mathcal MA\cap\mathcal MB=\emptyset$, if $\sigma$ is of the form $D_{\bf E}(A,B)$;
\item[-] $\mathcal MA\cap\mathcal MB\neq\emptyset$, if $\sigma$ is of the form $D_{\bf I}(A,B)$;
\item[-] $\mathcal MA\centernot\subseteq\mathcal MB$, if $\sigma$ is of the form $D_{\bf O}(A,B)$.
\end{itemize}
\end{enumerate}
\end{defi}

\begin{defi}
Let $\mathfrak O$ be an ologism and $\sigma$ be an $\mathfrak O$-proposition. We say that $\sigma$ is a {\bf semantical consequence} of $\mathfrak O$, and write it $\mathfrak O\models\sigma$, if 
every model of $\mathfrak O$ satisfies it.
\end{defi}

\begin{rem}
On the base of definition~\ref{modologism}, the syllogistic premisses in the deductive equipment of an ologism are among its semantical consequences.
\end{rem}

\begin{thm}[Soundness]
Let $\mathfrak O=(\mathcal O,\mathcal F,\alpha,\epsilon,\iota,o)$ be an ologism and
let $\sigma$ be an $\mathfrak O$-proposition. If $\sigma$ is a logical consequence of $\mathfrak O$, then it is a semantical consequence of $\mathfrak O$.
\end{thm}

\begin{proof}
Suppose that $\mathfrak O\vdash\sigma$ and let $\mathcal M$ be an arbitrary model of $\mathfrak O$. If $\sigma$ is an $\mathfrak O$-equality then $\mathcal M\models\sigma$ because $\mathcal M$ is a model of the olog $(\mathcal O_{\bullet},\mathcal F)$, in particular. If $\sigma$ is an $\mathfrak O$-categorical proposition, we proceed by induction on the structure of a diagrammatic proof of $\sigma$, by cases.
\begin{itemize}
\item[-] Suppose that $\sigma\in\alpha^*$ is of the form $\xymatrix{A\ar[r]^{\al{is}}&C}$. Thanks to the rejection criterion discussed in subsection~\ref{rejection}, the only way to obtain $\sigma$ is by means of an instance of the rule of inference~\eqref{cutcut}, say
$$
\AxiomC{$\xymatrix{A\ar[r]^{\al{is}}&B}$}
\AxiomC{$\xymatrix{B\ar[r]^{\al{is}}&C}$}
\BinaryInfC{$\xymatrix{A\ar[r]^{\al{is}}&C}$}
\DisplayProof
$$
with, $\mathcal MA\subseteq\mathcal MB$ and $\mathcal MB\subseteq\mathcal MC$, by inductive hypothesis. Thus, it is immediate to conclude $\mathcal MA\subseteq\mathcal MC$, namely that $\mathcal M\models\sigma$.
\item[-] Suppose that $\sigma\in\epsilon^*$ is of the form $\UN{A}{C}$. Thanks to the rejection criterion discussed in subsection~\ref{rejection}, the only way to obtain $\sigma$ is by means of an instance of the rules of inference~\eqref{ae} or~\eqref{ea}, say
$$
\AxiomC{$\UN{A}{B}$}
\AxiomC{$\xymatrix{B&C\ar[l]_{\al{is}}}$}
\BinaryInfC{$\UN{A}{C}$}
\DisplayProof
$$
in the first case, with $\mathcal MA\cap\mathcal MB=\emptyset$ and $\mathcal MC\subseteq\mathcal MB$, by inductive hypothesis. It is immediate to conclude that $\mathcal MA\cap\mathcal MC=\emptyset$, namely that $\mathcal M\models\sigma$. The thesis holds similarly in the case of an instance of the rule~\eqref{ea}.
\item[-] Suppose that $\sigma\in\iota^*$ is of the form $\PA{A}{C}$. Thanks to the rejection criterion discussed in subsection~\ref{rejection}, the only way to obtain $\sigma$ is by means of an instance of the rule of inference~\eqref{25} or~\eqref{26}, say
$$
\AxiomC{$\xymatrix{A\ar@{}[rr]^{{\bf I}_{AB}}&\bullet\ar[l]\ar[r]&B}$}
\AxiomC{$\xymatrix{B\ar[r]^{\al{is}}&C}$}
\BinaryInfC{$\PA{A}{C}$}
\DisplayProof
$$
in the first case, with $\mathcal MA\cap\mathcal MB\neq\emptyset$ and $\mathcal MB\subseteq\mathcal MC$. It is immediate to conclude that $\mathcal MA\cap\mathcal MC\neq\emptyset$, namely that $\mathcal M\models \sigma$. The thesis holds similarly in the case of an instance of the rule~\eqref{26}.
\item[-] Suppose that $\sigma\in o^*$ is of the form $\PN{A}{C}$. Thanks to the rejection criterion discussed in subsection~\ref{rejection}, the only way to obtain $\sigma$ is by means of an instance of the rule of inference~\eqref{oo} or~\eqref{210} or~\eqref{211}. Suppose it to be
$$
\AxiomC{$\PA{A}{B}$}
\AxiomC{$\UN{B}{C}$}
\BinaryInfC{$\PN{A}{C}$}
\DisplayProof
$$
in the first case, with $\mathcal MA\cap\mathcal MB\neq\emptyset$ and $\mathcal MB\cap\mathcal MC=\emptyset$ by inductive hypothesis. It is immediate to conclude that $\mathcal MA\centernot\subseteq\mathcal MC$, namely that $\mathcal M\models\sigma$. Now, suppose that $\sigma$ has been obtained by means of the following instance of the rule of inferenc~\eqref{210}:
$$
\AxiomC{$\xymatrix{A&B\ar[l]_{\al{is}}}$}
\AxiomC{$\PN{B}{C}$}
\BinaryInfC{$\PN{A}{C}$}
\DisplayProof
$$
with $\mathcal MB\subseteq\mathcal MA$ and $\mathcal MB\centernot\subseteq\mathcal MC$ by inductive hypothesis. It is immediate to conclude that $\mathcal MA\centernot\subseteq\mathcal MC$, namely that $\mathcal M\models\sigma$. The thesis holds similarly in the case of an instance of the rule~\eqref{211}.
\end{itemize}
Now, by the arbitrariness of $\mathcal M$, it follows that $\mathfrak O\models\sigma$.
\end{proof}

\begin{thm}[Completeness]
Let $\mathfrak O=(\mathcal O,\mathcal F,\alpha,\epsilon,\iota,o)$ be an ologism and
let $\sigma$ be an $\mathfrak O$-proposition. If $\sigma$ is a semantical consequence of $\mathfrak O$, then it is a logical consequence of $\mathfrak O$.
\end{thm}
\begin{proof}
Suppose that $\mathfrak O\models\sigma$.
If $\sigma$ is an $\mathfrak O$-equality  $f_1;f_2;\cdots; f_n=g_1;g_2;\cdots;g_m$, then  $\mathfrak O\models\sigma$ holds if and only if  
$[f_1;f_2;\cdots; f_n]=[g_1;g_2;\cdots;g_m]$ holds in the equational theory generated by the olog $(\mathcal O_{\bullet},\mathcal F)$, namely the category $Th(\mathfrak O_{\bullet})$, by virtue of what has been discussed in subsections~\ref{comeqproof} and~\ref{proofeqsketch}, if and only if the pair $(f_1;f_2;\cdots; f_n, g_1;g_2;\cdots;g_m)$ is in the smallest congruence relation generated by $\mathcal F$ which fact, by definition~\ref{consrel}, means that $\mathfrak O\vdash\sigma$. On the other hand, if $\sigma$ is an $\mathfrak O$-categorical proposition $D_{\bf X}(A,B):{\bf X}\rightarrow Th(\mathfrak O)$, then by virtue of the way the graph $Th(\mathfrak O)$ has been built, see definition~\ref{logthologism}, 
$\sigma$ occurs as the conclusion of a suitable diagrammatic proof. By cases and by induction on the structure of such a diagrammatic proof we prove that $\sigma$ is like that from categorical premisses that are $\mathfrak O$-categorical propositions, that is 
already in the logical theory which is the graph $Th(\mathfrak O)$. 
\begin{itemize}
\item[-] In case $\sigma$ is of the form $\xymatrix{A\ar[r]^{\al{is}}&C}$, then on the base of the rejection criterion discussed in subsection~\ref{rejection} the only way to obtain $\sigma$ as the conclusion of a diagrammatic syllogistic deduction 
is by means of an instance of the rule of inference~\eqref{cutcut}, as
$$
\AxiomC{$\xymatrix{A\ar[r]^{\al{is}}&B}$}
\AxiomC{$\xymatrix{B\ar[r]^{\al{is}}&C}$}
\BinaryInfC{$\xymatrix{A\ar[r]^{\al{is}}&C}$}
\DisplayProof
$$
with $\mathfrak O\models\al{is}_{AB}$ and $\mathfrak O\models\al{is}_{BC}$, since otherwise there exists a particular model $\mathcal M$ of $\mathfrak O$ with $\mathcal M(A)\centernot\subseteq\mathcal M(B)$ or $\mathcal M(B)\centernot\subseteq\mathcal M(C)$ thus making it possible to find that 
$\mathcal M(A)\centernot\subseteq\mathcal M(C)$, against the hypothesis $\mathcal M(A)\subseteq\mathcal M(C)$. For instance, take $\mathcal M(A)=\{0,1\}$, $\mathcal M(B)=\{0\}$ and $\mathcal M(C)=\{0,2\}$ in the first case, or $\mathcal M(A)=\mathcal M(B)=\{0,1\}$, $\mathcal M(C)=\{0,2\}$ in the second case. Thus, by inductive hypothesis
$\al{is}_{BC}\in\alpha^*$ and $\al{is}_{AB}\in\alpha^*$, thus $\al{is}_{AC}\in\alpha^*$, that is $\mathfrak O\vdash\sigma$.
\item[-] In case $\sigma$ is of the form $\UN{A}{C}$, then on the base of the rejection criterion discussed in subsection~\ref{rejection} the only way to obtain $\sigma$ as the conclusion of a diagrammatic syllogistic deduction is by means of an instance of the rules of inference~\eqref{ea} or~\eqref{ae}. In the first case, as
$$
\AxiomC{$\xymatrix{A\ar[r]^{\al{is}}&B}$}
\AxiomC{\UN{B}{C}}
\BinaryInfC{\UN{A}{C}}
\DisplayProof
$$
with $\mathfrak O\models{\bf E}_{BC}$ and $\mathfrak O\models\al{is}_{AB}$, since otherwise there exists a particular model $\mathcal M$ of $\mathfrak O$ with $\mathcal M(A)\centernot\subseteq\mathcal M(B)$ or $\mathcal M(B)\cap\mathcal M(C)\neq\emptyset$ thus making it possible to find that 
$\mathcal M(A)\cap\mathcal M(C)\neq\emptyset$, against the hypothesis $\mathcal M(A)\cap\mathcal M(C)=\emptyset$. For instance, take $\mathcal M(A)=\{0,1\}$, $\mathcal M(B)=\{0\}$ and $\mathcal M(C)=\{1\}$ in the first case, or $\mathcal M(A)=\mathcal M(C)=\{0\}$ and $\mathcal M(B)=\{0,1\}$ in the second case. Thus, by inductive hypothesis ${\bf E}_{BC}\in\epsilon^*$ and $\al{is}_{AB}\in\alpha^*$, thus ${\bf E}_{AC}\in\epsilon^*$, that is $\mathfrak O\vdash\sigma$.
Similarly, if $\sigma$ is deduced by an instance of the rule of inference~\eqref{ae}. 

\item[-] In case $\sigma$ is of the form $\PA{A}{C}$, then on the base of the rejection criterion discussed in subsection~\ref{rejection} the only way to obtain $\sigma$ as the conclusion of a diagrammatic syllogistic deduction is by means of an instance of the rule of inference~\eqref{25} or~\eqref{26}. In the first case, as
$$
\AxiomC{$\xymatrix{A\ar@{}[rr]^{{\bf I}_{AB}}&\bullet\ar[l]\ar[r]&B}$}
\AxiomC{$\xymatrix{B\ar[r]^{\al{is}}&C}$}
\BinaryInfC{$\PA{A}{C}$}
\DisplayProof
$$
with $\mathfrak O\models{\bf I}_{AB}$ and $\mathfrak O\models\al{is}_{BC}$ since otherwise there exists a particular model $\mathcal M$ of $\mathfrak O$ with $\mathcal M(B)\centernot\subseteq\mathcal M(C)$ or $\mathcal M(A)\cap\mathcal M(B)=\emptyset$ thus making it possible to find that 
$\mathcal M(A)\cap\mathcal M(C)=\emptyset$, against the hypothesis $\mathcal M(A)\cap\mathcal M(C)\neq\emptyset$. For instance, take $\mathcal M(A)=\{0\}$, $\mathcal M(B)=\{0,1\}$ and $\mathcal M(C)=\{1\}$ in the first case, or $\mathcal M(A)=\{2\}$, $\mathcal M(B)=\{0\}$ and $\mathcal M(C)=\{0,1\}$ in the second case. Thus, by inductive hypothesis ${\bf I}_{AB}\in\iota^*$ and $\al{is}_{BC}\in\alpha^*$, thus ${\bf I}_{AC}\in\iota^*$, that is $\mathfrak O\vdash\sigma$. Similarly, if $\sigma$ is deduced by means of an instance of the rule of inference~\eqref{26}.
\item[-] In case $\sigma$ is of the form \PN{A}{C}, then on the base of the rejection criterion discussed in subsection~\ref{rejection} the only way to obtain $\sigma$ as the conclusion of a diagrammatic syllogistic deduction is by means of an instance of the rule of inference~\eqref{oo}, or~\eqref{210}, or~\eqref{211}. In the first case, as
$$
\AxiomC{$\PA{A}{B}$}
\AxiomC{$\UN{B}{C}$}
\BinaryInfC{$\PN{A}{C}$}
\DisplayProof
$$
with $\mathfrak O\models{\bf I}_{AB}$ and $\mathfrak O\models{\bf E}_{BC}$, since otherwise there exists a particular model $\mathcal M$ of $\mathfrak O$ with $\mathcal M(B)\cap\mathcal M(C)\neq\emptyset$ or $\mathcal M(A)\cap\mathcal M(B)=\emptyset$ thus making it possible to find that 
$\mathcal M(A)\subseteq\mathcal M(C)$, against the hypothesis $\mathcal M(A)\centernot\subseteq\mathcal M(C)$. For instance, take $\mathcal M(A)=\mathcal M(B)=\mathcal M(C)=\{0\}$ in the first case,
or $\mathcal M(A)=\mathcal M(C)=\{0\}$ and $\mathcal M(B)=\{1\}$ in the second case.
Thus, by inductive hypothesis ${\bf I}_{AB}\in\iota^*$ and ${\bf E}_{BC}\in\epsilon^*$, thus ${\bf O}_{AC}\in o^*$, that is $\mathfrak O\vdash\sigma$. Similarly, if $\sigma$ is deduced by means of the rule of inference~\eqref{210}, or~\eqref{211}.
\qedhere
\end{itemize}
\end{proof}


\section{Contradiction in ologisms}\label{contraologisms}
In this section we briefly address a particular topic: the managing of contradiction in ologisms. In subsection~\ref{contradiction} we showed in which sense the diagonally opposite diagrams in the diagrammatic square of opposition~\eqref{diagrams} are pairwise contradictory. In this section we discuss how the occurrence of pairs of contradictory syllogistic constraints in ologisms give rise to contradiction in them, so decreeing their structural unsoundness. We present the matter by going through the cases that have to be considered. Hypothetically, if one thinks of the authoring of ologisms as to a computer-assisted activity, the cases that we are going to discuss should be thought of as those in which the process of automatic checking of the structural soundness of an ologism under construction should point out a suitable message of error. In practice, contradictory ologisms may arise because more than one author is working on the planning of an ologism, and there could not be agreement about the assignment of a pertaining type to an entity under consideration or about which aspects should be more appropriately applied. For instance, contradictory ologisms such as those shown in examples~\eqref{contra1} and~\eqref{contra2} below may arise because there could not be agreement on what should be considered as a vertebrate that is able to fly.

\begin{enumerate}
\item For $A, B$ any types in an ologism, the syllogistic constraints 
$$
\PA{A}{B}\qquad\qquad\UN{A}{B}
$$
may occur together in a piece of ologism like
$$
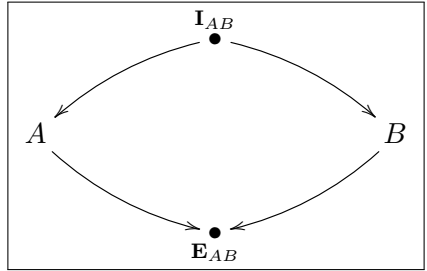
\fbox{\xymatrix{&&\bullet\ar@{}[dd]|(-.1){{\bf I}_{AB}}|(1.1){{\bf E}_{AB}}\ar@/^/[drr]\ar@/_/[dll]\\
A\ar@/_/[drr]&&&&B\ar@/^/[dll]\\
&&\bullet
}}
$$
in which two deduced arrows may arise:
\begin{itemize}
\item[-] the one through $A$,
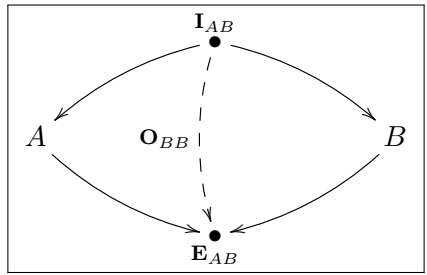
$$
\fbox{\xymatrix{&&\bullet\ar@{}[dd]|(-.1){{\bf I}_{AB}}|(1.1){{\bf E}_{AB}}
\ar@/_/@{-->}[dd]_{{\bf O}_{BB}}\ar@/^/[drr]\ar@/_/[dll]\\
A\ar@/_/[drr]&&&&B\ar@/^/[dll]\\
&&\bullet
}}
$$
forcing the contradictory syllogistic constraint 
$$
\PN{B}{B}
$$
\item[-] or the one through $B$,
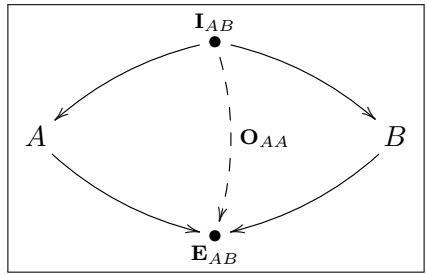
$$
\fbox{\xymatrix{&&\bullet\ar@{}[dd]|(-.1){{\bf I}_{AB}}|(1.1){{\bf E}_{AB}}
\ar@/^/@{-->}[dd]^{{\bf O}_{AA}}\ar@/^/[drr]\ar@/_/[dll]\\
A\ar@/_/[drr]&&&&B\ar@/^/[dll]\\
&&\bullet
}}
$$
forcing the contradictory syllogistic constraint 
$$
\PN{A}{A}
$$
\end{itemize}

\begin{exm}
An ologism that is contradictory in the way we just described is
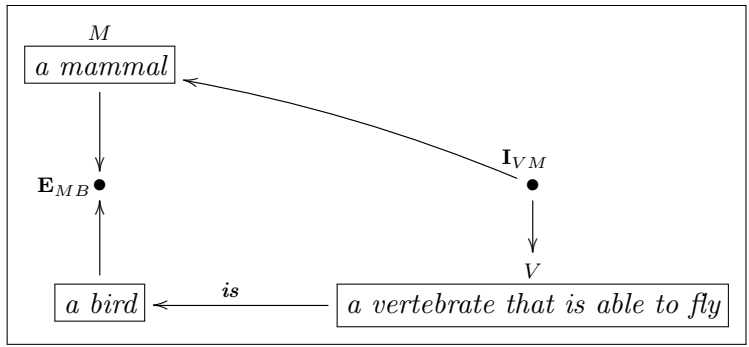
\begin{eqnarray}\label{contra1}
\fbox{\xymatrix{\stackrel{M}{\fbox{a mammal}}\ar@{}[dd]_{{\bf E}_{MB}}\ar[d]\\
\bullet&&\bullet\ar@/_/[ull]_(.05){{\bf I}_{VM}}\ar[d]\\
\fbox{a bird}\ar[u]&&\stackrel{V}{\fbox{a vertebrate that is able to fly}}\ar[ll]_(.7){\al{is}}
}}
\end{eqnarray}
although it may seem to be sound at first glance. In fact, the dashed arrow in
$$
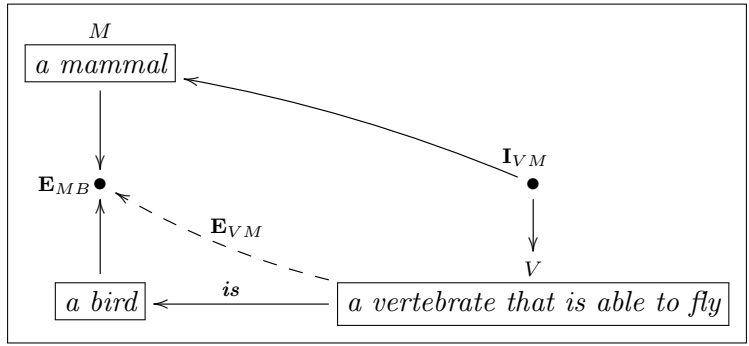
\fbox{\xymatrix{\stackrel{M}{\fbox{a mammal}}\ar@{}[dd]_{{\bf E}_{MB}}\ar[d]\\
\bullet&&\bullet\ar@/_/[ull]_(.05){{\bf I}_{VM}}\ar[d]\\
\fbox{a bird}\ar[u]&&\stackrel{V}{\fbox{a vertebrate that is able to fly}}\ar@{-->}@/^1pc/[ull]_(.7){{\bf E}_{VM}}\ar[ll]_(.7){\al{is}}
}}
$$
allows to read off the deduced syllogistic constraint ${\bf E}_{VM}$ for ``Every vertebrate that is able to fly is not a mammal'', which is in contradiction with the given syllogistic constraint  ${\bf I}_{VM}$ for ``Some vertebrate is a mammal''.
\end{exm}

\item For $A, B$ any types in an ologism, 
an aspect $\xymatrix{A\ar[r]^{\al{is}}&B}$, as a syllogistic constraint of the form ${\bf A}_{AB}$ 
and its contradictory ${\bf O}_{AB}$ may occur together in a piece of ologism like
$$
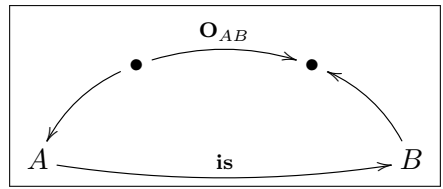
\fbox{\xymatrix{&\bullet\ar@/_/[dl]\ar@/^/[rr]^{{\bf O}_{AB}}&&\bullet\\
A\ar@/_/[rrrr]^{\al{is}}&&&&B\ar@/_/[ul]
}}
$$
in which, two deduced arrows arise:
\begin{itemize}
\item[-] the one through $B$,
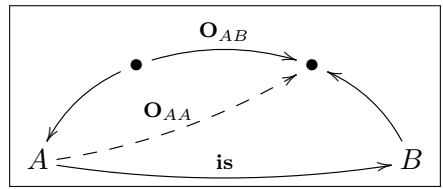
$$
\fbox{\xymatrix{&\bullet\ar@/_/[dl]\ar@/^/[rr]^{{\bf O}_{AB}}&&\bullet\\
A\ar@/_/@{-->}[urrr]^{{\bf O}_{AA}}\ar@/_/[rrrr]^{\al{is}}&&&&B\ar@/_/[ul]
}}
$$
forcing the contradictory syllogistic constraint
$$
\PN{A}{A}
$$
\item[-] and the one through $A$,
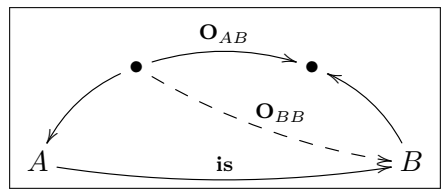
$$
\fbox{\xymatrix{&\bullet\ar@/_/@{-->}[drrr]^{{\bf O}_{BB}}\ar@/_/[dl]
\ar@/^/[rr]^{{\bf O}_{AB}}&&\bullet\\
A\ar@/_/[rrrr]^{\al{is}}&&&&B\ar@/_/[ul]
}}
$$
forcing the contradictory syllogistic constraint 
$$
\PN{B}{B}
$$
\end{itemize}

\begin{exm}
As an example of an ologism which is contradictory in the way we just described one can consider
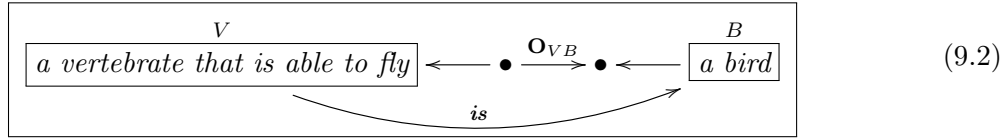
\begin{eqnarray}\label{contra2}
\fbox{\xymatrix{\stackrel{V}{\fbox{a vertebrate that is able to fly}}\ar@/_2pc/[rrr]^{\al{is}}
&\bullet\ar[l]\ar[r]^{{\bf O}_{VB}}&\bullet&\stackrel{B}{\fbox{a bird}}\ar[l]
}}
\end{eqnarray}
from which one can read off that the syllogistic constraint ${\bf A}_{VB}$ holds in the form ``Every vertebrate that is able to fly is a bird''
together with its negation ``Some vertebrate that is able to fly is not a bird'',
via ${\bf O}_{VB}$.
\end{exm}
\end{enumerate}


\section{Ologisms and Description Logics, a comparison}\label{forlogrepr}
In this section we give attention to the relationship between ologisms as systems for knowledge representation and reasoning, and Description Logics (DL) as languages that can be considered as the core of the knowledge representation systems formed by a DL knowledge base and an associated reasoning equipment as described in~\cite{MR1991592}, to which we refer the reader who is interested in deepening her knowledge about DL. Henceforth, we will write DL-KRS to refer to knowledge representation systems based on Description Logics. In what follows we will compare ologisms and DL-KRS as regards some of the basic features shared by the two formalisms. This is necessarily so, since the comparison at issue is between a well-established formalism for knowledge representation on one side and a formalism for knowledge representation that at the moment presents itself as an interesting, and hopefully promising, line of work. More to the point, ologisms and DL-KRS will be compared as intelligent systems for knowledge representation, in the sense specified below, then with respect to their being general-purpose appoaches to knowledge representation and with respect to the possibility of representing individuals, defining new concepts from previously defined ones, and representing binary relations other than the functional ones, finally with respect to the satisfaction of the so-called {\em open-world assumption}. 
We proceed by following~\cite{MR1991592}, mainly\\ 
In {\em loc. cit.} it is pointed out that the research in the field of knowledge representation and reasoning focuses on methods for high-level descriptions of the world that can be effectively used to build intelligent applications, with reference
to a system that allows to find implicit consequences of its explicitly represented knowledge. Clearly, ologisms are very far from being effectively usable as implementable intelligent knowledge-based applications; wherever in the paper we suggested to imagine the design of an ologism to happen on a computer and as an activity assisted by a computer, we have done that mostly in the form of an auspice in favour of the outcome of future investigations. Nonetheless, ologisms, although prototypical and not yet implementable, aim to become intelligent systems in the previously specified sense. Given the formal specification of an ologism as a knowledge base provided by a linear sketch for an olog and a set of syllogistic premisses, see definition~\ref{ologism} and terminology~\ref{terminpremiss}, it is possible to deduce further equational constraints as new facts from those provided in the underlying olog of the specified ologism, because this is typical of ologs, and to deduce new syllogistic constraints from the given syllogistic premisses, by means of the deductive equipment of the ologism, see definition~\ref{dedeq}, because this is what characterizes genuine ologisms extending the logical expressivity of ologs. Thus, similarly to DL-KRS, ologisms are knowledge representation systems formed by a knowledge base and a reasoning equipment.\\ 
Very briefly, DL-KRS arise from the need to merge two ways to approach knowledge representation. On one side there were logic-based approaches considering the predicate calculus as a tool that can be used to unambigously describe the world; because of this, those approaches were considered as general-purpose from the very beginning; in passing, let us mention that most DL are decidable fragments of first order logic. On the other side there were non-logic-based approaches (mainly {\em semantic networks} and {\em frames}) that although built on notions often developed for specific representational goals, turned out to employ formalisms that have been expected to be general purpose. More explicitly, non-logical systems created out of special lines of thinking expected to be applicable in different domains and on different kinds of problems; it turned out that this was not the case. DL-KRS arise from the merging of these two lines of approach to knowledge representation as the fruitful attempt to put theory and applications side by side. Ologisms retain the main feature of a logic-based approach to knowledge representation, since they are general-purpose from the start. This is so because their design takes place in accordance with rigorous mathematical rules, so that they turn out to possess an objectively recognizable structure which in principle, although typically stemming out in connection with the representation of a specific real situation, turns out to be transportable in different contexts of employment. It is actually like that for ologs since already exist knowledge-preserving transformations between them, see~\cite{DBLP:journals/corr/abs-1102-1889}, while, for genunine ologisms, it is like that only from a conjectural point of view, admittedly, because the investigation on which knowledge-preserving transformations between ologisms should be appropriately considered is an open issue.\\
In DL one works with three kinds of components. There are {\em individuals} or {\em constants}, {\em concepts} or {\em unary relations}, {\em roles} or {\em binary relations} and a typical feature of DL-KRS is their ability to allow the representation of relationships between concepts other than the IS-A relationships. Individuals can be represented in ologs as instances of types that have been passed at the olog level as types in themselves. Consequently, they turn out to be representable in ologisms too. Moreover, in ologs, thus in ologisms as well, individuals can enter in the representation of  {\em specific facts}. We did not previously emphasize this possibility since we chose to focus our attention on the description of the most basic constructions allowed in an olog. Without the intention of extensively pursue the matter here specific facts can be recorded in ologs by passing an instance of a type to a type in itself, as previously mentioned. For example, if ``Moby Dick'' is a documentable instance of the type 
$\ulcorner\textrm{a whale}\urcorner$ then to record that ``Moby Dick's weight is 30 tons'' one can consider the fact
$$
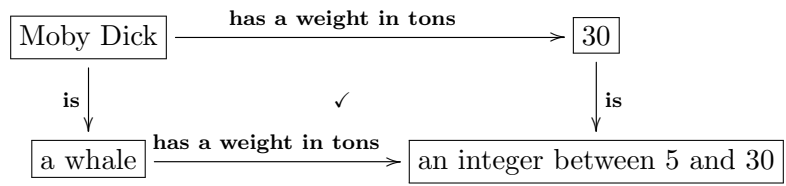
\xymatrix{\fbox{Moby Dick}\ar@{}[drrr]|{\checkmark}\ar[d]_{\al{is}}
\ar[rrr]^{\al{has a weight in tons}}&&&\fbox{30}\ar[d]^{\al{is}}\\
\fbox{a whale}\ar[rrr]^(.35){\al{has a weight in tons}}&&&\fbox{an integer between 5 and 30}}
$$
in which the instances ``Moby Dick'' and ``30'' have been passed at the olog level as types in themselves.\\
Concepts in DL-KRS are described as classes of individuals, so that they clearly correspond to types in ologisms.\\
Other than the functional binary relationships that identify the aspects of an ologism, which already testify the possibility of representing binary relations more general than the ``IS-A'' ones, non-functional relationships can be represented in ologisms as well, by suitably codifying them by means of aspects. For example, imagine that it arises the need to represent the binary relation ``is a parent of'' from the type $\ulcorner\textrm{a person}\urcorner$ to the type $\ulcorner\textrm{a child}\urcorner$. The relation at issue is non-functional since not every person is a parent of a child; nonetheless, it can be represented in an olog by means of the pair of aspects
$$
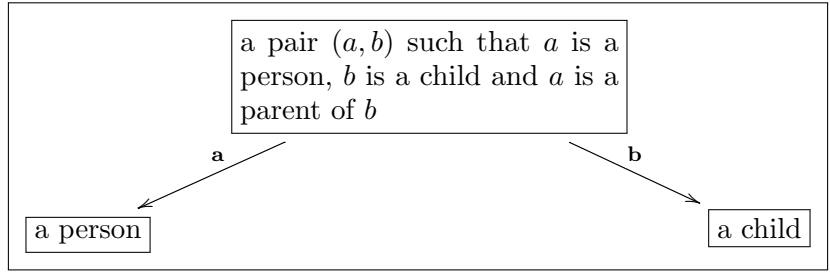
\fbox{\xymatrix{&\fbox{\parbox{.33\textwidth}{a pair $(a,b)$ such that $a$ is a person,
 $b$ is a child and $a$ is a parent of $b$}}\ar[dr]^(.6){\al b}\ar[dl]_(.6){\al a}\\
\fbox{a person}&&\fbox{a child}}}
$$
Now, for what concern the employment of DL as languages for defining a knowledge base and carry out inferences over it, we point out that a DL knowledge base is usually formed by a TBox and ABox. The TBox contains intensional knowledge in the form of declarations that describe general properties of concepts. Typically, a TBox is given in the form of a
{\em taxonomy}, which is a representaion of the generality-specificity relation that holds between concepts. For instance, a taxonomy could be represented as a network such as
$$
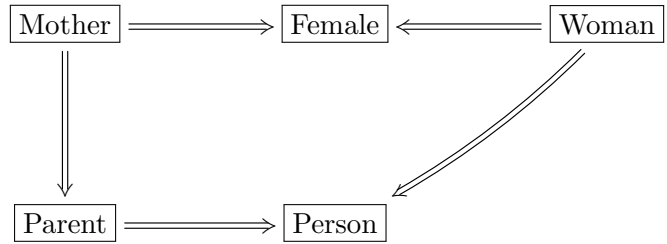
\xymatrix{\fbox{Mother}\ar@{=>}[dd]\ar@{=>}[rr]&&\fbox{Female}&&\fbox{Woman}\ar@{=>}[ll]\ar@{=>}@/^/[ddll]\\
\\
\fbox{Parent}\ar@{=>}[rr]&&\fbox{Person}}
$$
in which every link provides an ``IS-A'' relationship. The ABox contains extensional knowledge that is specific to the individuals of the domain of discourse.\\ 
The basic form of declaration in a TBox is the definition of a new concept in terms of other, previously defined, concepts. For instance, a woman can be defined as a female person, by writing  the declaration $\textrm{Woman}\equiv\textrm{Person}\sqcap\textrm{Female}$, where $\sqcap$ is the constructor for {\em intersection of concepts}. With respect to the possibility of defining new concepts from previously defined ones in ologisms, we can say that in this paper we introduced ologisms as founded on those that in~\cite{DBLP:journals/corr/abs-1102-1889} are more precisely called {\em basic ologs}, which we descibed as linear sketches. On one hand, as they stand, basic ologs do not allow any kind of category theoretic constructions that could be used to introduce new concepts as new types from previously defined ones; the category theoretic constructions at issue could be binary products or coproducts, pullbacks or pushouts, for instance. On the other hand, there is also to say that a whole hierarchy of more ``advanced'' ologs can be described in terms of more expressive sketches, see {\em loc. cit.}, allowing the implementation of the kind of category theoretic constructions previously mentioned, thus allowing the definition of new concepts from previously defined ones. In this paper we mainly focused on the extension of the logical expressivity of basic ologs, to begin with. Anyway, to give an idea of how to define a woman as a female person in an olog, hence in an ologism, suppose to have the possibility to perform image factorization of aspects, see~\cite{DBLP:journals/corr/abs-1102-1889}. Given this, the type $\ulcorner\textrm{a woman}\urcorner$ arises by virtue of the image factorization performed on the aspect $\xymatrix{\ulcorner\textrm{a female person}\urcorner\ar[r]^(.55){\al{is}}&\ulcorner\textrm{a person}\urcorner}$ as witnessed by the fact
$$
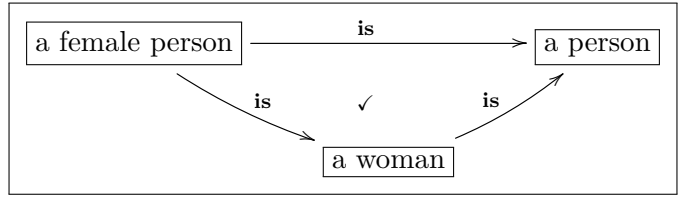
\fbox{\xymatrix{\fbox{a female person}\ar@{}[drr]|{\checkmark}\ar@/_/[dr]^{\al{is}}\ar[rr]^{\al{is}}&&\fbox{a person}\\
&\fbox{a woman}\ar@/_/[ur]^{\al{is}}&}}
$$
where the aspect $\xymatrix{\ulcorner\textrm{a female person}\urcorner\ar[r]^(.55){\al{is}}&\ulcorner\textrm{a woman}\urcorner}$ turns out to identify a bijective functional relationship, by construction. Moreover, imagine that in the ABox of a DL knowledge base that is going to be specified there should be represented the information ``Anna is a female person'', about the individual ANNA. This is usually expressed by means of a {\em membership assertion} such as
$\textrm{Female}\sqcap\textrm{Person}(\textrm{ANNA})$ from which, by virtue of the way the concept ``Woman'' has been previously defined, one consequently derives that ANNA is an individual of that concept, namely that Anna is a woman. The same deduction is represented in the olog
$$
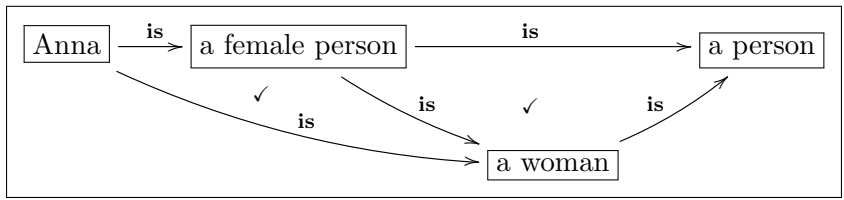
\fbox{\xymatrix{\fbox{Anna}\ar@/_1pc/[drr]^{\al{is}}\ar@{}[drr]|(.4){\checkmark}\ar[r]^(.38){\al{is}}&\fbox{a female person}\ar@{}[drr]|{\checkmark}\ar@/_/[dr]^{\al{is}}\ar[rr]^{\al{is}}&&\fbox{a person}\\
&&\fbox{a woman}\ar@/_/[ur]^{\al{is}}&}}
$$
where the aspect $\xymatrix{\ulcorner\textrm{Anna}\urcorner\ar[r]^(.4){\al{is}}&\ulcorner\textrm{a woman}\urcorner}$ is there because of the imposition of an identification between the aspect $\al{is}$ and the aspect $\al{is a female person who is}$ which is the aspect given by composition of the aspect $\xymatrix{\ulcorner\textrm{Anna}\urcorner\ar[r]^(.3){\al{is}}&\ulcorner\textrm{a female person}\urcorner}$ followed by the aspect $\xymatrix{\ulcorner\textrm{a female person}\urcorner\ar[r]^(.6){\al{is}}&\ulcorner\textrm{a person}\urcorner}$, that is by imposing the fact $\al{is}=\al{is a female person who is}$ indicated by the checkmark. Now, the previous olog can be promoted to a genuine ologism by imposing the syllogistic constraint ${\bf O}_{PW}$ for ``Some person is not a woman'', as in the drawing
$$
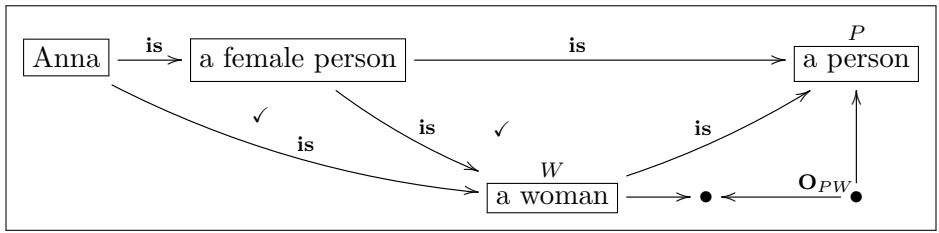
\fbox{\xymatrix{\fbox{Anna}\ar@/_1pc/[drr]^{\al{is}}\ar@{}[drr]|(.4){\checkmark}\ar[r]^(.38){\al{is}}&\fbox{a female person}\ar@{}[drr]|{\checkmark}\ar@/_/[dr]^{\al{is}}\ar[rrr]^{\al{is}}&&&\stackrel{P}{\fbox{a person}}\\
&&\stackrel{W}{\fbox{a woman}}\ar@/_/[urr]^{\al{is}}\ar[r]&\bullet&\bullet\ar@{}[ull]|(.1){{\bf O}_{PW}}\ar[l]\ar[u]}}
$$
from which it is immediate to deduce that there is at least one person other than Anna and that that person is a male person.\\
In the introduction, we pointed out that the reader should keep in mind that the matter treated in this paper could be thought of as connected to the conceptual planning of databases only potentially but, for sake of argument, let us 
consider that connection as actually taking place and spend a few words about the conceptual planning of databases by means of DL-KRS or ologisms with respect to the satisfaction of the so-called {\em open-world assumption}, in particular. In connection with the mentioned activity of planning DL-KRS satisfy the previously referred to assumption related to the semantics of ABoxes. This basically means that it cannot be assumed that the knowledge in the knowledge base is complete, contrary to what happens in connection with the ``classical'' planning of databases where the so-called {\em closed-world assumption} is satisfied. The main difference is the following: whereas absence of information is interpreted as negative information in the classical case, absence of information in an ABox only indicates lack of knowledge. Things go in parallel for ologisms. They are objects that have to be thought of as being continuosly under a process of improvement until an agreed upon version of them is reached. Because of this, the knowledge stored in an ologism cannot be considered complete but ``as complete as possible'' relatively to a certain level of understanding of a subjective view of a real situation; thus, in this sense ologisms satisfy the open-world assumption.\\
In conclusion, from the previous discussion it emerges that in comparison with DL-KRS, ologisms are knowledge representation systems that need to be much more discussed and developed, nonetheless we think that because of this they represent a challenge and at the same time an opportunity.


\bibliographystyle{plain} 
\bibliography{BiBlio}
\end{document}